\documentclass[sigconf, nonacm]{acmart}

\usepackage{array}
\usepackage[dvipsnames]{xcolor}
\usepackage{hyperref}
\usepackage{cleveref}
\usepackage{mathtools}
\usepackage{amsmath}
\usepackage{thmtools}
\usepackage{amsthm}
\usepackage[linesnumbered,ruled,vlined]{algorithm2e}
\usepackage{enumitem}
\usepackage{stackengine}
\usepackage[table]{xcolor}
\usepackage[most]{tcolorbox}
\usepackage{scalerel}
\usepackage{multirow}
\usepackage[skip=0.5pt]{subcaption}
\usepackage{multicol}
\usepackage{tabularx}
\usepackage{balance}
\usepackage{etoolbox}
\usepackage{listings}

\newcommand{\paratitle}[1]{\vspace{1mm}\noindent\textbf{{#1}.}}

\newcommand{\cut}[1]{}

\SetCommentSty{mycommfont}

\newif\ifpaper
 \paperfalse   

\newcounter{BalanceAtReference}
\ifpaper
\else
\fi
\newcounter{ReferenceIndexForBalancing}

\makeatletter

\global\@ACM@balancefalse

\def\@balancelastpageonce{%
  \ifnum\value{ReferenceIndexForBalancing}=\value{BalanceAtReference}
    \newpage
  \else
    \relax
  \fi
  \stepcounter{ReferenceIndexForBalancing}
}
\pretocmd{\bibitem}{\@balancelastpageonce}
  {}
  {\@latex@error{Patching \bibitem failed}{\@ehd}}

\makeatother

\newtheorem{theorem}{Theorem}[section]
\newtheorem{definition}[theorem]{Definition}

\crefname{algocf}{algorithm}{algorithms}
\Crefname{algocf}{Algorithm}{Algorithms}

\newcommand{\ag}[1]{{\texttt{\color{blue} AG: [{#1}]}}}

\newcommand{\todo}[1]{\textbf{\textcolor{red}{(TO DO)}}}

\newcommand{\attrset}{\ensuremath{\mathcal{A}}}
\newcommand{\dom}{\ensuremath{Dom}}
\newcommand{\prob}{\ensuremath{Pr_D}}
\newcommand{\probneighbor}{\ensuremath{Pr_{D'}}}
\newcommand{\margdiff}{\ensuremath{\mathcal{MD}}}
\newcommand{\tids}{\ensuremath{tids}}
\newcommand{\range}{\ensuremath{Range}}
\newcommand{\fairnessdef}{\ensuremath{F}}
\newcommand{\constraints}{\ensuremath{\mathcal{F}}}
\newcommand{\measure}{\ensuremath{\mathcal{U}}}

\newcommand{\mutual}{\ensuremath{\mathcal{U_\mathsf{MI}}}}

\newcommand{\mutualproxybayes}{\ensuremath{\mathcal{U}^{Bayes}_\mathsf{MI}}}

\newcommand{\mutualproxytvd}{\ensuremath{\mathcal{U}^{TVD}_\mathsf{MI}}}

\newcommand{\repair}{\ensuremath{\mathcal{U_\mathsf{R}}}}
\newcommand{\repairsat}{\ensuremath{\mathcal{U^{SAT}_\mathsf{R}}}}
\newcommand{\repairsatchunked}{\ensuremath{\mathcal{U^{SAT_\text{chunk}}_\mathsf{R}}}}
\newcommand{\contribution}{\ensuremath{\mathcal{U_\mathsf{TC}}}}
\newcommand{\jdb}{\ensuremath{D_{sj}}}
\newcommand{\crossdb}{\ensuremath{D_{cross}}}

\newcommand\symdif{\ThisStyle{\mathrel{\ensurestackMath{%
  \stackengine{.2\LMex}{\SavedStyle-}{\SavedStyle\dot{}}{U}{c}{F}{F}{L}}}}}

\newcommand{\sens}{\ensuremath{\Delta}}
\newcommand{\repaircost}{\ensuremath{dist}}
\newcommand{\neighbor}{\sim}

\newcommand{\protecteda}{\ensuremath{P}}
\newcommand{\response}{\ensuremath{O}}
\newcommand{\outcome}{\ensuremath{O}}
\newcommand{\admissible}{\ensuremath{A}}
\newcommand{\protectedset}{\ensuremath{\protecteda}}
\newcommand{\responseset}{\ensuremath{\response}}
\newcommand{\outcomeset}{\ensuremath{\outcome}}
\newcommand{\admissibleset}{\ensuremath{\admissible}}
\newcommand{\protectedvalue}{\ensuremath{p}}
\newcommand{\responsevalue}{\ensuremath{y}}
\newcommand{\admissiblevalue}{\ensuremath{a}}

\newcommand\vldbdoi{XX.XX/XXX.XX}
\newcommand\vldbpages{XXX-XXX}
\newcommand\vldbvolume{19}
\newcommand\vldbissue{9}
\newcommand\vldbyear{2026}
\newcommand\vldbauthors{\authors}
\newcommand\vldbtitle{\shorttitle} 
\newcommand\vldbavailabilityurl{https://github.com/makarenko1/Measures-for-estimating-unfairness-in-data-under-DP/blob/main/}
\newcommand\vldbpagestyle{empty} 

\newcommand\independent{\protect\mathpalette{\protect\independenT}{\perp}}
\def\independenT#1#2{\mathrel{\rlap{$#1#2$}\mkern2mu{#1#2}}}

\begin{document}

\title{Measuring Database Unfairness via Dependency Quantification Under Differential Privacy}

\author{Mariia Vologdin}
\affiliation{%
  \institution{Hebrew University}
  }
\email{mariia.vologdin@mail.huji.ac.il}

\author{Yuchao Tao}
\affiliation{%
  \institution{Independent Researcher}
}
\email{harry.t.chao@gmail.com}

\author{Amir Gilad}
\affiliation{%
  \institution{Hebrew University}
  }
\email{amirg@cs.huji.ac.il}

\begin{abstract}
Differential privacy (DP) has become the de facto standard for protecting sensitive data, providing strong guarantees that published statistics or models reveal limited information about any individual. However, privacy noise and restricted data access make it increasingly difficult to assess the fairness and reliability of private datasets. 
In this paper, we propose a formal framework for quantifying data unfairness under DP. 
We identify three core desiderata for unfairness measures based on previous work: positivity, monotonicity, and DP computability. We further instantiate them through three complementary measures: (1) a mutual information–based measure with a total variation distance proxy suitable for DP, (2) a data-repair–based measure approximated via a reduction to weighted MaxSAT, and (3) a top-$k$ tuple contribution measure that isolates the most influential records in fairness violations. We design privacy-preserving algorithms and analyze their sensitivity, accuracy, and efficiency. Extensive experiments on multiple real-world datasets demonstrate that our proposed measures faithfully approximate their non-private counterparts, effectively quantify bias under privacy constraints, and provide insights for data management.
\end{abstract}

\maketitle


\pagestyle{\vldbpagestyle}
\begingroup\small\noindent\raggedright\textbf{PVLDB Reference Format:}\\
\vldbauthors. \vldbtitle. PVLDB, \vldbvolume(\vldbissue): \vldbpages, \vldbyear.\\
\href{https://doi.org/\vldbdoi}{doi:\vldbdoi}
\endgroup
\begingroup
\renewcommand\thefootnote{}\footnote{\noindent
This work is licensed under the Creative Commons BY-NC-ND 4.0 International License. Visit \url{https://creativecommons.org/licenses/by-nc-nd/4.0/} to view a copy of this license. For any use beyond those covered by this license, obtain permission by emailing \href{mailto:info@vldb.org}{info@vldb.org}. Copyright is held by the owner/author(s). Publication rights licensed to the VLDB Endowment. \\
\raggedright Proceedings of the VLDB Endowment, Vol. \vldbvolume, No. \vldbissue\ %
ISSN 2150-8097. \\
\href{https://doi.org/\vldbdoi}{doi:\vldbdoi} \\
}\addtocounter{footnote}{-1}\endgroup


\section{Introduction}\label{sec:intro}

Differential privacy (DP)~\cite{dwork2014:textbook} has become the leading standard for safeguarding sensitive information in data analysis and machine learning (ML). It gives formal guarantees that published statistics or trained models reveal only limited information about any individual, even in the presence of auxiliary knowledge. DP has been widely adopted in practice by major organizations~\cite{erlingsson2014rappor,ding2017collecting,tang2017privacy} and governmental agencies~\cite{gdpr,ccpa,10.1145/3219819.3226070}, and has been applied to a variety of tasks including private query answering~\cite{Matrix,MWEM,PMW,HDMM,PrivateSQL,dong2022r2t}, synthetic data generation~\cite{ChenOF20,LiuV0UW21,Li21,AydoreBKKM0S21,GeMHI21,LH21,Torkzadehmahani19}, and private machine learning~\cite{DBLP:conf/ccs/AbadiCGMMT016,DBLP:conf/iclr/PapernotAEGT17,DBLP:conf/uss/Jayaraman019}. 
Despite its broad success, DP introduces a fundamental challenge: while it ensures privacy, the injected noise and limited data access make it difficult to evaluate the \emph{quality} and \emph{fairness} of private data. Users and analysts must often rely on privacy-protected datasets to make decisions or train models, without being able to assess whether these datasets are accurate, representative, or equitable. Furthermore, learning from noisy or biased data can yield inaccurate and discriminatory predictions~\cite{RenggliRK0023,shah2025comprehensive}. 

Fairness and bias in data are central concerns in data management and ML, since biased data can produce unfair outcomes even when algorithms are well-designed. Extensive research has focused on defining algorithmic fairness, proposing formal notions such as Demographic Parity ~\cite{DBLP:conf/icdm/CaldersKP09} and Conditional Statistical Parity~\cite{DBLP:conf/kdd/Corbett-DaviesP17}, among many others~\cite{KleinbergMR17,DworkHPRZ12,barocas2023fairness,heidari2019moral}. 
These notions constrain model behavior to ensure equitable treatment across protected groups. 
However, such definitions focus solely on the model, neglecting the data which itself may contain bias or unfairness. 
If the dataset encodes biased relationships between protected and outcome attributes, an algorithm may reproduce or amplify those disparities.

In this work, we shift the focus from algorithmic fairness to the \emph{fairness of data}.
A key observation underlying our work is that many algorithmic fairness definitions can be expressed as forms of \emph{conditional independence}~\cite{InterFair}. 
For example, Demographic Parity requires independence between the sensitive attribute and the model outcome, while Conditional Statistical Parity introduces a conditioning set of admissible attributes. 
Using this formulation, we reinterpret fairness as a property of the \emph{data distribution} itself rather than only of the learned classifier. This is motivated by a common assumption from a prior work~\cite{InterFair,DBLP:conf/kdd/FeldmanFMSV15,DBLP:conf/nips/CalmonWVRV17}: if the labels in a dataset reflect the outcomes a reasonable classifier would produce, then dataset fairness becomes a reliable proxy for assessing the fairness of a classifier trained on it. Moreover, using the data in other ways, such as getting summarized statistics, may lead to biased conclusions. 
Consequently, testing whether a dataset satisfies such conditional independence statements provides a natural and principled way to assess its fairness. 
This connection enables the adaptation of well-established algorithmic fairness notions into \emph{database-level fairness criteria}, which serve as the foundation for the framework we propose. 
 
Our goal is to develop a principled, quantitative framework for measuring data unfairness that remains meaningful even when sufficient noise is added to it to satisfy DP.
Prior research on measuring data quality has largely examined logical inconsistencies in data through integrity constraints, both without~\cite{DBLP:conf/sum/Bertossi18,DBLP:journals/corr/abs-1904-03403,LivshitsKTIKR21} and, recently, with privacy guarantees~\cite{mohapatra2025computinginconsistencymeasuresdifferential}. 
Yet, despite the growing literature on fairness in algorithms, no existing framework directly measures the unfairness of the private data itself, even though such measurements could give a valuable signal for data equity when the data cannot be fully observed.

Here, we introduce \emph{differentially private unfairness measures} that, given a database $D$ and a set of fairness criteria $\constraints$, compute a numerical score $\measure(\constraints,D)$, quantifying how far the data deviates from satisfying those criteria. 
We formalize three key desiderata that any unfairness measure should satisfy, following principles devised for inconsistency measures~\cite{LivshitsKTIKR21}. 
First, the \emph{Positivity} property~\cite{DBLP:journals/corr/abs-1904-03403,DBLP:journals/ijar/GrantH17} ensures that the measure is non-negative and equals zero if and only if $D$ satisfies all fairness criteria.
Second, the \emph{Monotonicity} property~\cite{DBLP:journals/corr/abs-1904-03403} states that expanding the set of fairness criteria ($\constraints_1 \subseteq \constraints_2$) cannot reduce the measured unfairness.
Third, the measure must be efficiently and accurately computable under DP, maintaining interpretability and utility despite the added noise.

\begin{figure}
\hspace{-3mm}
  \begin{subfigure}[b]{0.158\textwidth}
    \includegraphics[width=0.9\linewidth]{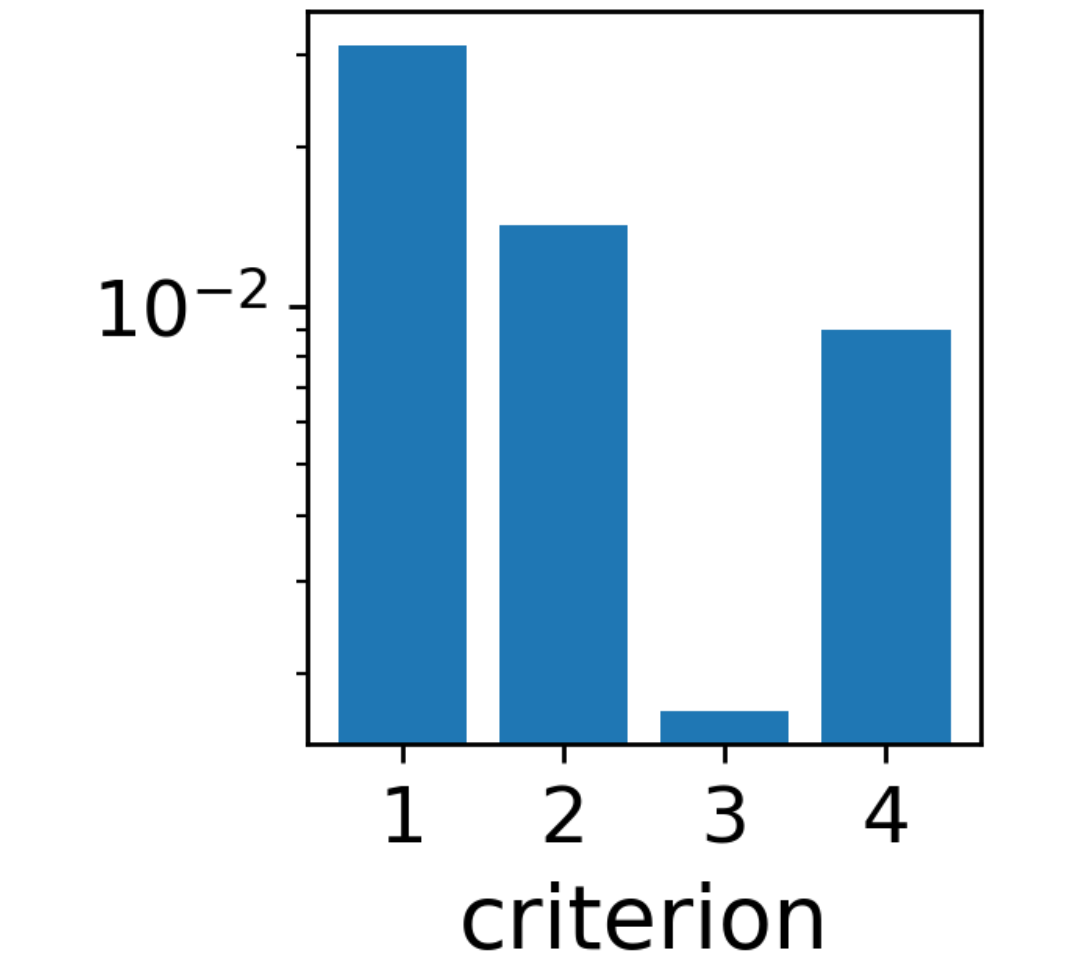}
    \caption{Values of $\mutualproxytvd$.}
    \label{fig:plot-0-1}
  \end{subfigure}
  \begin{subfigure}[b]{0.15\textwidth}
    \includegraphics[width=0.9\linewidth]{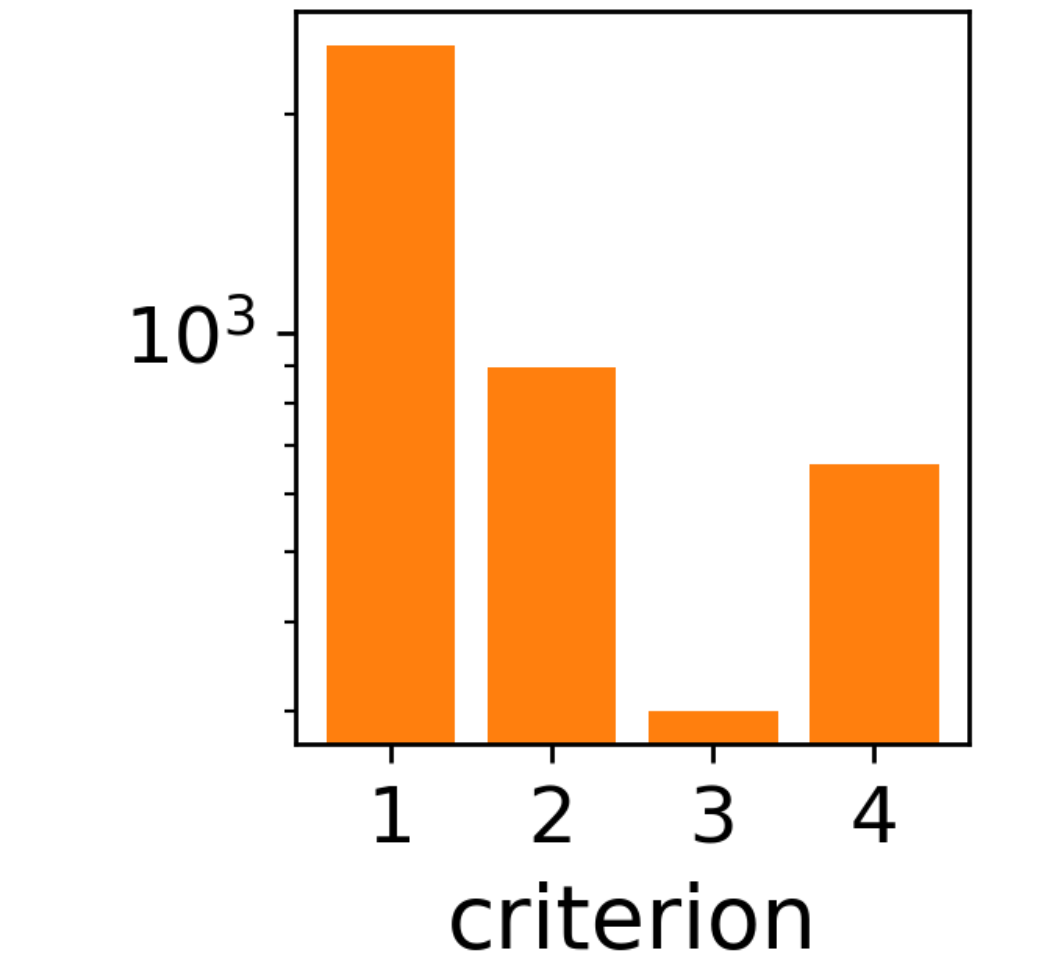}
    \caption{Values of $\repairsat$.}
    \label{fig:plot-0-2}
  \end{subfigure}
  \begin{subfigure}[b]{0.15\textwidth}
    \includegraphics[width=0.9\linewidth]{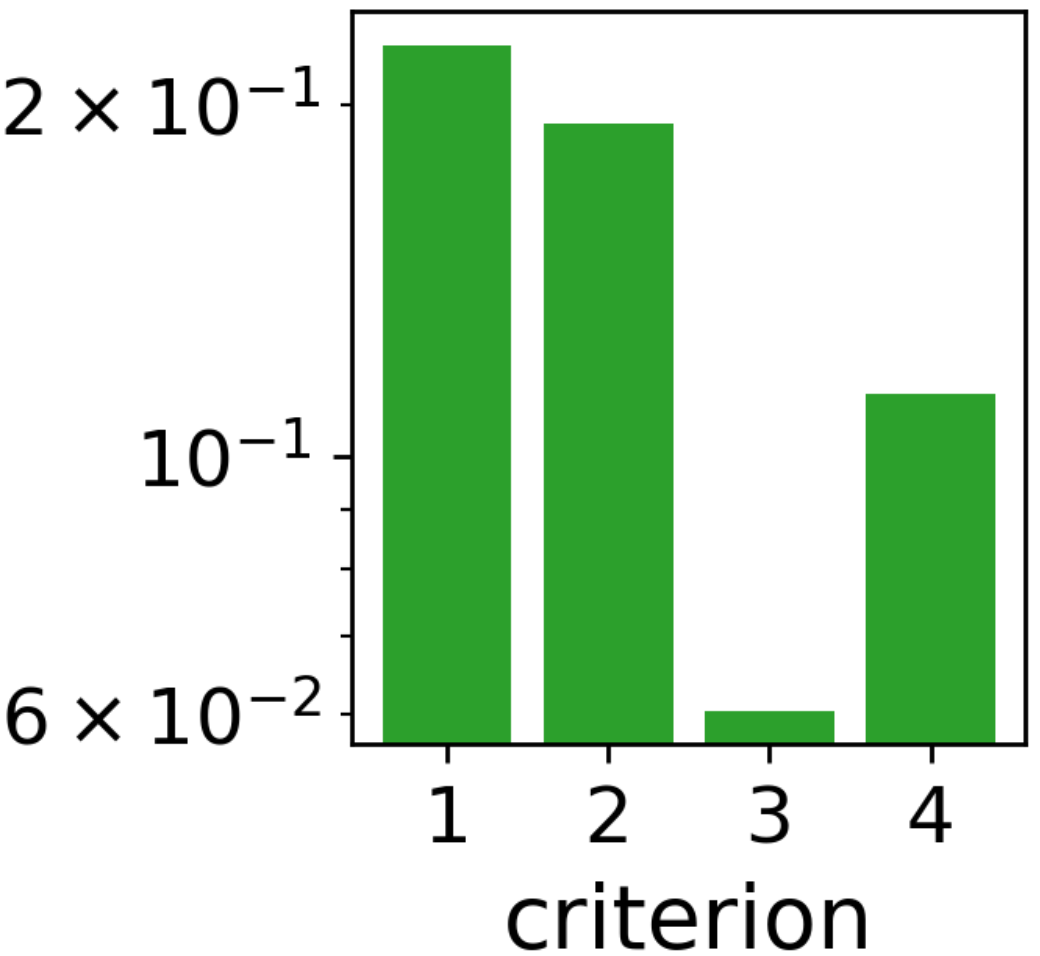}
    \caption{Values of $\contribution$.}
    \label{fig:plot-0-3}
  \end{subfigure}%
  \caption{Values of the unfairness measures (log scale) on the \texttt{Adult} dataset for the four fairness criteria in \Cref{ex:intro}.}
  \label{fig:intro-example-1}
\end{figure}

\begin{figure}
\vspace{-2mm}
  \begin{subfigure}[b]{0.45\linewidth}
    \includegraphics[width=0.8\linewidth]{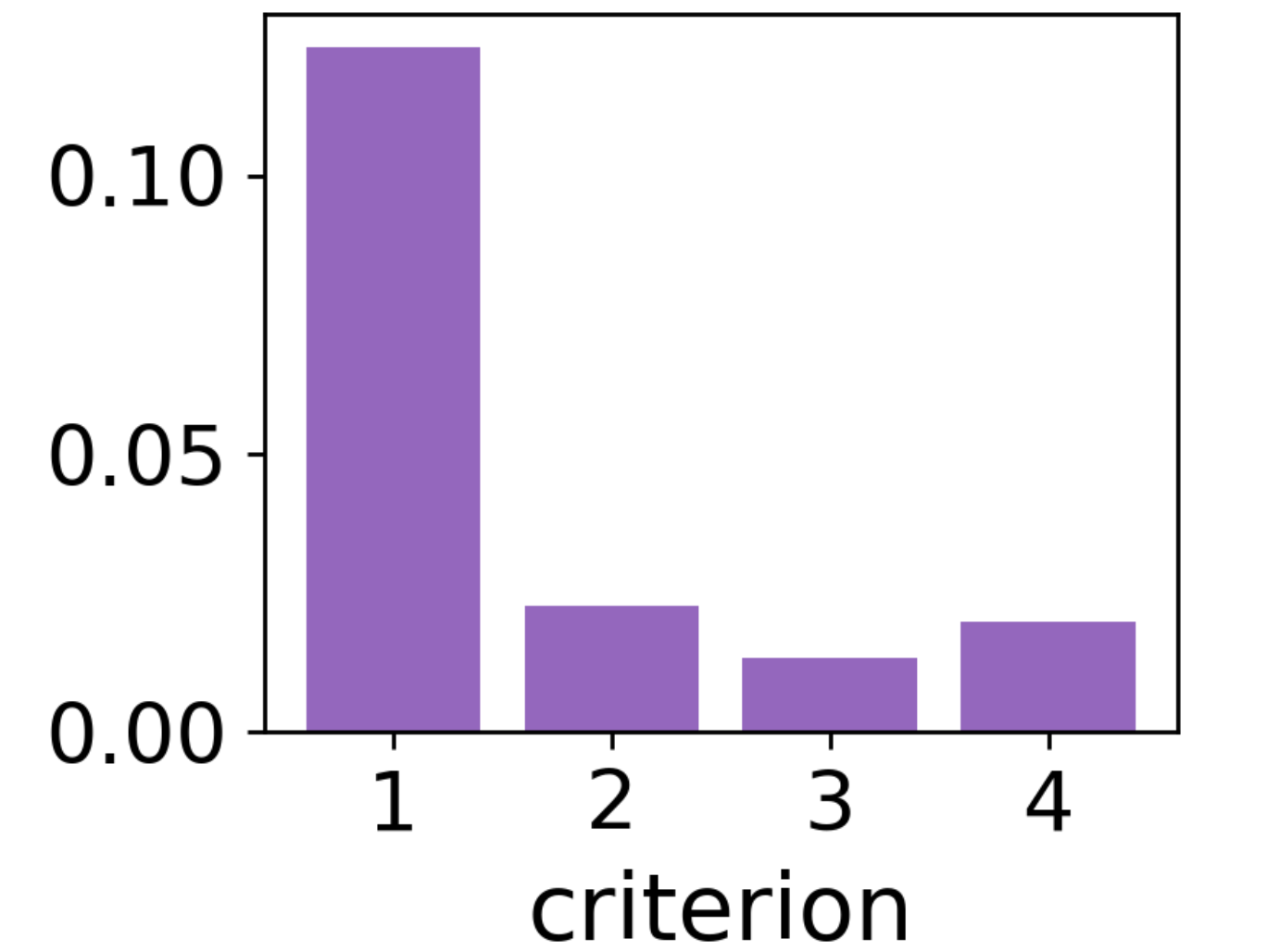}
    \caption{Random Forest.}
    \label{fig:plot-0-4}
  \end{subfigure}
  \begin{subfigure}[b]{0.41\linewidth}
    \includegraphics[width=0.8\linewidth]{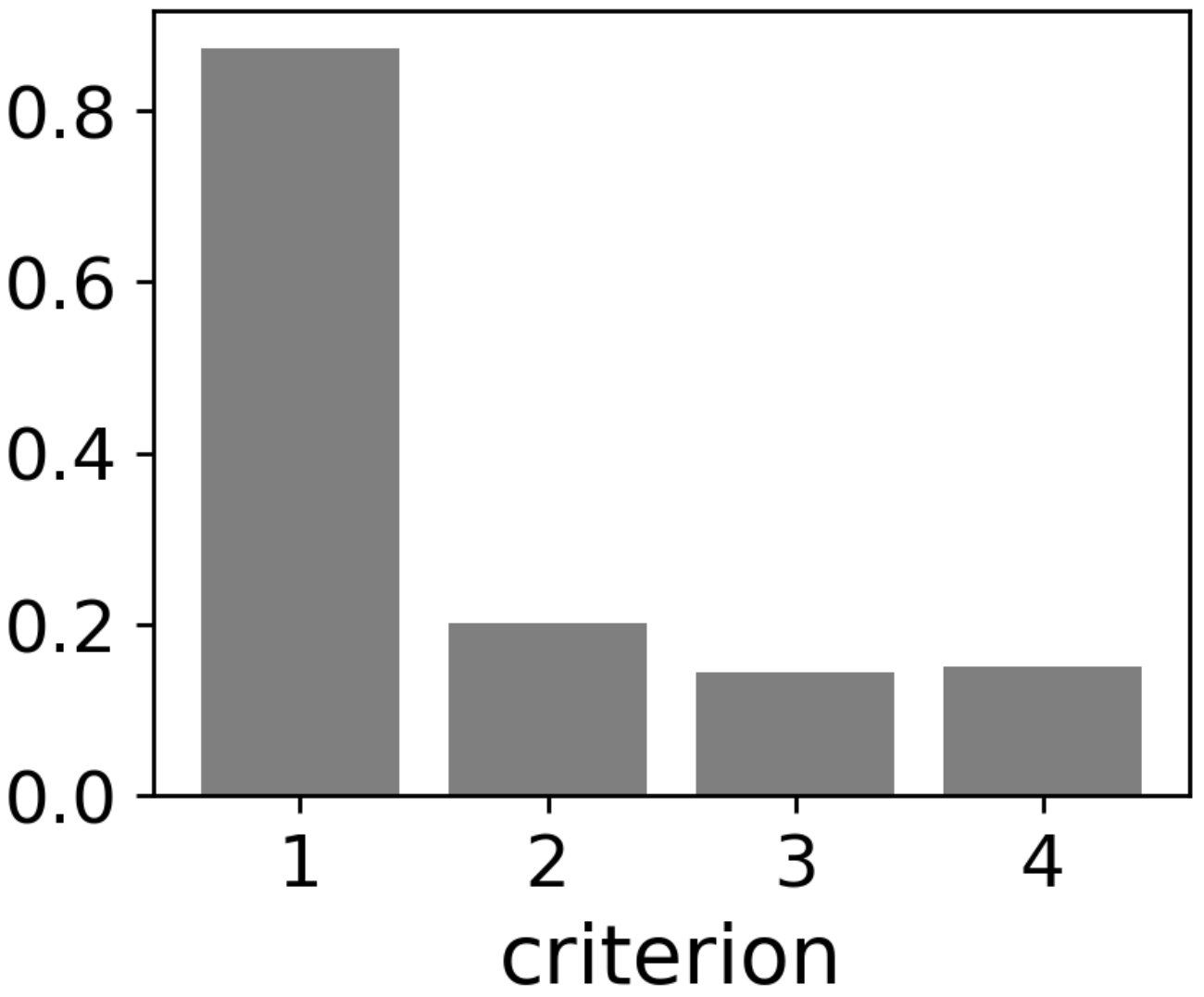}
    \caption{Neural Network.}
    \label{fig:plot-0-5}
  \end{subfigure}%
  \caption{Fairness values (Demographic Parity and Conditional Statistical Parity gaps) of the two privately-trained ML models on the \texttt{Adult} dataset in~\Cref{ex:intro}.}
  \label{fig:intro-example-2}
\end{figure}

We present three complementary measures grounded in probabilistic dependence. 
The first, $\mutualproxytvd$, uses mutual information (MI)~\cite{cover_thomas_2005} to quantify dependence between sensitive and outcome attributes. 
Because MI is highly sensitive and thus unsuitable for DP~\cite{PrivBays}, we develop a Total Variation Distance (TVD) proxy that closely approximates MI in both theory and practice. 
The second measure is based on data repair~\cite{LivshitsKR20,InterFair,ChuIP13}, quantifying the minimal number of tuple modifications required to make a dataset fair. 
Since this is computationally hard, we propose a proxy measure, $\repairsat$, reducible to the weighted MaxSAT problem~\cite{vazirani2001approximation}, following~\cite{InterFair}. 
Finally, inspired by tuple contribution analysis~\cite{wu2013scorpion,MeliouGMS11,DeutchFGS21,LivshitsBKS21}, we introduce a top-$k$ contribution measure, $\contribution$, that captures how much the most influential tuples contribute to fairness violations. 

While $\mutualproxytvd$ approximates MI directly, $\repairsat$ and $\contribution$ introduce novel tuple-level unfairness notions tailored for DP, offering restricted sensitivity and greater interpretability than a distribution-level formula. 

We prove that all three measures satisfy our desiderata, are close approximations to their originals, and exhibit low sensitivity relative to their range. 
We then design privacy-preserving algorithms for them that rely on the Laplace mechanism~\cite{LM}, and analyze their utility and complexity. 
Extensive experiments on five real-world datasets across diverse fairness settings confirm that our measures reliably capture comparative unfairness, remain faithful to their non-private baselines, and can be computed effectively under DP.

\begin{example}\label{ex:intro}
Consider the \texttt{Adult} dataset
~\footnote{See the descriptions and links to the datasets in \Cref{sec:setup}.\label{footnote:datasets}}
containing personal information of individuals in the US including whether their income is larger than $50$K (we sample $10,000$ tuples with $15$ attributes). 
\Cref{fig:intro-example-1} depicts the results of applying our three measures to four fairness criteria: (1) a person's income should not depend on their years of education, (2) a person's income should not depend on their sex, (3) a person's income should not depend on their race, and (4) a person's income should not depend on their sex, given that they work a certain number of hours per week (the criteria are phrased as independence statements in~\Cref{tab:fairness-criteria} in the sequel).
We expect the first criterion to clearly be violated, as more education often leads to higher salaries, while the latter three criteria are natural and desired for most applications to avoid discriminatory conclusions and decision-making. Furthermore, it is known that \texttt{Adult} has disparity in income for different sexes and races~\cite{DBLP:conf/kdd/ThanhRT11,DBLP:conf/eurosp/TramerAGHHHJL17}. 
Each measure was evaluated using our DP algorithms with a privacy budget of $\varepsilon=1$, with the experiment repeated ten times and the results averaged for each criterion.
All the measures show a similar trend for criteria (1), (2), (3) and (4).

We additionally privately trained two models on the same sample using a budget of $\varepsilon = 10$ over all attributes to predict the outcome: (a) a DP version of \texttt{RandomForest}, 
which achieved $78.9\%$ accuracy on average, and (b) a neural network with one hidden layer with $32$ nodes and ReLU activation, trained with \texttt{DP-SGD}, which achieved $84.1\%$ accuracy on average.
\cut{
\begin{itemize}
    \item{DP-version of \texttt{RandomForest}, which achieved $78.9\%$ accuracy on average.}
    \item A neural network with one hidden layer with $32$ nodes and ReLU activation, trained with \texttt{DP-SGD}, which achieved $84.1\%$ accuracy on average.
\end{itemize}
}
\Cref{fig:intro-example-2} shows the Demographic Parity gaps for criteria (1), (2) and (3), based on the classifier's predictions:
{\small $|\max_{\protectedvalue \in \protectedset} \prob\left( \outcome = 1 \mid \protectedset = \protectedvalue \right) - \min_{\protectedvalue \in \protectedset} \prob\left( \outcome = 1 \mid \protectedset = \protectedvalue \right)|$},
where $\outcome = 1$ is a prediction of the income being bigger than $50$K. 
Intuitively, this measures the disparity in positive outcomes for the most privileged group and the most discriminated one in a population consisting of more than two groups.
The figure also shows the Conditional Statistical Parity gap for the conditional criterion (4), which is calculated as the expectation of the Demographic Parity over the conditioned attribute~\cite{InterFair,DBLP:journals/pvldb/PujolGM23}.
We can see that the trend of our measures matches the trend of the `fairness measures' commonly applied to ML models, indicating that the three measures are able to estimate the underlying unfairness based on the criteria, and that they might also be useful in predicting whether private data can be effectively used for various applications.
\end{example}

\paratitle{Contributions}\label{par:contribution}
This is the first work that provides a practical framework for quantifying private data unfairness through three different notions, specifically adapted for DP. 
In summary, the paper makes the following contributions:
\begin{enumerate}[wide, labelwidth=!, labelindent=0pt]
    \item We introduce a \textbf{formal framework for measuring data unfairness}, focusing on quantifying bias directly at the data level rather than at the algorithmic level, and defining general desiderata of \emph{positivity}, \emph{monotonicity}, and \emph{DP computability}.
    \item We propose \textbf{three concrete unfairness measures}, grounded respectively in (a) \emph{mutual information} and its \emph{total variation distance} proxy suitable for DP, \mutualproxytvd, (b) a \emph{data repair} inspired measure, \repairsat, approximated via a reduction to \emph{weighted MaxSAT}, and (c) \emph{tuple-level contribution} measure, \contribution, identifying the top-$k$ most influential records in fairness violations.
    \item We design \textbf{DP algorithms} for computing the measures, and give formal analysis of their guarantees, error, and complexity.
    \item We conduct an \textbf{extensive experimental evaluation} on five real-world datasets, demonstrating that our proposed measures can be used in data analysis scenarios, 
    demonstrate that they provide a complementary, stable, and reliable alternative to estimating unfairness via machine learning models, and that they follow trends of increasing data unfairness. Finally, we show that the measures faithfully approximate their non-private counterparts and effectively quantify data unfairness under DP, and scale to large datasets. 
\end{enumerate}


\section{Preliminaries}\label{sec:prelim}

We now give the necessary background for the paper, including basic database and probabilistic notions, fairness concepts, and essential definitions for our use of differential privacy.


\subsection{Databases and Probabilities}
We consider a single-relation schema $\attrset = (A_1, \ldots, A_m)$, which is a vector of distinct attribute names $A_i$, each associated with a domain $\dom(A_i)$ of values. A database $D$ over $\attrset$ is associated with a set $\tids(D)$ of \emph{tuple identifiers}, and it maps every identifier $i\in\tids(D)$ to a tuple $D[i]=(a_1,\dots,a_m)$ in $A_1\times\dots\times A_m$. 
We denote the size of the database by $|D| = n$. 
A tuple $t_i \in D$ has a specific value in each of its attributes, denoted by $t_i[A] = a$ where $a\in \dom(A)$. 
We consider bag semantics where $t_i, t_j$ may share the same values across all attributes except their identifiers. 



The computation of empirical probabilistic quantities in databases is often practically performed by measuring database statistics~\cite{heidari2019holodetect,fariha2019example,zhang2020statistical}. 
Let $D$ be a dataset of size $n$, the {\em empirical marginal} of a value $A = a$ is defined as follows: 
$\prob(A = a) = \frac{|\{t \mid t[A] = a\}|}{n}$ 
Similarly, the {\em conditional marginal} of $A = a$ conditioned on $C = c$ is 
{\small
$$\prob(A = a \mid C = c) = \frac{|\{t \mid t[A] = a,t[C] = c\}|}{|\{t \mid t[C] = c\}|}$$
}
When a set of attributes $\protectedset$ equals a set of corresponding values $p$, we abuse notation and write $\protectedset = p$ when clear from context.



Employing these notations, we review the definition of mutual information, which is commonly used to quantify dependencies between variables and will be the basis for one of our measures.  

\begin{definition}[Mutual Information (MI) \cite{cover_thomas_2005}]\label{def:mutual-information}
\cut{
The mutual information between two attribute set $A$ and $B$ over their domain in a database $D$ is defined as follows.
\begin{small}
\begin{equation*}
\begin{split}
\mathcal{MI}&_D(A, B) = \\ 
    &\sum_{a \in \dom(A)}\quad \quad\sum_{\mathclap{b \in \dom(B)}} \prob(A = a, B = b) \log\left(\frac{\prob(A = a, B = b)}{\prob(A = a)\prob(B = b)}\right)
\end{split}
\end{equation*}
\end{small}
}

\cut{
Similarly, the conditional mutual information between two attributes $A$ and $B$ conditional on attribute $C$ over their domain in a database $D$ is defined as follows.

\begin{small}
\begin{equation*}
\begin{split}
& \mathcal{MI}_D(A, B \mid C) =  \sum_{c \in \dom(C)} \prob(C = c) \sum_{a \in \dom(A)} \sum_{b \in \dom(B)} \\
& \quad \prob(A = a, B = b \mid C = c) \log\left(\frac{\prob(A = a, B = b \mid C = c)}{\prob(A = a \mid C = c)\prob(B = b \mid C= c)}\right)
\end{split}
\end{equation*}
\end{small}
If $C$ is absent, it falls back to the non-conditional case.
}

The conditional mutual information between two attributes $A$ and $B$ conditioned on attribute $C$ over their domain in a database $D$ is defined as follows.
\begin{small}
\begin{equation*}
\begin{split}
& \mathcal{MI}_D(A, B \mid C) =  \sum_{c \in \dom(C)} \prob(C = c) \sum_{a \in \dom(A)} \sum_{b \in \dom(B)} \\
& \quad \prob(A = a, B = b \mid C = c) \log\left(\frac{\prob(A = a, B = b \mid C = c)}{\prob(A = a \mid C = c)\prob(B = b \mid C= c)}\right)
\end{split}
\end{equation*}
\end{small}
where $C$ may be absent for the unconditional case.
\end{definition} 

We say that two attributes $A, B \in \attrset$ in a database $D$ are independent if $\prob(A = a, B = b) = \prob(A = a)\cdot \prob(B = b)$ for all $a \in \dom(A)$ and $b \in \dom(B)$, or, equivalently, if $MI_D(A, B) = 0$, and we use the standard notation of independence $A \independent B$. 


The notion of database probabilities can be used to express associational fairness definitions, as we discuss next.

\subsection{Fairness as Independence Statements}\label{sec:fairness}

Multiple fairness definitions that have been considered by previous work can be described as (conditional) independence statements. 
Demographic Parity~\cite{DBLP:conf/icdm/CaldersKP09} measures whether the favorable prediction $\hat{\outcomeset}$ of an outcome (e.g., an income will be above a threshold) is affected by which demographic group the person belongs to, as captured by a protected attribute $\protectedset \in \attrset$ (e.g., Race, Sex). This statement can also be expressed by $\protectedset \independent \hat{\outcomeset}$. 
Conditional Statistical Parity~\cite{DBLP:conf/kdd/Corbett-DaviesP17} further relaxes this definition by allowing certain admissible attributes $\admissibleset \in \attrset$ (e.g., education level, work experience) to explain the unfairness, and can be written as $\protectedset \independent \hat{\outcomeset} \mid \admissibleset$. Equality of Opportunity~\cite{heidari2019moral} defines fairness by stating that for each case of the ground truth (e.g., the loan is indeed approved), as captured by the outcome attribute $\responseset \in \attrset$, the probability of a favorable outcome should be equal across all demographic groups, i.e., $\protectedset \independent \hat{\outcomeset} \mid \responseset$. Predictive Parity~\cite{Chouldechova17} measures whether predicted positives are equally accurate across demographic groups, or equivalently, whether the ground truth is independent of the protected attribute given the predicted outcome, i.e., $\protectedset \independent \responseset \mid \hat{\outcomeset}$.

We will simplify the discussion for the purpose of defining database fairness and distinguish between a protected attribute $\protectedset \in \attrset$, an outcome attribute $\responseset \in \attrset$, and an admissible attribute $\admissibleset \in \attrset$. 
Thus, we can generally define a {\em database fairness criterion} as follows.

\begin{definition}[Fairness criterion]\label{def:fairness}
Given a schema $\attrset$ with a protected attribute $\protectedset \in \attrset$, outcome attribute $\responseset \in \attrset$, and admissible attribute $\admissibleset \in \attrset$, 
a {\em fairness criterion} is an independence statement $\protectedset \independent \responseset \mid \admissibleset$, where $\admissibleset$ may be absent. 
\end{definition}



A set of multiple fairness criteria will be denoted by $\constraints$. 


We now consider the fact that the databases we wish to analyze are protected with differential privacy.

\subsection{Differential Privacy}\label{sec:privacy}
We next give the necessary definitions for DP.
\begin{definition}[Neighboring Databases]\label{def:neighboring-databases}

Two databases $D,D'$ with the same schema are called neighboring if they have the same size and differ in exactly one tuple, 
denoted by $D' \neighbor D$.
\end{definition}



We often use the notion of neighboring databases to distinguish the impact of any particular individual's input on the output of a function. We likewise measure the maximum change in any function 
due to the replacement of a single tuple in the database, often calling it the sensitivity of the function.

\begin{definition}[Sensitivity]\label{def:sensitivity}
Given a function $f$ the sensitivity of $f$ is $max_{D' \neighbor D}{|f(D)-f(D')|}$ and is denoted by $\sens_f$.
\end{definition}

Differential privacy (DP) \cite{DiffPriv} protects the private information of individuals in the data by ensuring similar results for similar databases with high probability.

\begin{definition}[Differential Privacy \cite{DiffPriv}]
Given a privacy budget $\varepsilon > 0$, an algorithm $\mathcal{M}$ is said to satisfy $\varepsilon$-DP if for all $S \subseteq \mbox{Range}(\mathcal{M})$ and for all $D \neighbor D'$,
    $$Pr(\mathcal{M}(D)\in S) \leq e^\varepsilon Pr(\mathcal{M}(D')\in S)\,$$
\end{definition}

The Laplace mechanism~\cite{LM} allows us to enforce DP by adding calibrated noise to algorithm results and is a common building block in DP mechanisms. 


\begin{definition}[Laplace Mechanism~\cite{LM}]\label{def:LM}
Given a database $D$, a function $f$ : $\mathcal{D} \rightarrow \mathbb{R}$, and a privacy budget $\varepsilon$, the Laplace mechanism $\mathcal{M}_L$ returns $f(D) + \nu_q$, where $\nu_q \sim \mathrm{Lap}(\Delta_f/\varepsilon)$. 
\end{definition}

\begin{theorem}[Privacy of the Laplace Mechanism~\cite{dwork2014:textbook}]\label{thm:LM-DP}
Let $f:\mathcal{D}\to\mathbb{R}$ be a query with the sensitivity $\sens_f$, and fix $\varepsilon>0$. The Laplace mechanism $\mathcal{M}_L(D)=f(D)+\nu$, with $\nu \sim \mathrm{Lap}(\Delta_f/\varepsilon)$, satisfies $\varepsilon$-differential privacy.
\end{theorem}

\section{Privately Evaluating Unfairness}\label{sec:model}

{\small
\begin{table*}
\centering
\caption{Properties of the discussed unfairness measures along with their computation and utility costs. Green rows denote the chosen measures and red rows denote measures that are unsuitable either due to relatively high sensitivity ($\mutual$), unsatisfactory faithfulness to the original measure ($\mutualproxybayes$), or high computation costs (\repair). All measures satisfy the positivity and monotonicity properties outlined in~\Cref{sec:model}, except $\mutualproxybayes$ that violates positivity.}
\label{tbl:measure-summary}
\renewcommand{\arraystretch}{1.75}
\begin{tabular}{|>{\centering\arraybackslash}p{5.2cm}|>{\centering\arraybackslash}p{1cm}|>{\centering\arraybackslash}p{1.5cm}|>{\centering\arraybackslash}p{1.6cm}|>{\centering\arraybackslash}p{2cm}|c|c|c|c|c|}
\hline
\textbf{Motivation} & \textbf{Measure} & \textbf{Range for \fairnessdef} & 
\textbf{$\Delta_\measure$ for $\fairnessdef$} & 
\textbf{$\Delta_\measure$ for $\constraints$} & \textbf{Complexity} & \textbf{Error Bound} \\
\hline
\multirow{3}{*}{Attribute correlation} & 
\cellcolor{red!18} $\mutual$ &\cellcolor{red!18} $[0, \infty)$ & \cellcolor{red!18} $\mathcal{O}\left(\frac{\log n}{n}\right)$ &\cellcolor{red!18} - &\cellcolor{red!18} - &\cellcolor{red!18} - \\
& \cellcolor{red!18} $\mutualproxybayes$ &\cellcolor{red!18} $[-1, 0]$ &\cellcolor{red!18} $\mathcal{O}\left(\frac{1}{n}\right)$ &\cellcolor{red!18} - &\cellcolor{red!18} - &\cellcolor{red!18} -\\
&\cellcolor{green!18} $\mutualproxytvd$ &\cellcolor{green!18} $[0, 2]$ &\cellcolor{green!18} $\mathcal{O}\left(\frac{1}{n}\right)$ &\cellcolor{green!18} $\mathcal{O}\left(\frac{|\constraints|}{n}\right)$ &\cellcolor{green!18} $\mathcal{O}\left(|\constraints| n\right)$ &\cellcolor{green!18} $\frac{16|\constraints|}{n \varepsilon}$ \\
\hline
\hline
\multirow{2}{*}{Distance to a fair database} & 
\cellcolor{red!18} $\repair$ 
&\cellcolor{red!18} $[0, n]$ &\cellcolor{red!18} $1$ 
&\cellcolor{red!18} - 
&\cellcolor{red!18} Exp. in $n$ &\cellcolor{red!18} -\\
&\cellcolor{green!18} $\repairsat$ &\cellcolor{green!18} $[0, n]$ &\cellcolor{green!18} $2$ &\cellcolor{green!18} $2|\constraints|$ &\cellcolor{green!18} $\mathcal{O}\left(|\constraints|\left(n^4 + SAT\right)\right)$ &\cellcolor{green!18} $\frac{2|\constraints|}{\varepsilon}$ \\
\hline
\hline
Contribution of the top-$k$ tuples to unfairness &\cellcolor{green!18} $\contribution$ &\cellcolor{green!18} $[0, \min\{\frac{k}{4},2\}]$ &\cellcolor{green!18} $\mathcal{O}\left(\frac{k}{n}\right)$ &\cellcolor{green!18} $\mathcal{O}\left(\frac{|\constraints|k}{n}\right)$ &\cellcolor{green!18} $\mathcal{O}\left(|\constraints|n \log n\right)$ &\cellcolor{green!18} $\frac{7k|\constraints|}{n\varepsilon}$ \\
\hline
\end{tabular}
\end{table*}
}

We formulate the problem of privately evaluating the unfairness of a database, starting with a discussion about measure requirements. 

\paratitle{Measure desiderata}
We define several requirements for private unfairness measures. The first two are inspired by the properties outlined for inconsistency measures in the context of logical integrity constraints~\cite{LivshitsKTIKR21}, and the third one is specific to the DP setting. 
We aim to design an unfairness measure that takes a private database $D$, a set of fairness criteria $\constraints$, and a privacy budget $\varepsilon$, and returns a non-negative number, i.e., $\measure(\constraints, D) \in [0,\infty)$.  
We want $\measure$ to have the following desiderata: 
\begin{enumerate}[leftmargin=*]
\item {\bf Positivity:} $\measure(\constraints, D) \geq 0$ and $\measure(\constraints, D) = 0$ iff $D$ satisfies $\fairnessdef$ for every $\fairnessdef \in \constraints$.
\item {\bf Monotonicity:} $\measure(\constraints, D) \leq \measure(\constraints', D)$ if $\constraints \subseteq \constraints'$. 
\item {\bf Computability under DP:} We can efficiently compute $\measure$ with an $\varepsilon$-DP algorithm with relatively small error.
\end{enumerate}

The first two properties have already been suggested for inconsistency measures. Positivity was proposed as a fundamental property in~\cite{DBLP:journals/corr/abs-1904-03403,
DBLP:journals/ijar/GrantH17} and as an axiom in~\cite{DBLP:conf/ecsqaru/MartinezPSSP07}. Monotonicity was proposed in different variations by previous work. 
\citet{DBLP:journals/corr/abs-1904-03403} proposed measure monotonicity for database containment, i.e., $D \subseteq D'$ implies $\measure(\constraints, D) \leq \measure(\constraints, D')$, yet, in our case, database containment does not guarantee an increase of fairness, and in fact, may decrease it. 
\citet{LivshitsKTIKR21} suggested monotonicity for logical implication of integrity constraints, i.e., if $\Sigma \models \Sigma'$ for two sets of integrity constraints $\Sigma, \Sigma'$, then $\measure(\Sigma', D) \leq \measure(\Sigma, D)$. 
Since comparing fairness criteria (see \Cref{def:fairness}) through logical implication is not possible, we define a new notion of monotonicity for fairness criteria {\em set containment}. Intuitively, when more fairness criteria are used, the requirements from the database are stricter and, therefore, the value of the unfairness measure should only increase. 
The third property will be shown in two parts. First, showing that our measures have low sensitivity (\Cref{sec:model}) and then providing a DP algorithm for computing them (\Cref{sec:algorithms}).

In the sequel, we will first define the measures for a single fairness criterion, $\measure(\fairnessdef,D)$, 
and then extend them to a set of criteria as 
a sum $\measure(\constraints, D) = \sum_{\fairnessdef \in \constraints} \measure(\fairnessdef, D)$. 
We will then prove the properties for this extension. 
As we will see, measures and techniques from previous work in the non-private setting can be unsuitable for computation under DP, motivating our third desideratum and leading us to search for proxies or replacements.

The other two properties listed in~\cite{LivshitsKTIKR21} are tailored for dynamic settings and are discussed in detail at the end of this section.

The remainder of \Cref{sec:model} presents our three measures, outlining their motivation and theoretical properties. We begin with $\mutualproxytvd$, which approximates mutual information between the attributes appearing in the fairness criteria (\Cref{sec:mutual}). We then introduce $\repairsat$, a proxy for the minimal number of tuple insertions or deletions required to transform the dataset into one that satisfies the criteria (\Cref{sec:repair}). 
Finally, we describe $\contribution$, which reports the cumulative contribution of the top-$k$ most influential tuples to the fairness criteria violation (\Cref{sec:contribution}). 
Unlike $\mutualproxytvd$, which approximates MI at the distribution level, $\repairsat$ and $\contribution$ are grounded in tuple-level notions of unfairness, compatible with DP, and greater interpretability.
\Cref{tbl:measure-summary} provides a consolidated overview of the properties of all three measures.

\subsection{Unfairness as Mutual Information}\label{sec:mutual}

A standard way to measure dependence between variables is computing the mutual information (MI) between them~\cite{kraskov2004estimating,RyanMnist,PrivBays}.  Thus, for a fairness criterion $\protectedset \independent \responseset \mid \admissibleset$, we can measure the MI between $\protectedset$ and $\responseset$ given $\admissibleset$ to obtain our first measure notion. 
We abuse the notation from \Cref{def:mutual} and denote the sum of (conditional) mutual information between the attributes in the fairness criteria contained in $\constraints$ by $MI_D(\constraints)$, i.e., $MI_D(\constraints) = \sum_{\fairnessdef \in \constraints} MI_D(\fairnessdef)$.

\begin{definition}[Mutual Information Unfairness]\label{def:mutual}
Given a database $D$ and a fairness criterion of the form $\fairnessdef = \protectedset \independent \responseset \mid \admissibleset$ where $\admissibleset$ can be an empty set, the mutual information unfairness measure is defined as $\mutual(\fairnessdef,D) = MI_D(\protectedset,  \responseset \mid \admissibleset)$.
\end{definition}

Previous work~\cite{PrivBays} has already shown that $\mutual$ is not suitable as an accurate measure in the DP setting due to its relatively high sensitivity of $\mathcal{O}\left(\log n / n\right)$. 
Such sensitivity means that Laplace noise of scale proportionate to $\log n / (\varepsilon n)$ should be added to achieve $\varepsilon$-differential privacy. However, the range of $\mutual$ is $[0,\infty)$, and the closer we get to $\fairnessdef$ being satisfied (that is, the more independent the attributes in the $\fairnessdef$ are), the smaller $\mutual$ becomes. Thus, when $\fairnessdef$ nearly holds and $\mutual$ is close to $0$, the added Laplace noise can severely distort the original $\mutual$ value, making it unusable.
The range of $\mutual$ also poses a problem. Since it is unbounded, it is difficult to interpret its values and distinguish which values indicate low dependence score between the sensitive and outcome attributes, and which values indicate high dependence score. 
Therefore, we need to find an alternative with lower sensitivity.

To mitigate these issues, one can use a proxy function for $\mutual$ that has low sensitivity and a bounded range. 
A possible proxy function, presented in~\cite{PrivBays} for a single fairness criterion, is $\mutualproxybayes(\fairnessdef, D) = -\frac{1}{2}min_{Pr^\diamond \in \mathcal{P}^\diamond} \|Pr^\diamond[\protectedset, \responseset \mid \admissibleset] - Pr[\protectedset, \responseset \mid \admissibleset]\|$, where $\Pr^{\diamond}[\protectedset, \responseset \mid \admissibleset]$ is the maximum joint distribution defined as the one that maximizes the mutual information between $\protectedset$ and $\responseset$, given $\admissibleset$. 
\citet{PrivBays} also showed that the sensitivity of $\mutualproxybayes$ is $\frac{1}{n}$ for an unconditional fairness criterion, and is therefore suitable for DP.


However, $\mutualproxybayes$ lacks the desired positivity property. Furthermore, the relation between $\mutual$ and $\mutualproxybayes$ can be difficult to visually comprehend. 
\ifpaper
\else
\Cref{fig:mi-proxies-comparison} demonstrates the faithfulness between $\mutual$ (\Cref{fig:mutual-comparison-1}) and $\mutualproxybayes$ (\Cref{fig:mutual-comparison-2}) for the \texttt{Adult}, \texttt{Stackoverflow} survey, and \texttt{Compas} datasets
(see~\Cref{footnote:datasets}) 
with fairness criteria from \Cref{tab:fairness-criteria}, numbered for each dataset. We also plotted $\mutualproxybayes$ with an offset of $\frac{1}{2}$ (\Cref{fig:mutual-comparison-3}) to make its values positive so the reader could easily see the variation between $\mutual$ and $\mutualproxybayes$ across all datasets and criteria.
\fi

This faithfulness gap motivates us to define a novel proxy function for $\mutual$ that has {\em low sensitivity and high correlation with $\mutual$} both theoretically and in practice.
This proxy measure is based on the concepts of Total Variation Distance (TVD), that is defined as $\mathrm{TVD}(P, Q) = \frac{1}{2} \sum_{x} |P(x) - Q(x)|$ for two distributions $P$ and $Q$.

\begin{definition}[$\mutualproxytvd$ as a proxy for \mutual]\label{def:proxy-tvd}
Given a database $D$ and a fairness criterion of the form $\fairnessdef = \protectedset \independent \responseset \mid \admissibleset$ where $\admissibleset$ can be an empty set, the $\mutualproxytvd$ proxy for $\mutual$ is defined as:
\begin{small}
\begin{equation*}
\mutualproxytvd(\protectedset \independent \responseset | \admissibleset, D) = 2 \cdot \left( \mathrm{TVD}(\prob(\protectedset, \responseset | \admissibleset), \prob(\protectedset | \admissibleset)\prob(\responseset | \admissibleset)) \right)^2  
\end{equation*}
\end{small}
\end{definition}

We first show that $\mutualproxytvd$ is a bounded approximation of $\mutual$ and further show that it leads to a better approximation of $\mutual$ in practical scenarios.

\begin{proposition}[$\mutualproxytvd$ is bounded by $\mutual$]\label{prop:tvd-proxy}
Let $D$ be a database such that the schema of $D$ is a superset of $\protectedset \cup \responseset$. Given a fairness constraint $\protectedset \independent \responseset$, the following holds: $\alpha \cdot \mutual(\protectedset,\responseset) \leq \mutualproxytvd(\protectedset \independent \responseset, D) \leq \mutual(\protectedset,\responseset)$ where $X = \prob(\protectedset,\responseset)$, $Y = \prob(\protectedset) \prob(\responseset)$, and $\alpha = min_{z\in Z, Y(z) > 0} Y(z)$.
\end{proposition}

\ifpaper
\else
The proofs for all propositions and lemmas can be found in \Cref{sec:appendix}.
\fi 

Though the lower bound could be loose when $\alpha$ is small given a skewed dataset, in practice $\mutualproxytvd$ approximates $\mutual$ tightly. 
\ifpaper
\else
\Cref{fig:mutual-comparison-4} shows 
\fi
the values of $\mutualproxytvd$ compared to $\mutual$, $\mutualproxybayes$, and its offset version. The trend of $\mutualproxytvd$ values emulates the trend of $\mutual$ values better than the other measures. Furthermore, \Cref{fig:mutual-comparison-5} (\Cref{sec:case-studies}) shows a similar faithfulness trend of $\mutualproxytvd$ to $\mutual$ on a synthetic dataset with increasing unfairness. 

We now show that $\mutualproxytvd$ satisfies the desired properties. The first two properties can be shown directly, while computability under DP will be proven in a two-part fashion: (1) bounding the sensitivity and range of the measure in the following proposition and (2) providing an $\varepsilon$-DP algorithm for computing it in \Cref{sec:algorithms}. 

\begin{proposition}[$\mutualproxytvd$ satisfies the desired properties]~\label{prop:tvd-properties}
\begin{enumerate}
    \item $\mutualproxytvd$ satisfies the Positivity property.
    \item $\mutualproxytvd$ satisfies the Monotonicity property. 
    \item The range of $\mutualproxytvd$ is $[0,2|\constraints|]$ for a set of criteria $\constraints$.
    \item The sensitivity of $\mutualproxytvd$ is $\frac{16|\constraints|}{n}$ for a set of criteria $\constraints$ and a database of size $n$.
\end{enumerate}
\end{proposition}


\subsection{Unfairness as Data Repair Cost}\label{sec:repair}
Inspired by the field of data repair~\cite{Afrati2009,ChuIP13,BertossiKL13,LivshitsKR20,GeertsMPS13,InterFair} and previous work on inconsistency measures for integrity constraints~\cite{DBLP:journals/corr/abs-1904-03403,LivshitsKTIKR21}, including a recent work that allows for their DP computation~\cite{mohapatra2025computinginconsistencymeasuresdifferential}, we use a similar idea to define another unfairness measure.

\begin{definition}[Data Repair Unfairness]\label{def:repair}
Given a database $D$ and a fairness criterion of the form $\fairnessdef = \protectedset \independent \responseset \mid \admissibleset$ where $\admissibleset$ can be empty, the repair unfairness measure is defined as $\repair(\fairnessdef,D) = |D \symdif D_R|$, where $D_R$ is the database with the smallest number of removed and added tuples to $D$ that satisfies $\fairnessdef$ and $\symdif$ is the symmetric difference.
\end{definition}

It follows that $\repair$ satisfies the Positivity and Monotonicity properties, its range is $[0,n|\constraints|]$, and its sensitivity is $|\constraints|$. 
However, the problem of finding a minimum data repair is known to be a particularly challenging one, as we shall next elaborate.


\paratitle{Computational challenge}
The problem of data repair has been thoroughly studied with different intervention models, such as tuples deletions~\cite{LivshitsKR20,GiladDR20}, value update~\cite{RekatsinasCIR17,ChuIP13,GeertsMPS13}, and combinations of tuple additions and deletions~\cite{InterFair}. Its computational hardness is well established, even under tuple deletion alone~\cite{LivshitsKR20}. 
Although our goal is to measure the size of the repair with deletions and additions and not to find the repair itself, the two problems are computationally equivalent. 
To address this challenge, 
in \Cref{sec:algorithms}, we adapt an algorithm from previous work~\cite{InterFair} that reduces the problem of computing $\repair$ to the Max-3SAT problem, and then uses a SAT solver to obtain a repair. 
Therefore, for $\repair$ we resort to a proxy measure called $\repairsat$ that is based on prior work~\cite{InterFair}, not for privacy purposes, but for computational purposes.


\paratitle{Review of the approach from~\cite{InterFair}}
To overcome intractability, \citet{InterFair} reduce the computation of an optimal repair satisfying $\protectedset\independent \responseset \mid \admissibleset$ to weighted MaxSAT (which can then be solved by a SAT solver), where $\attrset = \{\protectedset,\responseset,\admissibleset\}$.\footnote{In practice, the database schema may contain additional attributes that do not appear in the fairness criterion. We therefore compute the repair on the projection to the criterion's attributes and, when we lift this repaired projection back to a full repair of the original schema, 
we keep all other attributes as in the original database. Since these additional attributes are unchanged, all propositions that we prove in the sequel with the reduced schema assumption carry over to the full schema.}

Define the self-join database $\jdb = \Pi_{\protectedset,\admissibleset}(D) \bowtie \Pi_{\responseset,\admissibleset}(D)$. 
Then, any minimal repair $D'$ of $D$ satisfies $D' \subseteq \jdb$. The reduction converts the tuples in $\jdb$ to a CNF formula by adding a clause $x_t$ for every $t \in D$ and a clause $\neg x_t$ for every $t \in \jdb \setminus D$. These are the `soft clauses' of the formula. It also adds the clauses $\neg x_{t_1} \lor \neg x_{t_2} \lor x_{t_3}$ for tuples of the form $t_1[p_1, y_1, a], t_2[p_2, y_2, a]$ and $t_3[p_1, y_2, a]$. Intuitively, $t_1$ and $t_2$ are the lineage of $t_3$. Thus, these are the `hard clauses' of the formula, i.e., they must be satisfied by the obtained assignment.
\citet{InterFair} uses a SAT solver to find an assignment that maximizes the additive weight of the satisfied soft clauses of the constructed CNF formula. The assignment directly corresponds to the tuples that have to be removed and added to satisfy the fairness criterion. Next, we show how we utilize and extend this approach to define $\repairsat$, which is our proxy for $\repair$.

\subsubsection{Defining $\repairsat$}\label{sec:repair-sat}
Following previous work~\cite{InterFair}, we define $\repairsat$ in terms of an assignment to the relevant CNF formula. 
We begin by defining this CNF formula that $\jdb$ is mapped to. 

\begin{definition}[CNF formula for \jdb~\cite{InterFair}]\label{def:cnf}
Given a database $D$ whose schema contains $\protectedset \cup \responseset \cup \admissibleset$, and a fairness criterion $\fairnessdef = \protectedset \independent \responseset \mid \admissibleset$, the self-join database CNF formula is defined as 
\begin{small}
\begin{equation*}
\varphi(D,\jdb) = \mathcal{H}(\jdb) \land \bigwedge_{t \in D} x_t \land \bigwedge_{t \in \jdb \setminus D} (\neg x_t),
\end{equation*}
\end{small}
where each $x_t$ represents a soft clause for a tuple $t$ and $\mathcal{H}(\jdb)$ is the set of hard clauses constructed as follows. 
Let $C(\protectedset_1, \responseset_1, \protectedset_2, \responseset_2, \admissibleset) := \jdb(\protectedset_1, \responseset_1, \admissibleset) \land \jdb(\protectedset_2, \responseset_2, \admissibleset)$. That is, $C$ is the set of all tuples formed by joining $\jdb$ with itself on the attribute $\admissibleset$, pairing tuples with the same $\admissibleset$-values and extracting the relevant attributes. 
For each tuple $t \in C$, a clause of the form $(\neg x_{t_1} \lor \neg x_{t_2} \lor x_{t_3})$ is added to $\mathcal{H}(\jdb)$, where
$
t_1 = (\protectedset_1, \responseset_1, \admissibleset), 
t_2 = (\protectedset_2, \responseset_2, \admissibleset), 
t_3 = (\protectedset_1, \responseset_2, \admissibleset)
$.
\end{definition}

We treat the soft clauses $\bigwedge_{t \in D} x_t \land \bigwedge_{t \in \jdb \setminus D} (\neg x_t)$, as a bag, meaning that soft clauses for two identical tuples (except their IDs) will both appear in $\varphi(D,\jdb)$.

A {\em feasible assignment} for a CNF formula $\varphi(D,\jdb)$ is an assignment, $\alpha$, that satisfies all the hard clauses $\mathcal{H}(\jdb)$.

We now define the repair cost of a database based on an assignment from~\Cref{def:cnf}. We will then link this notion with the notion of data repair by means of tuple deletions and additions. 

\begin{definition}[Repair cost under assignment]\label{def:repair-cost-for-assignment}
Given a database $D$, a fairness criterion $\fairnessdef$, a CNF formula $\varphi(D,\jdb)$ and $\alpha$ a feasible assignment for $\varphi$, let $D_R$ denote the set of tuples $t \in \jdb$ such that $x_t$ is assigned \texttt{True} under the assignment $\alpha$. We define the cost of repairing $D$ under the assignment $\alpha$ as $\repaircost(\varphi(D,\jdb), \alpha) = |D \symdif D_R|$. That is, the size of the set of tuples in $D$ whose corresponding variables are assigned to \texttt{False}, or are not present in $D$ and assigned to \texttt{True}.
\end{definition}

The following lemma discusses the direct translation between the number of soft clauses satisfied by an assignment to the CNF formula and the number of changes in the database required for it to satisfy the fairness criterion. 
In particular, a larger number of satisfied clauses means a smaller change in the original database. 

\begin{lemma}\label{lem:repair-cost-dominance}
Given a database $D$ and a fairness criterion, let $\jdb$ be the self-join database. Let $\varphi(D,\jdb)$ be a CNF formula defined according to \Cref{def:cnf} and let $\alpha_1$ and $\alpha_2$ be two feasible assignments for $\varphi$. Finally, let $\repaircost(\varphi(D,\jdb), \alpha_1)$ and $\repaircost(\varphi(D,\jdb), \alpha_2)$ be defined according to \Cref{def:repair-cost-for-assignment}. If $\alpha_1$ satisfies more soft clauses than $\alpha_2$, then:
$
\repaircost(\varphi(D,\jdb), \alpha_1) < \repaircost(\varphi(D,\jdb), \alpha_2)
$.
\end{lemma}

By applying \Cref{lem:repair-cost-dominance}, we can define a proxy for $\repair$, $\repairsat$, as the repair with the minimum cost based on~\Cref{def:repair-cost-for-assignment}. 

\begin{definition}[$\repairsat$ as a proxy for $\repair$]\label{def:repair-sat-proxy}
Given a database $D$, a fairness criterion $\fairnessdef$, and the CNF formula $\varphi(D,\jdb)$, the cost of an optimal repair of $D$ through $\varphi(D,\jdb)$ is $$\repairsat(\fairnessdef,D) := \min_{\alpha \models \mathcal{H}(\jdb)} \repaircost(\varphi(D,\jdb), \alpha)$$
\end{definition}

\cut{
We now establish the connection between $\repairsat$ and the assignment to the CNF formula. This will drive the algorithm design.

\begin{proposition}[Expressing $\repairsat$ in terms of \Cref{def:cnf}]\label{prop:repair-sat-cnf}
Let $D$ be a database, let $\fairnessdef$ be a fairness criterion, and let $\varphi(D,\jdb)$ be the CNF from \Cref{def:cnf}. Denote by $\mathcal{H}(\jdb)$ and $\mathcal{S}(\jdb)$ the sets of hard and soft clauses in $\varphi(D,\jdb)$, respectively. The following holds.
$$
\repairsat(\fairnessdef,D) = \left|\mathcal{S}(\jdb)\right| - \min_{\alpha \models \mathcal{H}(\jdb)}
|\{\text{soft clauses not satisfied by }\alpha\}|
$$
\end{proposition}
}

We detail the properties of $\repairsat$, showing that it satisfies the required desiderata and its bounded sensitivity relative to its range. 

\begin{proposition}[$\repairsat$ satisfies the desired properties]\label{prop:repairsat-properties}~
\begin{enumerate}
\item $\repairsat$ satisfies the Positivity property.
\item $\repairsat$ satisfies the Monotonicity property. 
\item The range of $\repairsat$ is $[0,n|\constraints|]$ for a set of criteria \constraints.
\item The sensitivity of $\repairsat$ is $2|\constraints|$ for a set of criteria \constraints.
\end{enumerate}
\end{proposition}

\subsection{Unfairness as Top Contributions}\label{sec:contribution}

Previous work has measured tuple contributions as both measuring inconsistency~\cite{LivshitsKTIKR21} and as means of explaining results~\cite{MeliouGMS11,DeutchFGS21,LivshitsK21,LivshitsBKS21}. 
We adapt this notion to define an unfairness measure that 
measures individual tuple contribution to unfairness. 
First, we define the notion of marginal difference for a fairness criterion $\fairnessdef = \protectedset \independent \responseset \mid \admissibleset$ where $\admissibleset$ can be absent, as follows:
\begin{small}
\begin{equation*}
\begin{aligned}
\margdiff(\fairnessdef, D, t)
&= \prob\!\bigl(\admissibleset=t[\admissibleset]\bigr)
   \Bigl|
   \prob\!\bigl(
      \protectedset=t[\protectedset],\,
      \responseset=t[\responseset]
      \mid
      \admissibleset=t[\admissibleset]
   \bigr)
 \\
&
   -\;
   \prob\!\bigl(
      \protectedset=t[\protectedset]
      \mid
      \admissibleset=t[\admissibleset]
   \bigr)
   \prob\!\bigl(
      \responseset=t[\responseset]
      \mid
      \admissibleset=t[\admissibleset]
   \bigr)
   \Bigr|,
\end{aligned}
\end{equation*}
\end{small}
where $\prob(\admissibleset=t[\admissibleset]) = 1$ when $\admissibleset$ is absent.

We now define the $\contribution$ measure using the $\margdiff$ notion. 

\begin{definition}[Top-$k$ Tuple Contribution Unfairness]\label{def:contribution}
Given a database $D$, a natural number $k$, and a fairness criterion of the form $\fairnessdef = \protectedset \independent \responseset \mid \admissibleset$, define top-$k$ as the set of 
as the set of $k$ tuples with the largest $\margdiff$ values. 
We assume that each tuple in the database has a unique ID, so that tuples with the same values but different IDs can both appear in the top-$k$. Then, the tuple contribution unfairness measure is defined as $\contribution(\fairnessdef,D) = \sum_{t \in top-k} \margdiff(\fairnessdef,D,t)$.
\end{definition}

While \Cref{def:contribution} resembles \Cref{def:proxy-tvd} as a cumulative sum of residual tuple differences, they diverge when databases have similar sums but different contribution distributions, e.g., a `long tail' of outliers. We demonstrate this distinction empirically in
\ifpaper
\else
\Cref{sec:tvd-contribution-comparison}
\fi.

In the following proposition, we assume that if $\protectedset \independent \responseset \mid \admissibleset \in \constraints$, all values of $\admissibleset$ occur more than once in the dataset. 
This was the case in all real-world datasets included in our experiments.

\begin{proposition}[$\contribution$ satisfies the desired properties]\label{prop:contribution-properties}
Given a natural number $k$, the following holds for $\contribution$:
\begin{enumerate}
    \item $\contribution$ satisfies the Positivity property.
    \item $\contribution$ satisfies the Monotonicity property. 
    \item The range of $\contribution$ is $[0,\min\{\frac{k}{4},2\}|\constraints|]$ for a set of criteria $\constraints$.
    \item The sensitivity of $\contribution$ is 
    \begin{small}
    $\frac{3k}{n}|\constraints|$ 
    \end{small}
    for a set of unconditional $\constraints$, and 
    \begin{small}
    $\frac{7k}{n}|\constraints|$
    \end{small}
    for a set of conditional $\constraints$. 
\end{enumerate}
\end{proposition}


\noindent\paratitle{Practical use of the measures}\label{par:practical-use-of-the-measures}
$\mutualproxytvd$, $\repairsat$, and $\contribution$ satisfy monotonicity with a lower bound of $0$ (absolute fairness) and an upper bound (\Cref{prop:tvd-properties,prop:repairsat-properties,prop:contribution-properties}), enabling detection of both fair and extremely unfair cases, demonstrated by one of our use cases in \Cref{sec:case-studies}. For nuanced assessment, users can calibrate against a likely-satisfied baseline criterion; 
e.g., in \Cref{tab:query-measures} we show the values of the three unfairness measures with the criteria $\texttt{sex} \independent \texttt{income>50K}$ and $\texttt{race} \independent \texttt{income>50K}$ in  
the first two rows and indicate a more obvious bias for income w.r.t. race.


\paratitle{Relation to other suggested properties}\label{par:continuity-and-progression}
Previous work~\cite{LivshitsKTIKR21} has also proposed the \emph{continuity} and \emph{progression} properties for inconsistency measures. 
Underlying both properties is a dynamic data repair process with an associated operation set $O$, e.g., a single tuple deletion and/or modification, and a parameter $\delta \geq 1$. 

Informally, continuity limits the rate of inconsistency change caused by a single operation from $O$. In our setting, for every two databases $D_1,D_2$, criteria $\constraints$, and operation $o_1 \in O$, there exists $o_2 \in O$ such that 
\begin{small}
$\left|\measure(\constraints, D_1) - \measure(\constraints, o_1(D_1)\right| \le \delta \left|\measure(\constraints, D_2) - \measure(\constraints, o_2(D_2)\right|$. 
\end{small}
Assuming $O$ allows the modification, deletion, or insertion of a single tuple, 
we can derive a similar property from the sensitivity analysis of the measures (\Cref{tbl:measure-summary}), which guarantees that the rate of change is bounded for neighboring datasets, where $D$ and $D' = o_1(D)$ are neighbors: $\left|\measure(\constraints, D_1) - \measure(\constraints, o_1(D_1)\right| \le \Delta_\measure$. 

Progression states that in any case where a database violates the fairness criteria, the repair process always allows for some path towards more database consistency, i.e., there is always an operation $o\in O$ such that inconsistency is reduced after applying $o$. 
As opposed to the integrity constraints considered in~\cite{LivshitsKTIKR21} which are anti-monotonic, fairness criteria do not have this property.
Therefore, if $O$ only allows for tuple deletions for example, progression does not necessarily hold. 
To see this, consider the criterion $\texttt{sex} \independent \texttt{income>50K}$ and a database where all individuals have $\texttt{sex} = M$ and $\texttt{income>50K} = 1$. This database violates the fairness criterion but any tuple deletion will not reduce the unfairness expressed via \mutualproxytvd\ and \contribution. 
For \repairsat, consider the criteria $\constraints = \{A \independent B, C \independent D\}$ and the database $D$ with the schema $(A,B,C,D)$ and tuples  
$(1, 1, a, x), (1, 0, b, y), (0, 1, b, y), (0, 0, a, x)$. 
$D$ satisfies $A \independent B$ but violates $C \independent D$. Any single tuple deletion will make the database violate $A \independent B$, necessitating at least a single increment of the \repairsat\ measure. 
\section{Private Measure Computation}\label{sec:algorithms}

We describe the algorithms that compute the unfairness measures and comply with DP, showing the third desideratum from \Cref{sec:model}. 

\subsection{Computing $\mutualproxytvd$}\label{sec:compute-tvd-proxy}
The pseudocode for computing $\mutualproxytvd$ is summarized in \Cref{alg:tvd-proxy}. 
In line \ref{line:1_1}, the algorithm initializes to zero the variable that will accumulate $\mutualproxytvd(\fairnessdef, D)$ for all $\fairnessdef$ in the set of criteria $\constraints$.
For each fairness criterion of the form $\fairnessdef = \protectedset \independent \responseset \mid \admissibleset$ in the set of criteria $\constraints$, in line \ref{line:1_2-4} the algorithm computes empirical probabilities derived from the dataset $D$: the conditional joint probability $\prob(\protectedset,\responseset \mid \admissibleset)$ and the conditional marginal probabilities $\prob(\protectedset \mid \admissibleset)$ and $\prob(\responseset \mid \admissibleset)$. In case $\admissibleset$ is not given, we assume that $\prob(\protectedset,\responseset \mid \admissibleset)$ is equal to $\prob(\protectedset,\responseset)$, $\prob(\protectedset \mid \admissibleset)$ is equal to $\prob(\protectedset)$, and $\prob(\responseset \mid \admissibleset)$ is equal to $\prob(\responseset)$.

The algorithm splits into two cases: unconditional (no $\admissibleset$), computing TVD directly in line~\ref{line:1_5}; and conditional, computing $\prob(\admissibleset=\admissiblevalue)$ for each $\admissiblevalue$ in line~\ref{line:1_6}, then summing the probability-weighted conditional TVDs in line~\ref{line:1_7}. 
Finally, the algorithm computes $\mutualproxytvd(\fairnessdef, D)$ as $2 \cdot \mathrm{TVD}^2$ (\Cref{def:proxy-tvd}) and updates the cumulative sum with this value in line \ref{line:1_8}. 
After computing \mutualproxytvd, the algorithm applies the Laplace mechanism in line \ref{line:1_9} by adding noise according to the sensitivity $\frac{16|\constraints|}{n}$ (Item 4 in \Cref{prop:tvd-properties}), and the privacy budget $\varepsilon$. The resulting value $\widetilde{\mutualproxytvd}(\constraints, D)$ is then returned. 

The complexity of \Cref{alg:tvd-proxy} is $\mathcal{O}\left(|\constraints|n\right)$. 
\ifpaper
\else
For the detailed analysis of complexity, see \Cref{sec:complexity}.
\fi

\begin{algorithm}
\caption{Compute $\mutualproxytvd$ under DP }\label{alg:tvd-proxy}
\KwIn{Database $D$; set of fairness criteria $\constraints$; privacy budget $\varepsilon$}
\KwOut{$\widetilde{\mutualproxytvd}(\constraints, D)$}
$\mutualproxytvd(\constraints, D) \gets 0$\;\label{line:1_1}
\ForEach{$\fairnessdef = (\protectedset \independent \responseset \mid \admissibleset) \in \constraints$}{
    Compute $\prob\left(\protectedset=\protectedvalue,\responseset=\responsevalue \mid \admissibleset=\admissiblevalue\right)$, $\prob\left(\protectedset=\protectedvalue \mid \admissibleset=\admissiblevalue\right)$, $\prob\left(\responseset=\responsevalue \mid \admissibleset=\admissiblevalue\right)$\; \label{line:1_2-4}
    \tcc{Unconditional criterion}
    \If{$\admissibleset$ is $\emptyset$}
    {
        $\mathrm{TVD} \gets \frac{1}{2}\sum_{\protectedvalue \in \protectedset,\responsevalue \in \responseset} \bigl|\prob\left(\protectedset=\protectedvalue,\responseset=\responsevalue\right)$ \nonumber
        $\quad - \prob\left(\protectedset=\protectedvalue\right)\prob\left(\responseset=\responsevalue\right)\bigr|$\; \label{line:1_5}
    }
    \tcc{Conditional criterion}
    \Else{
        Compute $\prob\left(\admissibleset=\admissiblevalue\right)$ for every $\admissiblevalue \in \admissibleset$\; \label{line:1_6}
        $\mathrm{TVD} \gets \sum_{\admissiblevalue \in \admissibleset} \prob(\admissibleset=\admissiblevalue) \cdot \Bigl(\frac{1}{2}\sum_{\protectedvalue \in \protectedset,\responsevalue \in \responseset} \bigl|\prob\left(\protectedset=\protectedvalue,\responseset=\responsevalue \mid \admissibleset=\admissiblevalue\right)$ \nonumber
        \\
        $\quad - \prob\left(\protectedset=\protectedvalue \mid \admissibleset=\admissiblevalue\right)\prob\left(\responseset=\responsevalue \mid \admissibleset=\admissiblevalue\right)\bigr|\Bigr)$\; \label{line:1_7}
    }
    $\mutualproxytvd(\constraints, D) \gets \mutualproxytvd(\constraints, D) + 2 \cdot \mathrm{TVD}^2$\; \label{line:1_8}
}
\Return $\mutualproxytvd(\constraints, D) + \mathrm{Lap}\left(0, \frac{16|\constraints|}{n\varepsilon}\right)$\; \label{line:1_9}
\end{algorithm}

We next show that \Cref{alg:tvd-proxy} satisfies DP with bounded error, and thus $\mutualproxytvd$ satisfies the third property in \Cref{sec:model}.

\begin{proposition}[DP and error bound of \Cref{alg:tvd-proxy}]~
\begin{enumerate}
    \item \Cref{alg:tvd-proxy} is $\varepsilon$-DP. 
    \item For a database $D$ and a set of fairness criteria $\constraints$, \Cref{alg:tvd-proxy} returns $\widetilde{\mutualproxytvd}(\constraints, D)$ such that for any $\varepsilon > 0$, it holds that {\footnotesize$$\mathbb{E}\left[ \left| \widetilde{\mutualproxytvd}(\constraints, D) - \mutualproxytvd(\constraints, D) \right| \right] = 
    \frac{16|\constraints|}{n\varepsilon}$$}
\end{enumerate}
\end{proposition}



\subsection{Computing $\repairsat$}

\Cref{alg:repairsat} combines the weighted CNF conversion from~\cite{InterFair} with a SAT solver to solve the weighted MaxSAT problem, using the noise scale determined in~\Cref{sec:repair-sat} to ensure DP.

In line \ref{line:2_1}, the algorithm initializes the variable that will accumulate $\repairsat(\fairnessdef, D)$ for all $\fairnessdef$ in the set $\constraints$. 
For each fairness criterion $\fairnessdef = \protectedset \independent \responseset \mid \admissibleset$ in the set $\constraints$, the code block \ref{block:2_2} runs the algorithm from~\cite{InterFair}. This block first initializes $\varphi$ to an empty set, and then adds soft and hard clauses as in~\cref{def:cnf}. Then, the algorithm calls a SAT solver in line \ref{line:2_3} to obtain an assignment $\alpha$ that satisfies all the hard clauses and maximizes the number of satisfied soft clauses in $\varphi$. From this assignment, the algorithm constructs the repaired database in line \ref{line:2_4}, which has all the tuples that $\alpha$ assigns \texttt{True}. In line \ref{line:2_5} it computes $\repairsat(\fairnessdef, D)$ according to \Cref{def:repair-sat-proxy}, 
and updates the cumulative sum with this value.

Finally, the algorithm applies the Laplace mechanism in line \ref{line:2_6} and adds noise according to the sensitivity $2|\constraints|$ (see Item 4 in \Cref{prop:repairsat-properties}). The noisy version of $\repairsat$ is then returned. 

The complexity of \Cref{alg:repairsat} is $\mathcal{O}\left(|\constraints|\left(n^4 + SAT\right)\right)$. 
\ifpaper
\else
For the detailed analysis of complexity, see \Cref{sec:complexity}.
\fi 

\begin{algorithm}
\caption{Compute $\repairsat$ under DP}
\label{alg:repairsat}
\KwIn{Database $D$; set of fairness criteria $\constraints$; privacy budget $\varepsilon$}
\KwOut{$\widetilde{\repairsat}(\constraints, D)$}

$\repairsat(\constraints, D) \gets 0$\; \label{line:2_1}

\ForEach{$\fairnessdef=(\protectedset \independent \responseset \mid \admissibleset) \in \constraints$}{
    \tcc{Algorithm from~\cite{InterFair}}
    \colorbox{gray!15}{\parbox{\dimexpr\linewidth-9\fboxsep\relax}{%
        $\jdb(\protectedset_1,\responseset_2,\admissibleset) 
        \gets D(\protectedset_1,\responseset_1,\admissibleset)
        \bowtie D(\protectedset_2,\responseset_2,\admissibleset)$\; \label{l:join} 
        $\varphi \gets \emptyset$\;
        \ForEach{$t \in \jdb$}{
            \If{$t \in D$}{Add the soft clause $x_t$ to $\varphi$\;}
            \If{$t \notin D$}{Add the soft clause $\neg x_t$ to $\varphi$\;}
        }
        $C(\protectedset_1,\responseset_1,\protectedset_2,\responseset_2,\admissibleset) 
        \gets \jdb(\protectedset_1,\responseset_1,\admissibleset)
        \land \jdb(\protectedset_2,\responseset_2,\admissibleset)$\; \label{l:compute-costs} 
        \ForEach{$t \in C$}{
            $t_1 \gets t(\protectedset_1,\responseset_1,\admissibleset)$, $t_2 \gets t(\protectedset_2,\responseset_2,\admissibleset)$, $t_3 \gets t(\protectedset_1,\responseset_2,\admissibleset)$\;
            Add the hard clause $(\neg x_{t_1} \lor \neg x_{t_2} \lor x_{t_3})$ to $\varphi$\; \label{block:2_2}
        }
    }} 

    $\alpha \gets Solver(\varphi)$\; \label{line:2_3}
    $D_R \gets \{t \mid \alpha(x_t) = \texttt{True}\}$ \label{line:2_4}
    $\repairsat(\constraints, D) \gets \repairsat(\constraints, D) + |D \symdif D_R|$\; \label{line:2_5}
}

\Return $\repairsat(\constraints, D) + \mathrm{Lap}\left(0, \frac{2|\constraints|}{\varepsilon}\right)$\; \label{line:2_6}
\end{algorithm}


\begin{proposition}[DP and error bound on \Cref{alg:repairsat}]\label{prop:repairsat-alg-properties}~
\begin{enumerate}
    \item \Cref{alg:repairsat} is $\varepsilon$-DP. 
    \item For a database $D$ and a set of fairness criteria $\constraints$, \Cref{alg:repairsat} returns $\widetilde{\repairsat}(\constraints, D)$ such that for any $\varepsilon > 0$, it holds that 
    {\footnotesize$\mathbb{E}\left[ \left| \widetilde{\repairsat}(\constraints, D) - \repairsat(\constraints, D) \right| \right] = \frac{2|\constraints|}{\varepsilon}$}.
\end{enumerate}
\end{proposition}

\subsection{Computing $\contribution$}
We describe the pseudocode in \Cref{alg:topk-contribution} for computing $\contribution$. 
In line \ref{line:3_1}, the algorithm initializes to zero the variable that will accumulate $\contribution(\fairnessdef, D)$ for all $\fairnessdef$ in the set of criteria $\constraints$. 
For each fairness criterion of the form $\fairnessdef = \protectedset \independent \responseset \mid \admissibleset$ in the set of criteria $\constraints$, in line \ref{line:3_2-4} the algorithm computes empirical probabilities derived from the dataset $D$: the conditional joint probability $\prob(\protectedset,\responseset \mid \admissibleset)$ and the conditional marginal probabilities $\prob(\protectedset \mid \admissibleset)$ and $\prob(\responseset \mid \admissibleset)$. 
In case $\admissibleset$ is not given, we assume that $\prob(\protectedset,\responseset \mid \admissibleset)$ is equal to $\prob(\protectedset,\responseset)$, $\prob(\protectedset \mid \admissibleset)$ is equal to $\prob(\protectedset)$, and $\prob(\responseset \mid \admissibleset)$ is equal to $\prob(\responseset)$. 
Then, the algorithm iterates over every tuple $t$ in $D$ and in line \ref{line:3_5} it computes the marginal difference $\margdiff(\fairnessdef, D, t)$ (defined in \Cref{sec:contribution}). 
The resulting value quantifies the contribution of each tuple $t$ to the deviation of the observed joint probability from the joint probability in the case when the independence defined by $\fairnessdef$ would hold.  
After computing marginal differences for all tuples in the dataset, the algorithm sorts them in descending order in line \ref{line:3_6}. In line \ref{line:3_7} it adds the sum of the $k$ largest marginal differences to the global value $\contribution(\constraints, D)$. 
Intuitively, for each criterion in $\constraints$, it adds the violation influence of the $k$ most influential tuples to the total score. 
The algorithm performs all the previous stages for every $\fairnessdef$ in $\constraints$. After iterating over all of them, 
the algorithm adds Laplace noise in line \ref{line:3_8} according to the sensitivity 
{$\frac{7k}{n}|\constraints|$} 
\normalsize{if there is a conditional $\fairnessdef \in \constraints$, and} 
{$\frac{3k}{n}|\constraints|$} \normalsize{otherwise} (Item 4 in \Cref{prop:contribution-properties}). 
This noisy value is returned as the DP version of \contribution.

The complexity of \Cref{alg:topk-contribution} is $\mathcal{O}\left(|\constraints|n \log n\right)$. 
\ifpaper
\else
For the detailed analysis of complexity, see \Cref{sec:complexity}.
\fi 

\begin{algorithm}
\caption{Compute $\contribution$ under DP}\label{alg:topk-contribution}
\KwIn{Database $D$; set of fairness criteria $\constraints$; accuracy parameter $k \in \mathbb{N}$; privacy budget $\varepsilon$}
\KwOut{$\widetilde{\contribution}(\constraints, D)$}

$\contribution(\constraints, D) \gets 0$\; \label{line:3_1}

\ForEach{$\fairnessdef=(\protectedset \independent \responseset \mid \admissibleset) \in \constraints$}{ 
  Compute $\prob[\protectedset=\protectedvalue,\responseset=\responsevalue \mid \admissibleset=\admissiblevalue]$, $\prob[\protectedset=\protectedvalue \mid \admissibleset=\admissiblevalue]$, $\prob[\responseset=\responsevalue \mid \admissibleset=\admissiblevalue]$\; \label{line:3_2-4}
  
  \ForEach{$t \in D$}{
    $\margdiff(\fairnessdef,D,t) \gets 
      \prob[\admissibleset = \admissiblevalue] \Bigl|
        \prob[\protectedset=\protectedvalue,\responseset=\responsevalue \mid \admissibleset=\admissiblevalue]
        - \prob[\protectedset=\protectedvalue \mid \admissibleset=\admissiblevalue]\,
          \prob[\responseset=\responsevalue \mid \admissibleset=\admissiblevalue]
      \Bigr|$\; \label{line:3_5}
  }
  
  Sort $\margdiff(\fairnessdef,D,t)$ for all $t \in D$ in descending order\;\label{line:3_6}
  $\contribution(\constraints, D) \gets \contribution(\constraints, D) + \sum_{t \in \text{top-}k} \margdiff(\fairnessdef,D,t)$\; \label{line:3_7}
}

\If{$\exists \fairnessdef = (\protectedset, \responseset \mid \admissibleset) \in \constraints$ such that $\admissibleset$ is given}
{
    $\Delta \gets \frac{7k}{n}$\;
}
\Else{
    $\Delta \gets \frac{3k}{n}$\;
}
\Return $\contribution(\constraints, D) + \mathrm{Lap}\left(0, \frac{|\constraints| \cdot \Delta}{\varepsilon}\right)$\; \label{line:3_8}
\end{algorithm}

\begin{proposition}[DP and error bound of \Cref{alg:topk-contribution}]~
\begin{enumerate}
    \item \Cref{alg:topk-contribution} is $\varepsilon$-DP. 
    \item For a database $D$, \Cref{alg:topk-contribution} returns $\widetilde{\contribution}(\constraints, D)$ such that for any $\varepsilon > 0$, it holds that 
    \footnotesize{$${\mathbb{E}\left[ \left| \widetilde{\contribution}(\constraints, D) - \contribution(\constraints, D) \right| \right] = \frac{3k|\constraints|}{n\varepsilon}}$$} 
    \normalsize{for a set of unconditional fairness criteria $\constraints$, and } 
    \footnotesize{$${\mathbb{E}\left[ \left| \widetilde{\contribution}(\constraints, D) - \contribution(\constraints, D) \right| \right] = \frac{7k|\constraints|}{n\varepsilon}}$$} 
    \normalsize{for a set of conditional fairness criteria.}
\end{enumerate}
\end{proposition}



\section{Experiments}\label{sec:experiments}
We have experimentally studied our measures and the algorithms computing them by answering the following questions: (1) What is the usefulness of the measures in different scenarios and what is their ability to detect different levels of unfairness? (2) How do our algorithms scale? (3) What is the effect of the privacy budget on measure accuracy? (4) How do the measures perform in terms of faithfulness and under parameter variation? 


\paratitle{Summary of our findings}\label{par:summary-of-findings}
Our experiments confirm that the proposed unfairness measures are practical for real-world fairness assessment, serving both as pre-query trust indicators and as alternatives to model-based estimation. The measures detect unfairness correctly and proportionately (\Cref{fig:experiment-7}), all three increase monotonically with the Demographic Parity gap, with \repairsat increasing near-linearly and \mutualproxytvd, \contribution\ increasing approximately logarithmically. Their relative ordering agrees with the observed query output disparities (\Cref{sec:case-studies}), and \mutualproxytvd\ closely tracks \mutual\ (\Cref{fig:mutual-comparison-5}). All algorithms scale steadily with tuple count and fairness criteria (\Cref{fig:experiments-1-2}), where \Cref{alg:topk-contribution} is the fastest, followed by \Cref{alg:tvd-proxy}, both completing in seconds on large datasets, while \Cref{alg:repairsat} is $10^2-10^3$ times slower due to its weighted MaxSAT solver, yet remains valuable for capturing the most nuanced unfairness notion.

Additional experiments  
\ifpaper
\else
detailed in~\Cref{sec:drill-down} 
\fi
show that (i) $\mutualproxytvd$ is a good proxy for $\mutual$, (ii) $\contribution$ increases monotonically with $k$ up to some value, indicating that a dominant subset of tuples drives unfairness, and (iii) our heuristic for \repairsat\ improves runtime by over 50\% while preserving relative measure values.

\subsection{Experimental Setup}\label{sec:setup}
All algorithms and experiments were implemented in Python using standard packages such as \texttt{networkx}, \texttt{numpy}, \texttt{pandas}, 
and \texttt{z3-solver}. 
The experiments ran on a machine with \texttt{Intel Core i5-12400F}. 
For each data point in every experiment, we performed $10$ independent repetitions, unless stated otherwise, and averaged the result. If
one of these repetitions took longer than 24 hours to complete, we
did not include this data point in the graph. 
For experiments with fixed data size, we randomly sampled up to $100$K tuples from each dataset, or used the entire dataset if it is smaller, ensuring that the assumption of multiplicity in the $\admissibleset$ values holds (\Cref{prop:contribution-properties}).

\paratitle{Datasets} We used the following five datasets in our experiments.
\begin{itemize}[wide, labelwidth=!, labelindent=0pt]
    \item {\bf Adult}: A dataset downloaded from~\footnote{\url{https://archive.ics.uci.edu/dataset/2/adult}}, containing 15 attributes and 48,843 tuples. It contains personal data about individuals, including their sex, race, marital status, whether their income is over 50K, etc. Some attributes were renamed for clarity (e.g., the binary attribute \texttt{income} was renamed to \texttt{income>50K}).
    \item {\bf IPUMS-CPS}: A population survey dataset published by the U.S. Census Bureau~\cite{flood2021ipums} and downloaded from ~\footnote{\url{https://cps.ipums.org/cps}}, containing 8 attributes and 1,048,576 tuples, with data from 2011 till 2019. The dataset, like the \texttt{Adult} dataset, contains personal data about individuals.
    \item {\bf Stackoverflow}~\footnote{\url{https://kaggle.com/datasets/berkayalan/stack-overflow-annual-developer-survey-2024}}: A dataset containing 114 attributes and 65,438 tuples. The dataset contains the results of the Stackoverflow developers survey for 2024, and includes languages they use, where they learned to code, information on the companies they work at, etc.
    \item {\bf Compas}~\footnote{\url{https://github.com/propublica/compas-analysis/blob/master/compas-scores.csv}}: A dataset containing 47 attributes and 11,757 tuples. It contains the criminal history, jail and prison time, demographics and Compas risk scores for defendants from Broward County from 2013 and 2014. The data is usually used to predict whether a person will re-offend within the next 2 years.
    \item {\bf Healthcare}~\footnote{\url{https://github.com/stefan-grafberger/mlinspect/tree/master/example_pipelines}}: A dataset containing 11 attributes and 1000 tuples. It contains personal data of patients, such as the income, number of complications, county, race, age. 
\end{itemize}
All datasets were preprocessed so that negative numerical values were replaced by zeros since their domains are nonnegative, categorical values were encoded, with missing or undefined values being treated as a separate category. 
For the \texttt{IPUMS-CPS} dataset, the attribute \texttt{AGE} was discretized as 10 years per range, e.g., $[0,10]$ is considered a single value. We only included tuples with attribute \texttt{INCTOT} (total income) smaller than $200$K to avoid outliers, as in~\cite{DBLP:journals/pvldb/TaoGMR22}.

\begin{table}
\centering
\caption{Fairness criteria per dataset.}
\label{tab:fairness-criteria}
\begin{footnotesize}
\begin{tabular}{|c|l|}
\hline
\textbf{Dataset} & \textbf{Fairness criteria} \\ \hline
Adult &
\begin{tabular}[l]{@{}l@{}}
(1) \texttt{education-num} $\independent$ \texttt{income>50K} \\
(2) \texttt{sex} $\independent$ \texttt{income>50K} \\
(3) \texttt{race} $\independent$ \texttt{income>50K} \\
(4) \texttt{sex} $\independent$ \texttt{income>50K} $\mid$ \texttt{hours-per-week}
\end{tabular} \\ \hline
IPUMS-CPS & \begin{tabular}[l]{@{}l@{}}
(1) \texttt{HEALTH} $\independent$ \texttt{INCTOT} $\mid$ \texttt{EDUC} \\
(2) \texttt{HEALTH} $\independent$ \texttt{OCC} $\mid$ \texttt{EDUC} \\
(3) \texttt{HEALTH} $\independent$ \texttt{MARST} $\mid$ \texttt{AGE} \\
(4) \texttt{HEALTH} $\independent$ \texttt{INCTOT} $\mid$ \texttt{AGE}
\end{tabular} \\ \hline
Stackoverflow & \begin{tabular}[l]{@{}l@{}}
(1) \texttt{Country} $\independent$ \texttt{RemoteWork} $\mid$ \texttt{Employment} \\
(2) \texttt{Age} $\independent$ \texttt{PurchaseInfluence} $\mid$ \texttt{OrgSize} \\
(3) \texttt{Country} $\independent$ \texttt{MainBranch} $\mid$ \texttt{YearsCodePro} \\
(4) \texttt{Age} $\independent$ \texttt{MainBranch} $\mid$ \texttt{EdLevel}
\end{tabular} \\ \hline
Compas & \begin{tabular}[l]{@{}l@{}}
(1) \texttt{race} $\independent$ \texttt{is\_recid} $\mid$ \texttt{age\_cat} \\
(2) \texttt{sex} $\independent$ \texttt{is\_recid} $\mid$ \texttt{priors\_count} \\
(3) \texttt{race} $\independent$ \texttt{decile\_score} $\mid$ \texttt{c\_charge\_degree} \\
(4) \texttt{sex} $\independent$ \texttt{v\_decile\_score} $\mid$ \texttt{age\_cat}
\end{tabular} \\ \hline
Healthcare & \begin{tabular}[l]{@{}l@{}}
(1) \texttt{race} $\independent$ \texttt{complications} $\mid$ \texttt{age\_group} \\
(2) \texttt{smoker} $\independent$ \texttt{complications} $\mid$ \texttt{age\_group} \\
(3) \texttt{race} $\independent$ \texttt{income} $\mid$ \texttt{county} \\
(4) \texttt{smoker} $\independent$ \texttt{income} $\mid$ \texttt{num\_children}
\end{tabular} \\ \hline
\end{tabular}
\end{footnotesize}
\end{table}

\cut{
\paragraph{Preprocessing the data}
All datasets were preprocessed as follows: tuples with missing values in the relevant columns were dropped, negative numerical values were replaced by zeros since their domain is nonnegative, categorical values were encoded. Specifically for the \texttt{IPUMS-CPS} dataset, the  attribute \texttt{AGE} was discretized as 10 years per range, e.g., [0,10] is considered a single value. Also, we only include tuples with attribute \texttt{INCTOT} (the total income) smaller than 200k as a domain bound.
For each data point in every experiment, we performed $10$ independent repetitions and averaged the result. If one of these repetitions took longer than one hour to complete \ag{two hours?}, we skipped this data point and did not include it in the graph. 
During the experiments that do not vary the number of tuples, we randomly sampled 
up to 100,000 tuples from each dataset, or used the dataset in full if it is smaller.
For $\contribution$, during the experiments that do not change $k$, we used $k=500$. 
}

\paratitle{Fairness criteria}
Following our formalization in \Cref{sec:fairness} of fairness criteria as conditional independence statements, we list the criteria per dataset in~\Cref{tab:fairness-criteria}, numbered $1$--$4$ in the figures. The criteria reflect natural desiderata for protected attributes (e.g., sex) to be independent from the outcome attribute (e.g., income).

\paratitle{Algorithm variations and optimizations}\label{par:variations}
Since there are no baselines for unfairness measures, we compared \Cref{alg:tvd-proxy,alg:repairsat,alg:topk-contribution} with their several variations and used a heuristic for \Cref{alg:repairsat}:
\begin{itemize}[wide, labelwidth=!, labelindent=0pt]
    \item {\bf Non-private versions:} In \Cref{fig:experiment-3}, we compared \Cref{alg:tvd-proxy,alg:repairsat,alg:topk-contribution} with varying privacy budgets from $0.1$ to $10$ with their non-private counterparts, i.e., setting $\varepsilon = \infty$. For the rest of the experiments, we use a privacy budget of $\varepsilon = 1$, unless otherwise specified.
    \item {\bf $k$ in~\Cref{alg:topk-contribution}:} We fixed $k=500$ in all experiments.
    \item {\bf Heuristic for~\Cref{alg:repairsat}:} Due to scalability issues of the SAT solver in line~\ref{line:2_3} of~\Cref{alg:repairsat}, we employed a heuristic that partitions the database to consecutive chunks of $100$ tuples, then runs \Cref{alg:repairsat} separately for each chunk and each criterion, and sums up the results to yield an estimate of the $\repairsat$ value. Assuming that the tuple order is preserved in neighboring databases, as is standard for bounded DP~\cite{DBLP:journals/ftdb/NearH21}, the sensitivity of this heuristic remains identical to $\sens_{\repairsat}$, as proven in
    \ifpaper
\else
    ~\Cref{prop:repairsatchunked-sensitivity}.
    \fi
\end{itemize}
\paratitle{Relative $L1$ error}
In some of the experiments we measured relative $L1$ error of the algorithms. For a given $\varepsilon$, we define the relative error as $\frac{|X-Y|}{\max(Y, e^{-100})}$, where $X$ is the output of the algorithm with privacy budget $\varepsilon$, and $Y$ is the output of the algorithm without privacy. The use of $e^{-100}$ in the denominator is meant to avoid division by $0$.

\cut{
\paratitle{Examined downstream model}
As part of our study, we measured the connection between the unfairness measure values and the unfairness level of a downstream ML model, with the same fairness criterion. As our model we chose a fully connected multilayer perceptron that consists of a single hidden layer of 32 neurons, each followed by a ReLU activation to introduce non-linearity, and a final linear layer with one output unit. We used DP-SGD~\cite{DBLP:conf/ccs/AbadiCGMMT016} to train this model on the private data.
}

\subsection{Use Cases for the Measures}\label{sec:case-studies}

We next show three case studies for our unfairness measures in the context of (1) data exploration (2) ML models, and (3) increasing data unfairness. 

\definecolor{heatlight}{RGB}{255,235,235}
\definecolor{heatmid}{RGB}{255,200,200}
\definecolor{heatdark}{RGB}{255,160,160}
\begin{table}[t]
\centering
\caption{Average unfairness-measure values for the fairness criteria on the \texttt{Adult} and \texttt{Compas} datasets. Darker red indicate higher measure value, i.e., larger unfairness.}
\label{tab:query-measures}
\small
\setlength{\tabcolsep}{2pt}
\begin{tabular}{llccc}
\toprule
\textbf{Dataset} & \textbf{Criterion corresponding to (Query)} & \textbf{\mutualproxytvd} & \textbf{\repairsat} & \textbf{\contribution} \\
\midrule
\multirow{5}{*}{\texttt{Compas}} 
& $\texttt{race} \independent \texttt{decile\_sc} \mid \texttt{age\_cat}$ ($q_1$)
& \cellcolor{heatmid} 0.035
& \cellcolor{heatmid} 2517.2
& \cellcolor{heatmid} 0.263 \\
\cmidrule{2-5}
& $\texttt{age\_cat} \independent \texttt{decile\_sc} \mid \texttt{race}$ ($q_2$)
& \cellcolor{heatdark} 0.049
& \cellcolor{heatdark} 2704.8
& \cellcolor{heatdark} 0.310 \\
\cmidrule{2-5}
& $\texttt{race} \independent \texttt{decile\_sc} \mid \texttt{c\_charge\_d.}$ ($q_3$)
& \cellcolor{heatmid} 0.039
& \cellcolor{heatmid} 2662.8
& \cellcolor{heatmid} 0.282 \\
\cmidrule{2-5}
& $\texttt{c\_charge\_d.} \independent \texttt{decile\_sc} \mid \texttt{race}$ ($q_4$)
& \cellcolor{heatlight} 0.015
& \cellcolor{heatlight} 1897.7
& \cellcolor{heatlight} 0.169 \\
\midrule
\multirow{2}{*}{\texttt{Adult}} 
& $\texttt{sex} \independent \texttt{income>50K}$ ($q_5$)
& \cellcolor{heatdark} 0.015
& \cellcolor{heatdark} 4573.8
& \cellcolor{heatdark} 0.175 \\
\cmidrule{2-5}
& $\texttt{race} \independent \texttt{income>50K}$ ($q_6$)
& \cellcolor{heatlight} 0.002
& \cellcolor{heatlight} 1227.8
& \cellcolor{heatlight} 0.053 \\
\bottomrule
\end{tabular}
\end{table}

\paratitle{Assessing dataset biases in data exploration}

Suppose we are interested in understanding how estimated risk to reoffend varies across age, race, and charge-degree groups using the \texttt{Compas} dataset. 
Queries $q_1$--$q_4$ compute average or median \texttt{decile\_score} across such groups, where \texttt{decile\_score} denotes the estimated risk to reoffend. 
For example, $q_1$ computes, for each age group, the average decile score aggregated over race groups:
\begin{footnotesize}
\begin{lstlisting}
SELECT age_cat, AVG(avg_d) AS avg_d
FROM (SELECT race, age_cat, AVG(decile_score) AS avg_d
      FROM Compas GROUP BY race, age_cat) t
GROUP BY age_cat;
\end{lstlisting}
\end{footnotesize}

Each query corresponds to a criterion of the form $S \independent Y \mid A$, where $Y=\texttt{decile\_score}$, $S$ is the attribute over which disparities are measured, and $A$ is the conditioning attribute. 
The corresponding measures in \Cref{tab:query-measures} indicate disparities in the data, and the non-private results reflect them. 
For example, in $q_1$, the average \texttt{decile\_score} for individuals aged 25--45 is 3.788 without privacy, compared to 2.662 for individuals older than 45.

With $\varepsilon=1$, however, the added noise severely distorts the query results, making them misleading. 
For example, in $q_4$, for $\texttt{c\_charge\_degree}=O$, the median \texttt{decile\_score} is -6.324 with $\varepsilon=1$, compared to values close to 2 without privacy. 
This trend, observed by prior work~\cite{DBLP:journals/pvldb/TaoGMR22}, shows that differential-privacy noise can distort query results and even flip orderings, especially for small groups. 
Thus, the unfairness measures computed directly on the data with budget $\epsilon=1$, shown in \Cref{tab:query-measures}, are crucial for interpreting such noisy query results.

Now suppose we are interested in understanding how income varies across different sexes and races in the \texttt{Adult} dataset. 
Queries $q_5$ and $q_6$ compute the proportion of individuals with income above \$50K across groups defined by \texttt{sex} and \texttt{race}, respectively (see
\ifpaper
\else
~\Cref{tab:eda-queries}
\fi
for their SQL descriptions). 
Before running the queries, we computed the unfairness measures for the corresponding criteria 
$\texttt{sex} \independent \texttt{income>50K}$ and 
$\texttt{race} \independent \texttt{income>50K}$. 
The last two rows of \Cref{tab:query-measures} indicate stronger income bias with respect to \texttt{sex} than \texttt{race}. 
This remains visible even with $\varepsilon=1$: in $q_5$, the proportion of males with income above \$50K is 0.304, compared to 0.110 for females, whereas in $q_6$ the largest racial disparity is smaller, between 0.118 for Amer-Indian-Eskimo and 0.261 for Asian-Pac-Islander. 
Thus, because the \texttt{Adult} dataset is much larger, the added privacy noise weakens but does not substantially obscure these trends.

Overall, the measures serve as an early warning mechanism and identify queries whose results may reflect dataset bias, as well as help interpret query results distorted by privacy noise.

\begin{table}[t]
\centering
\caption{Fairness values (Demographic Parity gaps) of privately-trained variations of Random Forest (RF) and Neural Network (NN) on the \texttt{Adult} dataset.}
\label{tab:models-unfairness}
\setlength{\tabcolsep}{4pt}
\begin{tabular}{lcccc}
\toprule
\textbf{Model} 
& \multicolumn{2}{c}{$\texttt{sex} \independent \texttt{income>50K}$} 
& \multicolumn{2}{c}{$\texttt{race} \independent \texttt{sex}$} \\
& \textbf{Accuracy} & \textbf{DP gap} & \textbf{Accuracy} & \textbf{DP gap} \\
\midrule
RF, 10 trees  & 78.91\% & 0.058 & 78.22\% & 0.213 \\
RF, 50 trees  & 78.24\% & 0.037 & 78.50\% & 0.235 \\
RF, 100 trees & 78.41\% & 0.044 & 77.54\% & 0.234 \\
RF, 200 trees & 77.91\% & 0.032 & 76.06\% & 0.188 \\
NN, 1 hidden layer  & 84.16\% & 0.197 & 80.37\% & 0.216 \\
NN, 2 hidden layers & 83.82\% & 0.201 & 80.03\% & 0.236 \\
NN, 3 hidden layers & 82.40\% & 0.225 & 79.33\% & 0.231 \\
\bottomrule
\end{tabular}
\end{table}

\paratitle{A complementary view to Machine Learning models}
To examine how our measures complement ML models in quantifying unfairness, we trained a variety of models with two different classes and hyperparameters to predict \texttt{income>50K} and \texttt{sex} based on the other attributes in the \texttt{Adult} dataset. 
We measured the Demographic Parity gap of the predictions of \texttt{income>50K} with respect to \texttt{sex}, and the gap in the predictions of \texttt{sex} with respect to \texttt{race}. The latter serves as a baseline (which did not appear in~\Cref{ex:intro}), as in real-world settings and intuitively these attributes are expected to be independent, and thus any observed dependence indicates bias.
The models, accuracy rates and gaps are detailed in \Cref{tab:models-unfairness}. We observe that the models consistently estimate a stronger dependence between $\texttt{sex}$ and $\texttt{race}$ than between $\texttt{income>50K}$ and $\texttt{sex}$. 
Moreover, despite all models achieving similar accuracy (76.06\%-84.16\%), the Demographic Parity gaps vary substantially both across model classes (0.032-0.058 for Random Forest, 0.197-0.255 for Neural Network) and within a single class (for Random Forest, 0.032-0.058 for \texttt{sex} and \texttt{income>50K}, and 0.188-0.235 for \texttt{race} and \texttt{sex}), making it difficult to reliably compare the two disparities.
Conversely, evaluating the unfairness measures directly on the \texttt{Adult} dataset for the corresponding criteria, $\texttt{sex} \independent \texttt{income>50K}$ and $\texttt{race} \independent \texttt{sex}$ generates the intuitive results consistent with prior work~\cite{DBLP:conf/kdd/ThanhRT11,DBLP:conf/eurosp/TramerAGHHHJL17} that indicate a strong dependence between sex and income in the \texttt{Adult} dataset: $0.014$ and $0.002$ for \mutualproxytvd, $4424$ and $1432$ for $\repairsat$, and $0.164$ and $0.076$ for \contribution, respectively. 

Finally, the models are limited to assessing unfairness only for the attribute they were trained for. For example, if we compute the Demographic Parity gap for \texttt{race} given \texttt{sex} using a model trained to predict \texttt{income>50K}, we obtain a value of 0.044 for a Random Forest with 200 trees. In contrast, when the same model is trained to predict \texttt{sex} (i.e., aligned with the criterion), the corresponding Demographic Parity gap is 0.188. In contrast, the unfairness measures can run over arbitrary fairness criteria and thus form a complementary alternative to model-based fairness evaluation.

\begin{figure}[h]
    \centering
    \begin{subfigure}[t]{0.25\textwidth}
    \vspace{2pt}
        \centering
        \includegraphics[width=\linewidth]{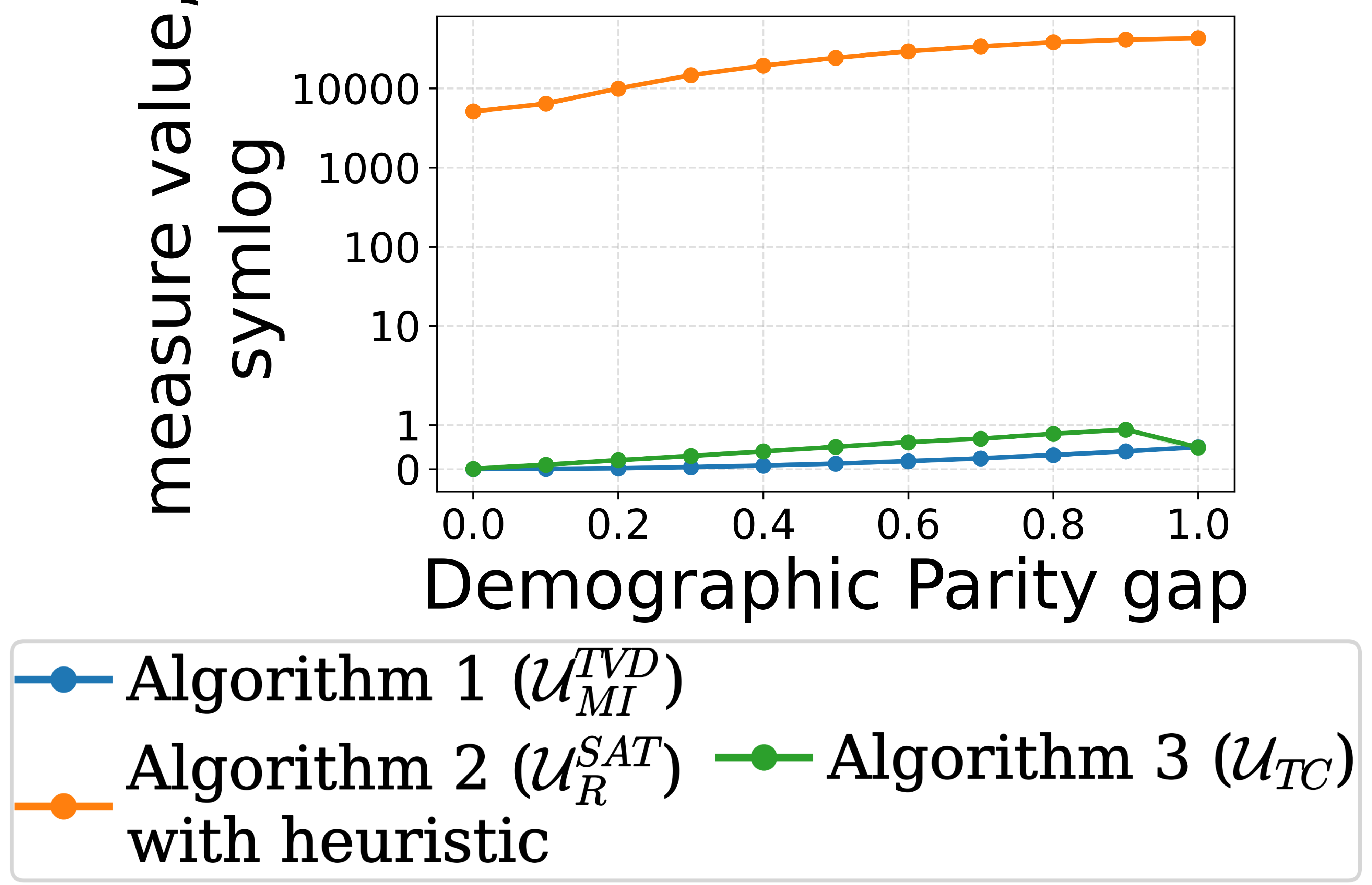}
        \caption{Measure values.}
        \label{fig:experiment-7}
    \end{subfigure}
    \hfill
    \begin{subfigure}[t]{0.22\textwidth}
    \vspace{0pt}
        \centering
        \includegraphics[width=\linewidth]{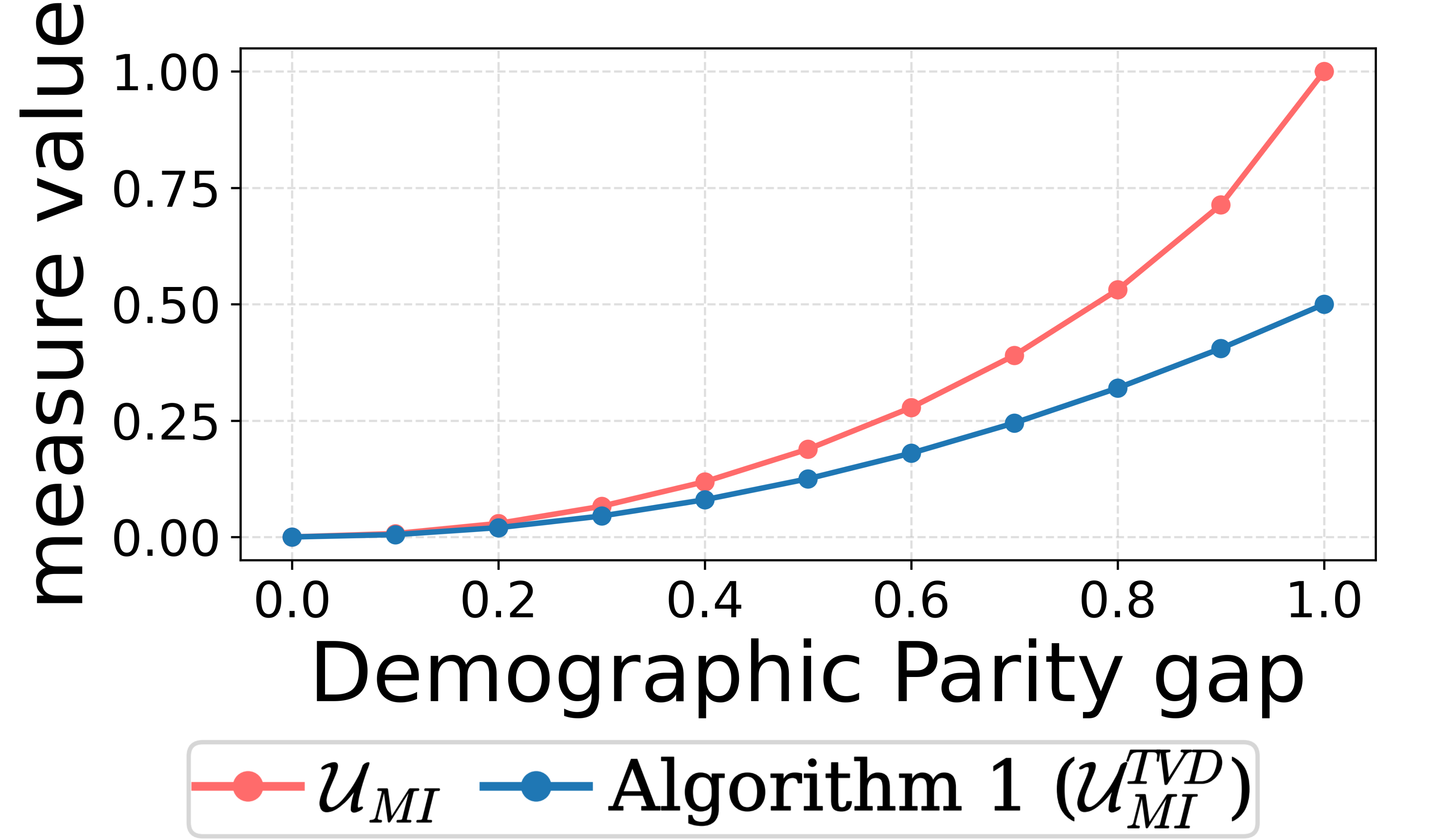}
        \vspace{8pt}
        \caption{$\mutualproxytvd$ vs. $\mutual$.}
        \label{fig:mutual-comparison-5}
    \end{subfigure}
    \caption{Behavior of unfairness measures as the Demographic Parity in the dataset gradually increases.}
    \label{fig:experiment-7-10}
\end{figure}

\paratitle{Detecting varying unfairness}
We studied how our unfairness measures react to gradually increasing unfairness in a dataset. We constructed a binary synthetic dataset with attributes \texttt{sex} and \texttt{income>50K} containing 100,000 tuples, with an equal number of male and female records. Initially, the dataset was completely fair (i.e., 50\% of each sex group had \texttt{income>50K }$ = 1$). We then progressively flipped an increasing fraction of male records to have income above 50K and female records to have income below 50K, thereby monotonically increasing the demographic parity gap
$|\prob(\texttt{income}>50K \mid \texttt{sex}=M) - \prob(\texttt{income}>50K \mid \texttt{sex}=F)|$.
In~\Cref{fig:experiment-7}, we plot the values of the measures as a function of the demographic parity gap. $\repairsat$ (as computed by~\Cref{alg:repairsat} with our heuristic) exhibits almost linear growth as the demographic parity gap increases, while $\mutualproxytvd$ (as computed by~\Cref{alg:tvd-proxy}) and $\contribution$ (as computed by~\Cref{alg:topk-contribution}) exhibit growth that appears logarithmic. $\contribution$ also decreases slightly toward the right-hand side of the graph, since when the dataset becomes fully polarized, $\margdiff$ for all tuples become similar (i.e., there are no outliers), and so the top-$k$ sum diminishes. Additionally, \Cref{fig:mutual-comparison-5} shows that on the same synthetic dataset, $\mutualproxytvd$ closely follows the trend observed for $\mutual$. 
Overall, all measures correctly detect small changes in unfairness. 

\subsection{Scalability}\label{sec:scalability}
\Cref{fig:experiments-1-2} depicts the effect of data size and the number of fairness criteria on the runtime of the algorithms.
In \Cref{fig:experiment-1}, we plot the runtime of the algorithms for each dataset as a function of the number of tuples.
Across all datasets, the computation of \Cref{alg:repairsat} is consistently the slowest. 
The runtime increase of \Cref{alg:repairsat} stems from its dependence on the weighted MaxSAT solver used to find the minimal repair. However, this step also enables it to capture a more nuanced notion of unfairness than the other measures.

\Cref{alg:topk-contribution} is frequently the fastest, with an exception being the \texttt{IPUMS-CPS} dataset, on which the runtime is higher and mostly insensitive to the number of tuples. This is attributed to the fact that, as formulated in 
\ifpaper
\else
\Cref{contribution-complexity}
\fi, \Cref{alg:topk-contribution} uses an empirical contingency table over the attribute sets, whose cardinalities are higher for \texttt{IPUMS-CPS} than for the other datasets. 
Specifically, it is sometimes faster than \Cref{alg:tvd-proxy} since the latter has to perform additional manipulations to compute the conditional TVD. 
Finally, \Cref{alg:tvd-proxy} is on average faster than \Cref{alg:repairsat} but slower than \Cref{alg:topk-contribution}.
The runtimes of all the algorithms increase steadily with the number of tuples. Still, runtimes for smaller numbers of tuples are sometimes higher and more volatile because they are more sensitive to randomness in the sampled data.

In \Cref{fig:experiment-2}, we show the runtime of the algorithms for each dataset as a function of the number of fairness criteria. For each dataset, we begin with only the first criterion in \Cref{tab:fairness-criteria} and then progressively add more, such that the final measurement includes all criteria for that dataset. We can see that, for all the algorithms, increasing the number of criteria inflates the runtime more gradually than increasing the number of tuples.


Still, the runtime for the \texttt{Healthcare} dataset varies significantly between two and three criteria when computing $\mutualproxytvd$ and $\contribution$ due to the fact that the third criterion groups tuples by \texttt{county}, which has much higher cardinality than \texttt{age\_group}. This increases the size of the empirical distributions computed by \Cref{alg:tvd-proxy,alg:topk-contribution}, leading to higher cost. In contrast, \Cref{alg:repairsat} is less affected, since its runtime is dominated by the MaxSAT solver.

Overall, \Cref{alg:topk-contribution} is the least affected by increasing the number of tuples and the number of criteria for moderately dense datasets, and it is the fastest algorithm in almost all settings. However, its runtime increases noticeably for sparse datasets such as \texttt{IPUMS-CPS}. As for the other algorithms, \Cref{alg:repairsat} is generally the slowest and most affected, while \Cref{alg:tvd-proxy} lies in between.

\begin{figure*}
  \centering
  \begin{subfigure}[b]{0.8\textwidth}
    \centering
    \includegraphics[width=\linewidth]{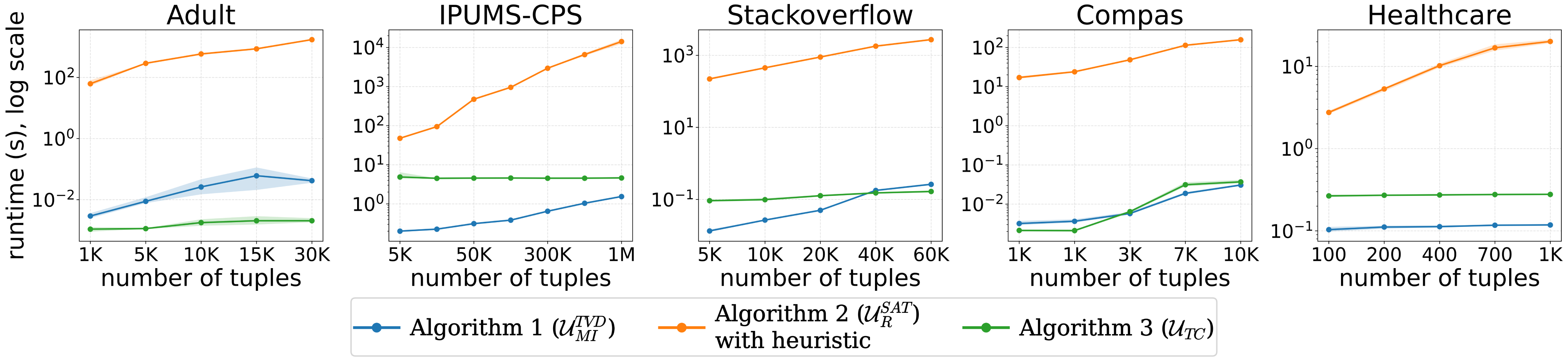}
    \caption{Runtime as function of number of tuples.}
    \label{fig:experiment-1}
  \end{subfigure}
  \hfill
  \begin{subfigure}[b]{0.8\textwidth}
    \centering
    \includegraphics[width=\linewidth]{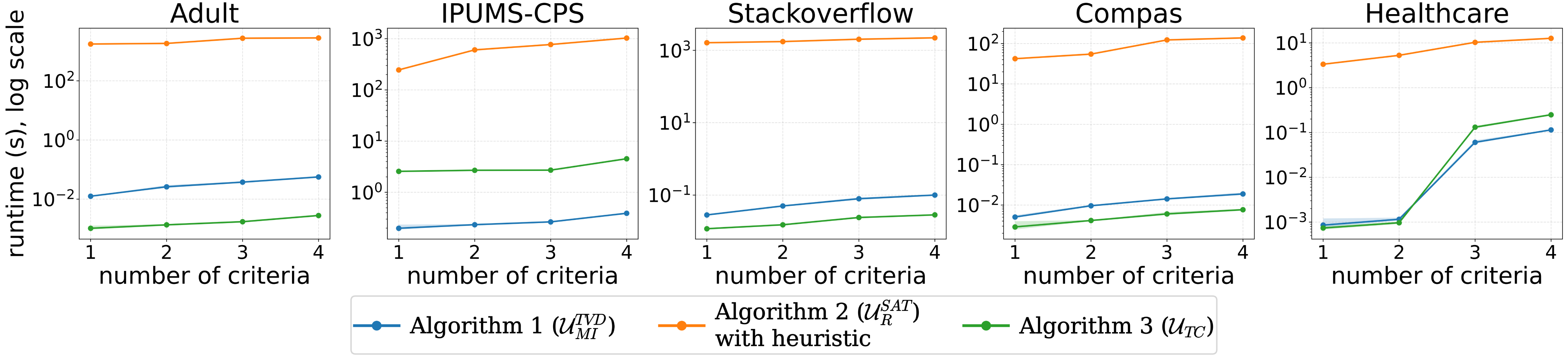}
    \caption{Runtime as function of number of fairness criteria.}
    \label{fig:experiment-2}
  \end{subfigure}
  \caption{Runtime analysis of the algorithms for the datasets and criteria in~\Cref{tab:fairness-criteria}.}
  \label{fig:experiments-1-2}
  \vspace{-4mm}
\end{figure*}

\begin{figure*}
    \centering
    \includegraphics[width=0.8\textwidth]{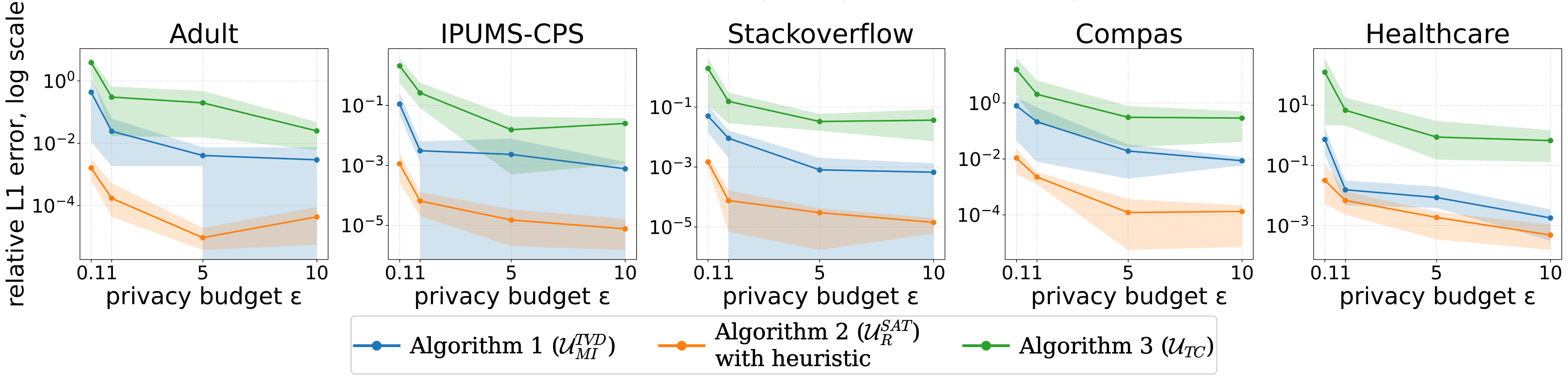}
    \caption{Relative $L1$ error as function of privacy budget for the datasets and criteria from~\Cref{tab:fairness-criteria}.}
    \label{fig:experiment-3}
    
    \vspace{-2mm}
\end{figure*}

\subsection{Accuracy and Privacy Tradeoff}
\Cref{fig:experiment-3} shows the effect of increasing the privacy budget on the relative $L1$ error of the algorithms.
\Cref{alg:topk-contribution} is consistently the least accurate of the three algorithms, due to its sensitivity being directly proportional to $k$ and inversely proportional to the minimum group size per admissible value (\Cref{prop:contribution-properties}). Consequently, for $k=500$, the noise added for all datasets is larger than that added for the other algorithms; for datasets with small minimum group sizes, such as \texttt{IPUMS-CPS} and \texttt{Healthcare}, this noise is particularly high. The parameter $k$ can be calibrated for each dataset to reduce the noise, as shown in a later experiment, thereby reducing the relative $L1$ error.
\Cref{alg:repairsat} is the most accurate of the three algorithms. This is because its values are much larger in magnitude, so the added DP noise changes the true value only by a small fraction, leading to a smaller relative error. Additionally, the relative error of \Cref{alg:repairsat} is the most stable across all datasets, since its sensitivity does not depend on the number of tuples (\Cref{prop:repairsat-properties}), causing the error to be primarily determined by $\varepsilon$.
\Cref{alg:tvd-proxy} is less noisy than \Cref{alg:topk-contribution}, which is consistent with its low sensitivity (\Cref{prop:tvd-properties}), but more noisy than \Cref{alg:repairsat}. Moreover, the relative error of \Cref{alg:tvd-proxy} exhibits large variance between runs, since its true values are small and therefore easily distorted by the added noise. However, for small values of $\varepsilon$, the noise dominates the output uniformly across runs, leading to reduced variance in the relative error.
Overall, the relative errors of all algorithms steadily decrease as the privacy budget increases, as expected, with the sharpest declining segment occurring up to $\varepsilon=1$.

\section{Related Work}
\label{sec:related}

Research on fairness in data management intersects several areas, including data valuation, private data pricing, and data inconsistency assessment. 
Although these directions are important foundations, {\em none develops dedicated measures for data-centric fairness tailored for use in the DP setting, which is the focus of this paper.}

A central line of work studies Data Valuation Assessment (DVA) in machine learning, where the goal is to estimate the contribution of individual records or subsets of data to predictive performance~\cite{DBLP:conf/iclr/KiCC23,DBLP:journals/tmlr/AlbalakEXLL0MHP24,DBLP:conf/iceis/BrennanAPNH19}. These methods quantify the marginal value of data with respect to accuracy or loss, and emphasize model driven utility rather than fairness or structural properties of the data itself.
Private data valuation methods provide mechanisms for pricing or compensating for the use of sensitive data without addressing fairness concerns~\cite{DBLP:conf/sigecom/GhoshR11,DBLP:journals/cacm/LiLMS17,DBLP:conf/uai/ZhangBL20}. These works typically quantify sensitivity, privacy loss, or economic value but do not analyze disparate impact or structural imbalances in the data.


Another major body of literature develops formal definitions of fairness
~\cite{proxyfair,InterFair,Counterfactual,VermaR18,Chouldechova17,DBLP:conf/kdd/Corbett-DaviesP17,DworkHPRZ12,Berk21,GalhotraBM17,KleinbergMR17,NabiS18} and introduces conceptual frameworks to analyze, understand, and mitigate bias in prediction systems~\cite{SelbstBFVV19,SureshG21,DBLP:conf/kdd/FeldmanFMSV15,DBLP:conf/icml/BrunetAAZ19}. These works largely focus on model behavior or causal interpretations rather than characterizing fairness properties inherent to the dataset itself. As we show in \Cref{sec:case-studies}, focusing on dataset unfairness rather than model behavior may be useful in some settings.



Inconsistency measures for relational data and integrity constraints~\cite{LivshitsKTIKR21,LivshitsK21,DBLP:journals/corr/abs-1904-03403,DBLP:conf/ecsqaru/MartinezPSSP07} offer tools for quantifying violations of integrity constraints. These measures are inspired by long standing work in Knowledge Representation and Logic~\cite{DBLP:conf/ijcai/KoniecznyLM03,DBLP:journals/jolli/Knight03,DBLP:journals/jiis/GrantH06,DBLP:journals/ai/HunterK10,DBLP:journals/kais/MuLJ11,DBLP:journals/ijar/GrantH17,DBLP:journals/ki/Thimm17}, which develops principles for reasoning under inconsistency. While conceptually related, these methods address logical coherence rather than fairness across protected groups.

Finally, fairness driven data repair methods aim to modify datasets to remove discriminatory patterns as pre- or post-processing solutions~\cite{DBLP:journals/kais/KamiranC11,DBLP:conf/kdd/FeldmanFMSV15,InterFair,DBLP:conf/icml/AgarwalBD0W18,DBLP:conf/nips/CalmonWVRV17,EqualOpp}. These interventions assume some fairness criterion but do not provide measures designed specifically for assessing fairness within the DP context.



\section{Conclusion and Limitations}\label{sec:conclusion}
We presented a principled framework for quantifying data unfairness under differential privacy, bridging concepts from probabilistic dependence, data repair, and contribution analysis. 
Our measures satisfy desirable consistency properties, exhibit low sensitivity, and remain tractable under standard DP mechanisms. 
Our three measures collectively provide complementary perspectives on data level bias.
Overall, the framework lays a foundation for a systematic evaluation of fairness in privacy protected data.

While our framework provides a principled foundation for measuring database unfairness under differential privacy, several limitations remain. First, $\repairsat$
relies on a MaxSAT-based proxy instantiated via a block-wise heuristic, which improves scalability but weakens estimation accuracy. Moreover, computing $\repairsat$ is slower than the other measures due to weighted MaxSAT over self-joined chunks, making efficient computation an important direction for future work.
Second, a principled approach for choosing k
would make $\contribution$ more user-friendly. Finally, our measures operate purely at the associational level and do not account for causal structure or data collection bias, which may influence fairness assessments.

\begin{acks}
The work of Mariia Vologdin and Amir Gilad was funded by the Israel Science Foundation (ISF) under grant 1702/24, the Scharf-Ullman Endowment, and the Alon Scholarship.
\end{acks}

\clearpage
\bibliographystyle{ACM-Reference-Format}
\bibliography{mybib}

\ifpaper

\else
\clearpage
\appendix

\section{Demonstrating the Faithfulness of \mutualproxytvd}
\Cref{fig:mi-proxies-comparison} demonstrates the faithfulness between $\mutual$ (\Cref{fig:mutual-comparison-1}) and $\mutualproxybayes$ (\Cref{fig:mutual-comparison-2}) for the \texttt{Adult}, \texttt{Stackoverflow} survey, and \texttt{Compas} datasets
(see~\Cref{footnote:datasets}) 
with fairness criteria from \Cref{tab:fairness-criteria}, numbered for each dataset. We also plotted $\mutualproxybayes$ with an offset of $\frac{1}{2}$ (\Cref{fig:mutual-comparison-3}) to make its values positive so the reader could easily see the variation between $\mutual$ and $\mutualproxybayes$ across all datasets and criteria.

\begin{figure}[h]
  \begin{subfigure}[b]{\linewidth}
    \includegraphics[width=\linewidth]{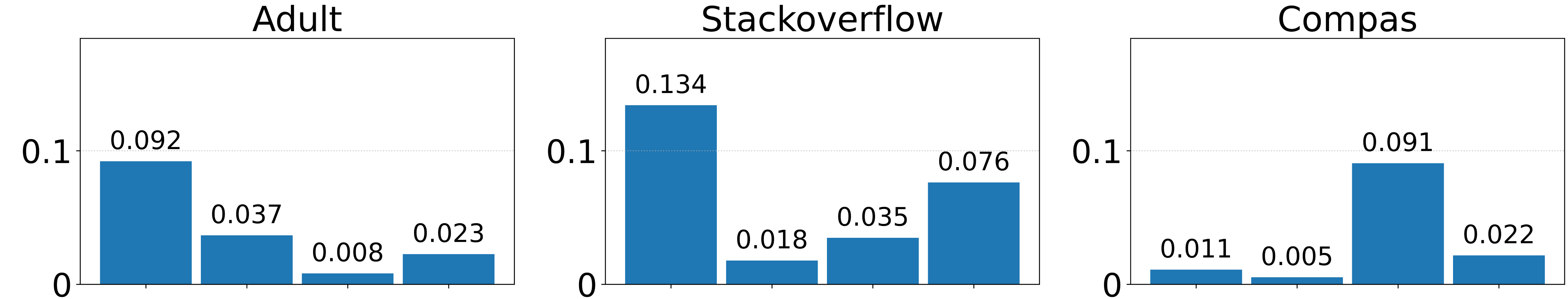}
    \caption{Values of $\mutual$.}
    \label{fig:mutual-comparison-1}
  \end{subfigure}
  \begin{subfigure}[b]{\linewidth}
    \includegraphics[width=\linewidth]{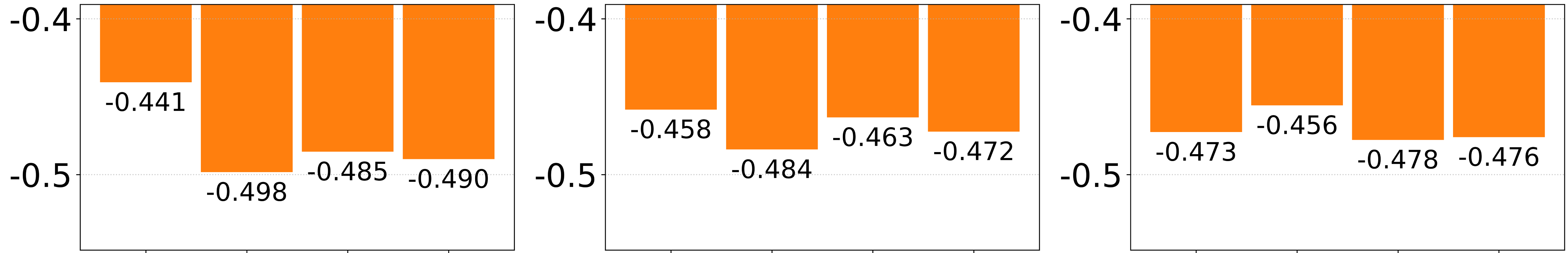}
    \caption{Values of $\mutualproxybayes$.}
    \label{fig:mutual-comparison-2}
  \end{subfigure}
  \begin{subfigure}[b]{\linewidth}
    \includegraphics[width=\linewidth]{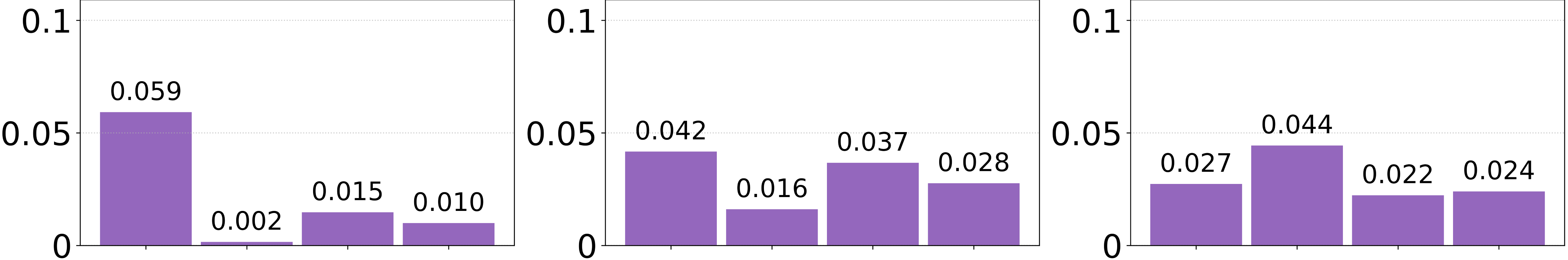}
    \caption{Values of $\mutualproxybayes$ with offset.}
    \label{fig:mutual-comparison-3}
  \end{subfigure}
  \begin{subfigure}[b]{\linewidth}
    \includegraphics[width=\linewidth]{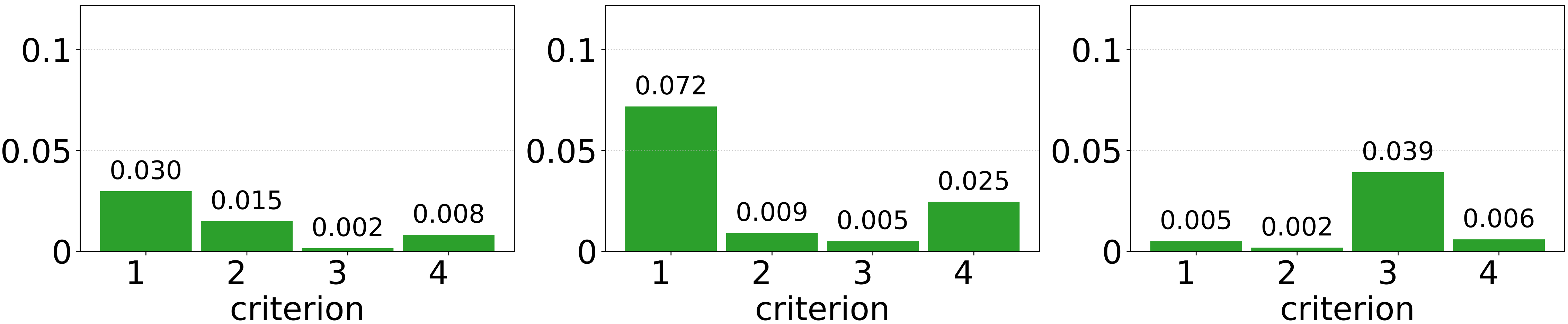}
    \caption{Values of $\mutualproxytvd$.}
    \label{fig:mutual-comparison-4}
  \end{subfigure}
  \caption{Demonstrating the faithfulness of $\mutualproxybayes$ and $\mutualproxytvd$ to $\mutual$ without privacy considerations.}
  \label{fig:mi-proxies-comparison}
\end{figure}

\section{Theorems and Proofs}\label{sec:appendix}

\subsection{Proofs for \Cref{sec:model}}

\begin{proof}[Proof of \Cref{prop:tvd-proxy}]
According to the Pinsker's inequality~\cite{csiszar2011information}:
{\scriptsize
$$\mathrm{TVD}(\prob(\protectedset,\responseset), \prob(\protectedset) \prob(\responseset))^2 \leq \frac{1}{2} D_{\mathrm{KL}}(\prob(\protectedset,\responseset) \,\|\, \prob(\protectedset) \prob(\responseset)) = \frac{1}{2} \mutual(\protectedset,\responseset)$$
}
it implies that
{\scriptsize
$$\mutualproxytvd(\protectedset \independent \responseset, D) = 2 \cdot \mathrm{TVD}(\prob(\protectedset,\responseset), \prob(\protectedset) \prob(\responseset))^2 \leq \mutual(\protectedset,\responseset)$$
}

From the other side, according to the reverse Pinsker's inequality, for any two distributions $P$ and $Q$ such that $Q(z) > 0$ for all $z$ it holds that
{\scriptsize
$$\mathrm{TVD}(P, Q)^2 \geq \frac{1}{2} \cdot \min_z Q(z) \cdot D_{\mathrm{KL}}(P \| Q)$$
}

Setting $P = \prob(\protectedset,\responseset)$ and $Q = \prob(\protectedset) \prob(\responseset)$, it implies that
{\scriptsize
$$\mathrm{TVD}(\prob(\protectedset,\responseset), \prob(\protectedset) \prob(\responseset))^2 \geq \frac{1}{2} \cdot \alpha \cdot \mutual(\protectedset,\responseset),$$
}
where $\alpha := \min_{(\protectedvalue,\responsevalue)} \prob(\protectedset=\protectedvalue) \prob(\responseset=\responsevalue)$ over the support of $\prob(\protectedset,\responseset)$.

Multiplying both sides by 2, we get
{\scriptsize
$$\mutualproxytvd(\protectedset \independent \responseset, D) = 2 \cdot \mathrm{TVD}(\prob(\protectedset,\responseset), \prob(\protectedset) \prob(\responseset))^2 \geq \alpha \cdot \mutual(\protectedset,\responseset)$$
}

Therefore
{\scriptsize
$$\alpha \cdot \mutual(\protectedset,\responseset) \leq \mutualproxytvd(\protectedset \independent \responseset, D) \leq \mutual(\protectedset,\responseset),$$
}
where $\alpha := \min_{(\protectedvalue,\responsevalue)} \prob(\protectedset=\protectedvalue) \prob(\responseset=\responsevalue)$.

Then, if we define $X = \prob(\protectedset,\responseset)$, $Y = \prob(\protectedset) \prob(\responseset)$, we would get that the following holds
{\scriptsize
$$\alpha \cdot \mutual(\protectedset,\responseset) \mutualproxytvd(\protectedset \independent \responseset, D) \leq \mutual(\protectedset,\responseset),$$
}
where $\alpha = min_{z\in Z, Y(z) > 0} Y(z)$.
\end{proof}

\begin{proof}[Proof of \Cref{prop:tvd-properties} Parts 1 and 2]
\textbf{Positivity:}

By the definition of total variation distance, for any two probability distributions $P$ and $Q$ it holds that
{\scriptsize
$$\mathrm{TVD}(P, Q) = \frac{1}{2} \sum_{x} |P(x) - Q(x)| \geq 0$$
}
and it holds that $\mathrm{TVD}(P, Q) = 0$ iff $P = Q$ by the properties of the absolute value.

For any database $D$ and any fairness criterion $\fairnessdef = \protectedset \independent \responseset \mid \admissibleset \in \constraints$ where $\admissibleset$ can be absent, it holds that
{\scriptsize
\begin{align*}
&\mutualproxytvd(\protectedset \independent \responseset \mid \admissibleset, D) \\
&= 2 \cdot \left( \mathrm{TVD}(\prob(\protectedset, \responseset \mid \admissibleset), \prob(\protectedset \mid \admissibleset)\prob(\responseset \mid \admissibleset)) \right)^2 \\
&\geq \mathrm{TVD}(\prob(\protectedset, \responseset \mid \admissibleset), \prob(\protectedset \mid \admissibleset)\prob(\responseset \mid \admissibleset)) \\
&\geq 0    
\end{align*}
}
and
{\scriptsize
\begin{align*}
&\mutualproxytvd(\protectedset \independent \responseset \mid \admissibleset, D) = 0 \\
&\iff 2 \cdot \left( \mathrm{TVD}(\prob(\protectedset, \responseset \mid \admissibleset), \prob(\protectedset \mid \admissibleset)\prob(\responseset \mid \admissibleset)) \right)^2 = 0 \\
&\iff \mathrm{TVD}(\prob(\protectedset, \responseset \mid \admissibleset), \prob(\protectedset \mid \admissibleset)\prob(\responseset \mid \admissibleset)) = 0 \\
&\iff \prob(\protectedset, \responseset \mid \admissibleset) = \prob(\protectedset \mid \admissibleset)\prob(\responseset \mid \admissibleset) \\
&\iff \fairnessdef(D) = 1    
\end{align*}
}

We can notice that the above also holds for any a set of fairness criteria $\constraints$ since $\mutualproxytvd(\constraints, D) := \sum_{\fairnessdef \in \constraints} \mutualproxytvd(\fairnessdef, D)$.

\textbf{Monotonicity:}

For any database $D$ and fairness criteria $\constraints,\constraints'$ such that $\constraints \subseteq \constraints'$, it holds that
{\scriptsize
$$\mutualproxytvd(\constraints, D) = \sum_{\fairnessdef \in \constraints} \mutualproxytvd(\fairnessdef, D) \leq \sum_{\fairnessdef' \in \constraints'} \mutualproxytvd(\fairnessdef', D) = \mutualproxytvd(\constraints', D)$$
}

This follows from the fact that $\mutualproxytvd$ is non-negative (from the Positivity property), and so its sum is monotonic.
\end{proof}

We will first state and prove the range and the sensitivity bounds of $\mutualproxytvd$ for a single fairness criterion, and then we will use them to prove the range and the sensitivity bounds for a set of criteria (Parts 3 and 4 of \Cref{prop:tvd-properties}).

\begin{lemma}[$\mutualproxytvd$ is suitable for DP for a single fairness criterion]\label{prop:tvd-properties-single-criterion}
Given a database $D$ of size $n$, the following holds:
\begin{enumerate}
    \item The range of $\mutualproxytvd$ is $[0,2]$ for a single criterion $\fairnessdef$.
    \item The sensitivity of $\mutualproxytvd$ for an unconditional fairness criterion $\fairnessdef = \protectedset \independent \responseset$ is at most $\frac{12}{n}$ and for a conditional fairness criterion, $\fairnessdef = \protectedset \independent \responseset \mid \admissibleset$, is at most $\frac{16}{n}$.
\end{enumerate}
\end{lemma}

\begin{proof}[Proof of \Cref{prop:tvd-properties-single-criterion} Part 1]
For any $P$ and $Q$ two probability distributions, it holds that:
{\scriptsize
$$0 \leq \frac{1}{2} \sum_{x} |P(x) - Q(x)| \leq \frac{1}{2} (|1-0| + |0-1|) = 1$$
}
by the properties of probability. The upper bound is for the case when all probability mass in $P$ is at one outcome, while all probability mass in $Q$ is at the other. In this case, only two terms contribute to the sum, and both contribute $1$.

Therefore by the definition of $\mathrm{TVD}$ it holds that:
{\scriptsize
$$0 = 2 \cdot 0^2 \leq 2 \cdot \left( \mathrm{TVD}(P,Q) \right)^2 \leq 2 \cdot 1^2 = 2$$
}

And so, by the definition, the range of $\mutualproxytvd$ is $[0,2]$.
\end{proof}

Now we will state and prove a few lemmas that we will use in a later proof of \Cref{prop:tvd-properties-single-criterion} Part 2.

\begin{lemma}[Reverse Triangle Inequality]\label{lem:reverse-triangle-ineq}
Given $x, y \in \mathbb{R}$, the following holds
$$\left|\,|x| - |y|\,\right| \leq |x - y|$$
\end{lemma}

\begin{proof}
From the triangle inequality,
{\scriptsize
$$|x| = |(x - y) + y| \leq |x - y| + |y|$$
}
Rearranging, we get
{\scriptsize
$$|x| - |y| \leq |x - y|$$
}

Similarly, by symmetry,
{\scriptsize
$$|y| = |(y - x) + x| \leq |y - x| + |x| = |x - y| + |x|$$
}
Rearranging, we get
{\scriptsize
$$|y| - |x| \leq |x - y|$$
}
and multiplying both sides by $-1$ we get
{\scriptsize
$$-|x - y| \leq -(|y| - |x|) = |x| - |y|$$
}

Combining the two, we get
{\scriptsize
$$-|x - y| \leq |x| - |y| \leq |x - y|$$
}
which by the definition of the absolute value means
{\scriptsize
$$\left|\,|x| - |y|\,\right| \leq |x - y|$$
}
\end{proof}

\begin{lemma}[Sensitivity of the square]\label{lem:sensitivity-of-square}
Let $f: \mathcal{D} \to \mathbb{R}$ be a function with sensitivity $\sens_f$. Define $g(D) := f(D)^2$. Then the sensitivity of $g$ satisfies:
$$\Delta_g := \max_{D \neighbor D'} |g(D) - g(D')| \leq 2 \cdot \sens_f \cdot \max(\dom(f))$$
\end{lemma}

\begin{proof}
Let $D \neighbor D'$ be neighboring databases. Define $g(D) = f(D)^2$, so
{\scriptsize
\begin{align*}
&|g(D) - g(D')| = |f(D)^2 - f(D')^2| \\
&= |(f(D) - f(D')) \cdot (f(D) + f(D'))| \\
&\leq |f(D) - f(D')| \cdot |f(D) + f(D')| \\ 
&\leq |f(D) - f(D')| \cdot 2 \cdot \max(\dom(f)) \\ 
&\leq \sens_f \cdot 2 \cdot \max(\dom(f))
\end{align*}
}

Taking the maximum over all neighboring databases, we get
{\scriptsize
\begin{align*}
\sens_g \leq 2 \cdot \sens_f \cdot \max(\dom(f))    
\end{align*}
}
\end{proof}

\begin{proof}[Proof of \Cref{prop:tvd-properties-single-criterion} Part 2]
\textbf{Unconditional case:}

Let there be $D \neighbor D'$ neighboring databases such that they differ in one tuple by replacement. Without loss of generality, assume that $D'$ contains one more occurrence of $t'' = (\protectedvalue_{t''},\responsevalue_{t''})$ and one less occurrence of $t' = (\protectedvalue_{t'},\responsevalue_{t'})$ than $D$.

From \Cref{def:proxy-tvd}, by the triangle inequality and \Cref{lem:reverse-triangle-ineq}, we have the following
{\scriptsize
\setlength{\abovedisplayskip}{4pt}
\setlength{\belowdisplayskip}{4pt}
\begin{align}
&\left| \mathrm{TVD}(\probneighbor(\protectedset,\responseset), \probneighbor(\protectedset)\probneighbor(\responseset)) 
      - \mathrm{TVD}(\prob(\protectedset,\responseset), \prob(\protectedset)\prob(\responseset)) \right| \\
&= \Bigl| \, \frac{1}{2} \sum_{\protectedvalue \in \protectedset, \responsevalue \in \responseset} 
   \left| \probneighbor(\protectedset=\protectedvalue,\responseset=\responsevalue) - \probneighbor(\protectedset=\protectedvalue)\probneighbor(\responseset=\responsevalue) \right| \nonumber\\[-4pt]
&\qquad\quad - \frac{1}{2} \sum_{\protectedvalue \in \protectedset, \responsevalue \in \responseset} 
   \left| \prob(\protectedset=\protectedvalue,\responseset=\responsevalue) - \prob(\protectedset=\protectedvalue)\prob(\responseset=\responsevalue) \right| \, \Bigr| \label{eq:1_1}\\
&\leq \Bigl| \frac{1}{2} \sum_{\protectedvalue \in \protectedset, \responsevalue \in \responseset}  
  \Bigl( \probneighbor(\protectedset=\protectedvalue,\responseset=\responsevalue) - \probneighbor(\protectedset=\protectedvalue)\probneighbor(\responseset=\responsevalue) \nonumber\\[-4pt]
  &\qquad\qquad\qquad - \prob(\protectedset=\protectedvalue,\responseset=\responsevalue) + \prob(\protectedset=\protectedvalue)\prob(\responseset=\responsevalue) \Bigr) 
\Bigr| \label{eq:1_2}\\
&= \frac{1}{2} \Bigl| \sum_{\protectedvalue \in \protectedset, \responsevalue \in \responseset} 
  \Bigl( \probneighbor(\protectedset=\protectedvalue,\responseset=\responsevalue) - \prob(\protectedset=\protectedvalue,\responseset=\responsevalue) \Bigr) \nonumber\\[-4pt]
  &\qquad\qquad\quad + \sum_{\protectedvalue \in \protectedset, \responsevalue \in \responseset} 
  \Bigl( \prob(\protectedset=\protectedvalue)\prob(\responseset=\responsevalue) - \probneighbor(\protectedset=\protectedvalue)\probneighbor(\responseset=\responsevalue) \Bigr) 
\Bigr| \label{eq:1_3}\\
&\leq \frac{1}{2} \left| \sum_{\protectedvalue \in \protectedset, \responsevalue \in \responseset} 
  \probneighbor(\protectedset=\protectedvalue,\responseset=\responsevalue) - \prob(\protectedset=\protectedvalue,\responseset=\responsevalue) \right| \nonumber\\[-4pt]
&\qquad\qquad\quad + \frac{1}{2} \left| \sum_{\protectedvalue \in \protectedset, \responsevalue \in \responseset} 
  \probneighbor(\protectedset=\protectedvalue)\probneighbor(\responseset=\responsevalue) - \prob(\protectedset=\protectedvalue)\prob(\responseset=\responsevalue) \right| \label{eq:1_4}\\
&\leq \frac{1}{2} \sum_{\protectedvalue \in \protectedset, \responsevalue \in \responseset} 
  \left| \probneighbor(\protectedset=\protectedvalue,\responseset=\responsevalue) - \prob(\protectedset=\protectedvalue,\responseset=\responsevalue) \right| \nonumber\\[-4pt]
&\qquad\qquad\quad + \frac{1}{2} \sum_{\protectedvalue \in \protectedset, \responsevalue \in \responseset} 
  \left| \probneighbor(\protectedset=\protectedvalue)\probneighbor(\responseset=\responsevalue) - \prob(\protectedset=\protectedvalue)\prob(\responseset=\responsevalue) \right| \label{eq:1_5}
\end{align}
}
The transition from \eqref{eq:1_1} to \eqref{eq:1_2} is due to \Cref{lem:reverse-triangle-ineq}, the transitions from \eqref{eq:1_3} to \eqref{eq:1_4} and from \eqref{eq:1_4} to \eqref{eq:1_5} are due to triangle inequality.

Looking at the first part of the sum separately, we get
{\scriptsize
\setlength{\abovedisplayskip}{4pt}
\setlength{\belowdisplayskip}{4pt}
\begin{align*}
&\frac{1}{2} \sum_{\protectedvalue \in \protectedset, \responsevalue \in \responseset} \left| \probneighbor(\protectedset=\protectedvalue,\responseset=\responsevalue) - \prob(\protectedset=\protectedvalue,\responseset=\responsevalue) \right| \\
&\leq \frac{1}{2} \Bigl(\left| \probneighbor(\protectedset=\protectedvalue_{t'},\responseset=\responsevalue_{t'}) - \prob(\protectedset=\protectedvalue_{t'},\responseset=\responsevalue_{t'}) \right| \\
&\quad + \left| \probneighbor(\protectedset=\protectedvalue_{t''},\responseset=\responsevalue_{t''}) - \prob(\protectedset=\protectedvalue_{t''},\responseset=\responsevalue_{t''}) \right| \Bigr) \\
&\leq \frac{1}{2} \cdot \left(\frac{1}{n} + \frac{1}{n}\right)  \\
&= \frac{1}{n}    
\end{align*}
}

Looking at the second part of the sum, we get
{\scriptsize
\begin{align}
&\frac{1}{2} \sum_{\protectedvalue \in \protectedset, \responsevalue \in \responseset} 
   \Bigl| \probneighbor(\protectedset=\protectedvalue)\,
          \probneighbor(\responseset=\responsevalue) - \prob(\protectedset=\protectedvalue)\,
          \prob(\responseset=\responsevalue) 
    \Bigr| \\
&= \frac{1}{2} \sum_{\protectedvalue \in \protectedset, \responsevalue \in \responseset} 
   \Bigl| \probneighbor(\protectedset=\protectedvalue)\,
          \probneighbor(\responseset=\responsevalue) + \prob(\protectedset=\protectedvalue)\,
          \probneighbor(\responseset=\responsevalue) \nonumber\\[-4pt]
&\qquad\qquad\qquad - \prob(\protectedset=\protectedvalue)\,
          \probneighbor(\responseset=\responsevalue) - \prob(\protectedset=\protectedvalue)\,
          \prob(\responseset=\responsevalue)\Bigr| \\
&= \tfrac{1}{2} \sum_{\protectedvalue \in \protectedset, \responsevalue \in \responseset} 
   \Bigl| \bigl(\probneighbor(\protectedset=\protectedvalue) 
          - \prob(\protectedset=\protectedvalue)\bigr)\,
          \probneighbor(\responseset=\responsevalue) \nonumber\\[-4pt]
&\qquad\qquad\qquad + \prob(\protectedset=\protectedvalue)\,
          \bigl(\probneighbor(\responseset=\responsevalue) 
          - \prob(\responseset=\responsevalue)\bigr) \Bigr| \label{eq:2_1}\\
&\leq \frac{1}{2} \sum_{\protectedvalue \in \protectedset, \responsevalue \in \responseset} 
   \Bigl(\bigl| \probneighbor(\protectedset=\protectedvalue) 
          - \prob(\protectedset=\protectedvalue) \bigr|\,
          \probneighbor(\responseset=\responsevalue) \nonumber\\
&\qquad\qquad\qquad + \prob(\protectedset=\protectedvalue)\,
          \bigl| \probneighbor(\responseset=\responsevalue) 
          - \prob(\responseset=\responsevalue) \bigr| \Bigr) \label{eq:2_2}\\
&= \frac{1}{2} \Biggl( 
   \sum_{\protectedvalue \in \protectedset} 
      \bigl| \probneighbor(\protectedset=\protectedvalue) 
          - \prob(\protectedset=\protectedvalue) \bigr|
      \sum_{\responsevalue \in \responseset} 
          \probneighbor(\responseset=\responsevalue) \nonumber\\[-4pt]
&\qquad + \sum_{\responsevalue \in \responseset} 
      \bigl| \probneighbor(\responseset=\responsevalue) 
          - \prob(\responseset=\responsevalue) \bigr|
      \sum_{\protectedvalue \in \protectedset} 
          \prob(\protectedset=\protectedvalue) \Biggr) 
   \label{eq:2_3}\\
&= \frac{1}{2} \Biggl( 
   \sum_{\protectedvalue \in \protectedset} 
      \bigl| \probneighbor(\protectedset=\protectedvalue) 
          - \prob(\protectedset=\protectedvalue) \bigr| \cdot 1 + \sum_{\responsevalue \in \responseset} 
      \bigl| \probneighbor(\responseset=\responsevalue) 
          - \prob(\responseset=\responsevalue) \bigr| \cdot 1 
   \Biggr) 
   \label{eq:2_4}\\
&= \frac{1}{2} \sum_{\protectedvalue \in \protectedset} 
       \bigl| \probneighbor(\protectedset=\protectedvalue) 
          - \prob(\protectedset=\protectedvalue) \bigr| + \frac{1}{2} \sum_{\responsevalue \in \responseset} 
       \bigl| \probneighbor(\responseset=\responsevalue) 
          - \prob(\responseset=\responsevalue) \bigr|
\end{align}
}
where the transition from \eqref{eq:2_1} to \eqref{eq:2_2} is due to applying the triangle inequality, the transition from \eqref{eq:2_3} to \eqref{eq:2_4} is due to the fact that the sum of probabilities for all the values is $1$.

In particular,
{\scriptsize
\begin{align*}
&\frac{1}{2} \sum_{\protectedvalue \in \protectedset} |\probneighbor(\protectedset=\protectedvalue) - \prob(\protectedset=\protectedvalue)| \\
&\leq \frac{1}{2} \cdot \left(\frac{1}{n} + \frac{1}{n}\right) \\
&= \frac{1}{n}
\end{align*}
}

A similar expression can be achieved by swapping $\protectedset$ and $\protectedvalue$ for $\responseset$ and $\responsevalue$, respectively.

Putting it all together,
{\scriptsize
\begin{align*}
&\left| \mathrm{TVD}(\probneighbor(\protectedset,\responseset), \probneighbor(\protectedset)\probneighbor(\responseset)) 
     - \mathrm{TVD}(\prob(\protectedset,\responseset), \prob(\protectedset)\prob(\responseset)) \right| \\
&\leq 
\frac{1}{2} \sum_{\protectedvalue \in \protectedset, \responsevalue \in \responseset} |\probneighbor(\protectedvalue,\responsevalue) - \prob(\protectedvalue,\responsevalue)| \\
&\quad + \frac{1}{2} \sum_{\protectedvalue \in \protectedset, \responsevalue \in \responseset} |\probneighbor(\protectedvalue)\probneighbor(\responsevalue) - \prob(\protectedvalue)\prob(\responsevalue)| \\
&\leq \frac{1}{n} + \frac{1}{2} \sum_{\protectedvalue \in \protectedset} |\probneighbor(\protectedvalue) - \prob(\protectedvalue)| 
+ \frac{1}{2} \sum_{\responsevalue \in \responseset} |\probneighbor(\responsevalue) - \prob(\responsevalue)| \\
&\leq \frac{3}{n}
\end{align*}
}

Taking the maximum over all neighboring databases $D \neighbor D'$, we get
{\scriptsize
$$\sens_{\mathrm{TVD}} \leq \frac{3}{n}$$
}

From \Cref{lem:sensitivity-of-square}, it holds that
{\scriptsize
\begin{align*}
&\sens_{\mutualproxytvd} = \sens_{2\mathrm{TVD}^2} \\
&= 2 \cdot \sens_{\mathrm{TVD}^2} \\
&\leq 2 \cdot 2 \cdot \sens_{\mathrm{TVD}} \cdot \max(\dom(\mathrm{TVD})) \\
&\leq^{\mathrm{TVD} \leq 1} 4 \cdot \frac{3}{n} \cdot 1 \\
&= \frac{12}{n}    
\end{align*}
}

Therefore, the sensitivity bound for the unconditional case is $\frac{12}{n}$.

\textbf{Conditional case:}

Similarly, assume that $D'$ contains one more occurrence of $t'' = (\protectedvalue_{t''},\responsevalue_{t''},\admissiblevalue_{t''})$ and one less occurrence of $t' = (\protectedvalue_{t'},\responsevalue_{t'},\admissiblevalue_{t'})$ than $D$.

Denote $f_D(\admissiblevalue) := \mutualproxytvd(\protectedset \independent \responseset \mid \admissibleset=\admissiblevalue, D)$, and similarly $f_{D'}(\admissiblevalue)$. Then from \Cref{def:proxy-tvd} it holds that
{\scriptsize
\begin{align}
&\bigl|\mutualproxytvd(\protectedset \independent \responseset \mid \admissibleset, D') 
     - \mutualproxytvd(\protectedset \independent \responseset \mid \admissibleset, D)\bigr| \\
&= \left|\sum_{\admissiblevalue} \left( \probneighbor(\admissibleset=\admissiblevalue) f_{D'}(\admissiblevalue) - \prob(\admissibleset=\admissiblevalue) f_{D}(\admissiblevalue) \right)\right| \label{eq:3_1}\\
&\leq \sum_{\admissiblevalue} \bigl|\probneighbor(\admissibleset=\admissiblevalue) f_{D'}(\admissiblevalue) - \prob(\admissibleset=\admissiblevalue) f_{D}(\admissiblevalue)\bigr| \\
&= \sum_{\admissiblevalue} \left|\Bigl(\left( \probneighbor(\admissibleset=\admissiblevalue) - \prob(\admissibleset=\admissiblevalue) \right) f_{D'}(\admissiblevalue) \Bigr) + \Bigl( \prob(\admissibleset=\admissiblevalue) \left( f_{D'}(\admissiblevalue) - f_{D}(\admissiblevalue) \right)\Bigr)\right| \label{eq:3_2}\\
&\leq \sum_{\admissiblevalue} \left|\Bigl|\left( \probneighbor(\admissibleset=\admissiblevalue) - \prob(\admissibleset=\admissiblevalue) \right) f_{D'}(\admissiblevalue) \Bigr| + \Bigl| \prob(\admissibleset=\admissiblevalue) \left( f_{D'}(\admissiblevalue) - f_{D}(\admissiblevalue) \right)\Bigr|\right| \label{eq:3_3}\\
&\leq \sum_{\admissiblevalue} \left| \probneighbor(\admissibleset=\admissiblevalue) - \prob(\admissibleset=\admissiblevalue) \right| \left| f_{D'}(\admissiblevalue) \right| + \sum_{\admissiblevalue} \prob(\admissibleset=\admissiblevalue) \left| f_{D'}(\admissiblevalue) - f_{D}(\admissiblevalue) \right|, \label{eq:3_4}
\end{align}
}
where the transitions are due to:
\begin{itemize}
  \item $\eqref{eq:3_1} \rightarrow \eqref{eq:3_2}$: triangle inequality.
  \item $\eqref{eq:3_2} \rightarrow \eqref{eq:3_3}$: triangle inequality.
  \item $\eqref{eq:3_3} \rightarrow \eqref{eq:3_4}$: $xy \leq |x||y|$ for $x,y \in \mathbb{R}$.
\end{itemize}

Looking at the first sum separately, because only $\admissiblevalue \in \{\admissiblevalue_{t'}, \admissiblevalue_{t''}\}$ change, and since $|f_{D'}(\admissiblevalue)| \leq 2$ by the Range property from \Cref{prop:tvd-properties} proved earlier, then
{\scriptsize
$$
\sum_{\admissiblevalue} \left| \probneighbor(\admissibleset=\admissiblevalue) - \prob(\admissibleset=\admissiblevalue) \right| \left| f_{D'}(\admissiblevalue) \right| \leq \left( \frac{1}{n} + \frac{1}{n} \right) \cdot 2 = \frac{4}{n}
$$
}

Looking at the second sum, by the unconditional case and by the \Cref{lem:sensitivity-of-square}, for a fixed $\admissiblevalue$ with $n_{\admissiblevalue} = |\{ t \in D \mid t[\admissibleset] = \admissiblevalue \}|$, it holds that the sensitivity of $2 \cdot \mathrm{TVD}^2$ is at most $\frac{12}{n_{\admissiblevalue}}$ when both removing $t'$ and adding $t''$ are included in the computation. So, $\left| f_{D'}(\admissiblevalue_{t'}) - f_{D}(\admissiblevalue_{t'}) \right| = \frac{1}{2} \cdot \frac{12}{n_{\admissiblevalue_{t'}}}$, and similarly for $t''$. Therefore
{\scriptsize
\begin{align*}
&\sum_{\admissiblevalue} \prob(\admissibleset=\admissiblevalue) \left| f_{D'}(\admissiblevalue) - f_{D}(\admissiblevalue) \right| \\
&= \prob(\admissibleset=\admissiblevalue_{t'}) \left| f_{D'}(\admissiblevalue_{t'}) - f_{D}(\admissiblevalue_{t'}) \right| + \prob(\admissibleset=\admissiblevalue_{t''}) \left| f_{D'}(\admissiblevalue_{t''}) - f_{D}(\admissiblevalue_{t''}) \right| \\
&\leq \frac{n_{\admissiblevalue_{t'}}}{n} \cdot \frac{6}{n_{\admissiblevalue_{t'}}} + \frac{n_{\admissiblevalue_{t''}}}{n} \cdot \frac{6}{n_{\admissiblevalue_{t''}}} \\
&= \frac{12}{n}
\end{align*}
}

Putting it all together,
{\scriptsize
$$
\sens_{\mutualproxytvd} = \max_{D' \neighbor D} \left| \mutualproxytvd(\protectedset \independent \responseset \mid \admissibleset, D') 
     - \mutualproxytvd(\protectedset \independent \responseset \mid \admissibleset, D) \right| \leq \frac{4}{n} + \frac{12}{n} = \frac{16}{n}
$$
}

Therefore, the sensitivity bound for the conditional case is $\frac{16}{n}$.

\end{proof}

\begin{proof}[Proof of \Cref{prop:tvd-properties} Parts 3 and 4]
\textbf{Range:}

By \Cref{def:proxy-tvd} and by \Cref{prop:tvd-properties-single-criterion}, it holds that
{\scriptsize
$$
\mutualproxytvd(\constraints, D) = \sum_{\fairnessdef \in \constraints} \mutualproxytvd(\fairnessdef, D) \geq \sum_{\fairnessdef \in \constraints} 0 = 0
$$
}
and
{\scriptsize
$$
\mutualproxytvd(\constraints, D) = \sum_{\fairnessdef \in \constraints} \mutualproxytvd(\fairnessdef, D) \leq \sum_{\fairnessdef \in \constraints} 2 = 2|\constraints|
$$
}

\textbf{Sensitivity:}

From \Cref{prop:tvd-properties-single-criterion}, the sensitivity for a single fairness criterion $\fairnessdef$ is at most
{\scriptsize
$$\sens_{\mutualproxytvd(\fairnessdef, D)} \leq \frac{16}{n}$$
}

Therefore, for $D' \neighbor D$, using the triangle inequality and linearity of the sum, it holds that
{\scriptsize
\begin{align*}
\bigl| \mutualproxytvd(\constraints, D') - \mutualproxytvd(\constraints, D) \bigr| &= \Bigl| \sum_{\fairnessdef \in \constraints} \bigl(\mutualproxytvd(\fairnessdef, D')\bigr) - \sum_{\fairnessdef \in \constraints} \bigl(\mutualproxytvd(\fairnessdef, D)\bigr) \Bigr| \\
&\leq \sum_{\fairnessdef \in \constraints} \bigl| \mutualproxytvd(\fairnessdef, D') - \mutualproxytvd(\fairnessdef, D) \bigr| \\
&\leq \sum_{\fairnessdef \in \constraints} \sens_{\mutualproxytvd(\fairnessdef, D)} \\
&\leq \frac{16|\constraints|}{n}
\end{align*}
}
\end{proof}

\begin{proposition}[$\repair$ satisfies the desired properties]~\label{prop:repair-properties}
\begin{enumerate}
\item $\repair$ satisfies the Positivity property.
\item $\repair$ satisfies the Monotonicity property. 
\item The range of $\repair$ is $[0,n]$ for a set of criteria $\constraints$.
\item The sensitivity of $\repair$ is $1$ for a set of criteria $\constraints$.
\end{enumerate}
\end{proposition}

\begin{proof}[Proof of \Cref{prop:repair-properties} Parts 1 and 2]
\textbf{Positivity:}

The positivity property holds, as \Cref{def:repair} defines $\repair(\constraints, D) = |D \symdif D_R|$ with an absolute value. 

\textbf{Monotonicity:}

For monotonicity, assume that $\constraints \subseteq \hat{\constraints}$ and that $D_R$ and $\hat{D_R}$ are the minimum repaired databases w.r.t. $D$ and $\constraints, \hat{\constraints}$ respectively.

By definition $\hat{D_R}$ satisfies $\hat{\constraints}$ and thus in particular satisfies $\constraints$. Hence, $|D \symdif \hat{D_R}|$ can only be equal or larger than $|D \symdif D_R|$, since by our definition $D_R$ is the database that {\bf minimizes} the symmetric difference with $D$ and satisfies $\constraints$. 
Therefore, $\repair(\constraints, D) = |D \symdif D_R| \leq |D \symdif \hat{D_R}| = \repair(\hat{\constraints}, D)$.
\end{proof}

We will first state and prove the range and the sensitivity bounds of $\repair$ for a single fairness criterion, and then we will use them to prove the range and the sensitivity bounds for a set of criteria (Parts 3 and 4 of \Cref{prop:repair-properties}).

\begin{lemma}[$\repair$ is suitable for DP for a single fairness criterion]\label{prop:repair-properties-single-criterion}
Given a database $D$ of size $n$, the following holds:
\begin{enumerate}
    \item The range of $\repair(\fairnessdef, D)$ is $[0, n]$ for a single criterion $\fairnessdef$.
    \item The sensitivity of $\repair$ is at most $1$ for a single criterion $\fairnessdef$.
\end{enumerate}
\end{lemma}

\begin{proof}
\textbf{Range:}

It is immediate that $\range(\repair(\fairnessdef, D)) = [0,n]$ because $\repair(\fairnessdef, D) \geq 0$ from the Positivity property, and $\repair(\fairnessdef, D) \leq n$ because in the most extreme case, deleting all tuples in $D$ will create the empty database that satisfies $\fairnessdef$.

\textbf{Sensitivity:}

Given $\fairnessdef$ and $D \neighbor D'$ where without loss of generality $D' = D \setminus \{t'\} \cup \{t''\}$ for some tuples $t'$ and $t''$, we have $\repair(\fairnessdef,D) = |D \symdif D_R|$, where $D_R$ is the database that satisfies $\fairnessdef$ with the minimum number of deletions and additions.

Since $D'$ differs from $D$ by replacing a single tuple $t'$ with $t''$, we can obtain the repair process for $D'$ from the repair process for $D$ as follows:
\begin{itemize}
    \item If $t''$ interferes with $\fairnessdef$, replace $t''$ with $t'$ to obtain $D$, and then do the repair for $D$. In this case $\repair(\fairnessdef, D') = \repair(\fairnessdef, D) + 1$.
    \item Otherwise, $t''$ does not interfere with $\fairnessdef$, and so in this case $\repair(\fairnessdef, D') = \repair(\fairnessdef, D)$.
\end{itemize}

Therefore, $\sens_{\repair(\fairnessdef, D)} = \max_{D' \neighbor D} \left| \repair(\fairnessdef, D') - \repair(\fairnessdef, D) \right| \leq 1$, and the sensitivity of $\repair$ is at most $1$ for a single criterion $\fairnessdef$.
\end{proof}

\begin{proof}[Proof of \Cref{prop:repair-properties} Parts 3 and 4]
\textbf{Range:}

Again, $\range(\repair(\constraints, D)) = [0,n]$ because $\repair(\constraints, D) \geq 0$ from the Positivity property, and $\repair(\constraints, D) \leq n$ because in the most extreme case, deleting all tuples in $D$ will create the empty database that satisfies $\constraints$.

\textbf{Sensitivity:}

Let there be a database $D_R$ with the smallest number of removed and added tuples with respect to $D$ that satisfies all the fairness criteria $\constraints$. Then
{\scriptsize
\begin{align*}
\repair(\constraints, D') \leq |D' \symdif D_R| \leq |D \symdif D_R| + 1 = \repair(\constraints, D) + 1
\end{align*}
}

From another side, let there be a database ${D'}_R$ with the smallest number of removed and added tuples with respect to $D'$ that satisfies all the fairness criteria $\constraints$. Then
{\scriptsize
\begin{align*}
\repair(\constraints, D) \leq |D \symdif {D'}_R| \leq |D' \symdif {D'}_R| + 1 = \repair(\constraints, D') + 1
\end{align*}
}

Therefore
{\scriptsize
$$
\sens_{\repair(\constraints, D)} = \max_{D' \neighbor D} \left| \repair(\constraints, D') - \repair(\constraints, D) \right| \leq 1
$$
}
\end{proof}

\begin{proof}[Proof of \Cref{prop:contribution-properties} Parts 1 and 2]
\textbf{Positivity:}

Recall that by \Cref{def:contribution}, given a single fairness criterion $\fairnessdef$, $\contribution$ is an absolute value. Since its aggregate form is a sum $\contribution$ is always non-negative and its value is $0$ only if the databases satisfies all the given fairness criteria in $\constraints$. 

\textbf{Monotonicity:}

Monotonicity follows directly from positivity.
\end{proof}

We will first state and prove the range and the sensitivity bounds of $\contribution$ for a single fairness criterion, and then we will use them to prove the range and the sensitivity bounds for a set of criteria (Parts 3 and 4 of \Cref{prop:contribution-properties}).

\begin{lemma}[$\contribution$ is suitable for DP for a single fairness criterion]\label{prop:contribution-properties-single-criterion}
Given a database $D$ and a parameter $k \in \mathbb{N}$, the following holds:
\begin{enumerate}
    \item The range of $\contribution(\fairnessdef, D)$ is $[0, \min\{\frac{k}{4},2\}]$ for a single criterion $\fairnessdef$.
    \item The sensitivity of $\contribution(\fairnessdef, D)$ is at most $\frac{3k}{n}$ for an unconditional criterion $\fairnessdef$, and, assuming that
$|\{t \in D \mid t[\admissibleset] = \admissiblevalue\}| \ge 2$ for every
$\admissiblevalue \in \admissibleset$, is at most $\frac{7k}{n}$ for a conditional one.
\end{enumerate}
\end{lemma}

\begin{proof}[Proof of \Cref{prop:contribution-properties-single-criterion} Part 1]
By \Cref{def:contribution}, given a single fairness criterion $\fairnessdef$, $\contribution$ is an absolute value. Therefore, from the properties of the absolute value, $\contribution(\fairnessdef, D) \geq 0$.

For every $t=(\protectedvalue, \responsevalue, \admissiblevalue) \in top-k$, define:
{\scriptsize
\begin{align*}
&q_t =\prob(\protectedset=t[\protectedset],\responseset=t[\responseset]\mid \admissibleset=t[\admissibleset]),\\
&p_t =\prob(\protectedset=t[\protectedset]\mid \admissibleset=t[\admissibleset]),\\
&r_t =\prob(\responseset=t[\responseset]\mid \admissibleset=t[\admissibleset]).
\end{align*}
}
Then by the Fr\'echet inequalities~\cite{frechet1935generalisation}, it holds that $\max\{0, p_t + r_t - 1\} \leq q_t \leq \min\{p_t, r_t\}$, and so for the convex function $q_t \rightarrow |q_t - p_{t}r_{t}|$, we get that:
{\scriptsize
$$
|q_t - p_{t}r_{t}| \leq \max\{|\min\{p_t, r_t\} - p_{t}r_{t}|, |\max\{0, p_t + r_t - 1\} - p_{t}r_{t}|\}
$$
}

Without loss of generality, assume that $p_t \leq r_t$ (the other case is symmetric). Then:
{\scriptsize
$$
|\min\{p_t, r_t\} - p_{t}r_{t}| = |p_t - p_{t}r_{t}| = p_t(1 - r_t) \leq r_t(1 - r_t) \leq^{\text{max is at }r_t=\frac{1}{2}} \frac{1}{4} 
$$
}
For the other side there are two cases. If $p_t + r_t < 1$, then:
{\scriptsize
$$
|\max\{0, p_t + r_t - 1\} - p_{t}r_{t}| = p_{t}r_{t} \leq \left(\frac{p_t + r_t}{2}\right)^2 \leq \left(\frac{1}{2}\right)^2 = \frac{1}{4}
$$
}
Otherwise, $p_t + r_t \geq 1$, and:
{\scriptsize
\begin{align*}
|\max\{0, p_t + r_t - 1\} - p_{t}r_{t}| &= |(p_t + r_t - 1) - p_{t}r_{t}| \\
&= |-(1 - p_t)(1 - r_t)| \\
&= (1 - p_t)(1 - r_t) \\
&\leq^{\text{max is at }p_t=r_t=\frac{1}{2}} \frac{1}{4}    
\end{align*}
}

Therefore, we got that $|q_t - p_{t}r_{t}| \leq \max\{\frac{1}{4}, \frac{1}{4}\} = \frac{1}{4}$, and so
{\scriptsize
\begin{align*}
\contribution(\fairnessdef,D)
&= \sum_{t=(\protectedvalue,\responsevalue,\admissiblevalue) \in top-k} \margdiff(\protectedset \independent \responseset \mid \admissibleset,D,t) \\
&= \sum_{t=(\protectedvalue,\responsevalue,\admissiblevalue) \in top-k}
\prob(\admissibleset=t[\admissibleset])\bigl|\prob(\protectedset=t[\protectedset],\responseset=t[\responseset] \mid \admissibleset=t[\admissibleset]) \\
&\quad - \prob(\protectedset=t[\protectedset] \mid \admissibleset=t[\admissibleset])\,\prob(\responseset=t[\responseset] \mid \admissibleset=t[\admissibleset])\bigr| \\
&\leq \sum_{t=(\protectedvalue,\responsevalue,\admissiblevalue) \in top-k}
1 \cdot \bigl|\prob(\protectedset=t[\protectedset],\responseset=t[\responseset] \mid \admissibleset=t[\admissibleset]) \\
&\quad - \prob(\protectedset=t[\protectedset] \mid \admissibleset=t[\admissibleset])\,\prob(\responseset=t[\responseset] \mid \admissibleset=t[\admissibleset])\bigr| \\
&= \sum_{t\in top-k} |q_t - p_{t}r_{t}| \\
&\leq \sum_{t\in top-k} \frac{1}{4} \\
&\leq \frac{k}{4}
\end{align*}
}

Additionally, from the proof of \Cref{prop:tvd-properties-single-criterion}, it holds that
{\scriptsize
\begin{align*}
\contribution(\fairnessdef,D)
&\leq \sum_{t \in top-k} |q_t - p_{t}r_{t}| \\
&\leq \sum_{t \in D} |q_t - p_{t}r_{t}| \\
&= 2 \cdot \mathrm{TVD}(q_t, p_{t}r_{t}) \\
&\leq 2 \cdot 1 \\
&\leq 2
\end{align*}
}

Therefore:
{\scriptsize
$$
\contribution(\fairnessdef,D) \leq \min\{\frac{k}{4},2\}
$$
}
\end{proof}

Now we will state and prove a few lemmas that we will use in a later proof of \Cref{prop:contribution-properties-single-criterion} Part 2.

\begin{lemma}[Sensitivity of the unconditional empirical probabilities]\label{lem:unconditional-empirical-probabilities-sensitivity}
Given a database $D$ of size $n$ such that the schema of $D$ is a superset of $\protectedset \cup \responseset$, the sensitivity of the {\em unconditional empirical probabilities} $\prob(\protectedset,\responseset),\prob(\protectedset),\prob(\responseset)$ in $D$ is at most $\frac{1}{n}$.
\end{lemma}

\begin{proof}
We will prove this claim for the \textbf{joint} empirical probabilities. The same calculation can be done for the marginal empirical probabilities.

Let there be $D \neighbor D'$ neighboring databases that differ by one tuple $t'=(\protectedvalue_{t'},\responsevalue_{t'})$. Assume without loss of generality that $t' \in D'$ and $t' \notin D$.

For any $(\protectedvalue,\responsevalue)$, it holds that
{\scriptsize
$$\prob(\protectedset=\protectedvalue,\responseset=\responsevalue) = \frac{|\{t \in D \mid t[\protectedset]=\protectedvalue,t[\responseset]=\responsevalue\}|}{n}$$
}

For $\probneighbor$ we divide into two cases. For any $(\protectedvalue,\responsevalue) \neq (\protectedvalue_{t'},\responsevalue_{t'})$, it holds that
{\scriptsize
\begin{align*}
\probneighbor(\protectedset=\protectedvalue,\responseset=\responsevalue) &= \frac{|\{t \in D' \mid t[\protectedset]=\protectedvalue,t[\responseset]=\responsevalue\}|}{n} \\
&= \frac{|\{t \in D \mid t[\protectedset]=\protectedvalue,t[\responseset]=\responsevalue\}|}{n} \\
&\leq \frac{|\{t \in D \mid t[\protectedset]=\protectedvalue,t[\responseset]=\responsevalue\}| + 1}{n}    
\end{align*}
}

Otherwise, for $(\protectedvalue,\responsevalue) = (\protectedvalue_{t'},\responsevalue_{t'})$, and
{\scriptsize
\begin{align*}
\prob'(\protectedset=\protectedvalue_{t'},\responseset=\responsevalue_{t'}) &= \frac{|\{t \in D' \mid t[\protectedset] = \protectedvalue_{t'},t[\responseset] = \responsevalue_{t'}\}|}{n} \\
&= \frac{|\{t \in D \mid t[\protectedset] = \protectedvalue_{t'},t[\responseset] = \responsevalue_{t'}\}| + 1}{n}    
\end{align*}
}

Therefore, for any $(\protectedvalue,\responsevalue)$, we obtain that
{\scriptsize
\begin{align*}
&\left|\probneighbor(\protectedset=\protectedvalue,\responseset=\responsevalue) - \prob(\protectedset=\protectedvalue,\responseset=\responsevalue)\right| \\
&\leq \left|\frac{|\{t \in D \mid t[\protectedset]=\protectedvalue,t[\responseset]=\responsevalue\}| + 1}{n} - \frac{|\{t \in D \mid t[\protectedset]=\protectedvalue,t[\responseset]=\responsevalue\}|}{n}\right| \\
&= \frac{1}{n}
\end{align*}
}

And so
{\scriptsize
$$\sens_{\prob} \leq \frac{1}{n}$$
}
\end{proof}

Now we will state and prove a few lemmas that we will use in a later proof of \Cref{prop:contribution-properties-single-criterion} Part 2.

\begin{lemma}[Sensitivity of the conditional empirical probabilities]\label{lem:conditional-empirical-probabilities-sensitivity}
Given a database $D$ of size $n$ such that the schema of $D$ is a superset of
$\protectedset \cup \responseset \cup \admissibleset$, and assuming that
$|\{t \in D \mid t[\admissibleset] = \admissiblevalue\}| \ge 2$ for every
$\admissiblevalue \in \admissibleset$, the sensitivity of the
{\em conditional empirical probabilities}
$\prob(\protectedset,\responseset \mid \admissibleset)$,
$\prob(\protectedset \mid \admissibleset)$,
$\prob(\responseset \mid \admissibleset)$ in $D$ is at most
$\max_{\admissiblevalue}\frac{1}{|\{t \in D \mid t[\admissibleset] = \admissiblevalue\}| - 1}$.
\end{lemma}

\begin{proof}
We will prove this claim for the \textbf{joint} empirical probabilities.
The same calculation can be done for the marginal empirical probabilities.

Let there be $D \neighbor D'$ neighboring databases such that they differ
in one tuple by replacement. Without loss of generality, assume that $D'$
contains one more occurrence of $t''$ and one less occurrence of $t'$ than $D$.

There are three cases. For any $(\protectedvalue,\responsevalue,\admissiblevalue)$
such that $a \notin \{\admissiblevalue_{t'},\admissiblevalue_{t''}\}$, it holds that
{\scriptsize
\begin{align}
&\probneighbor(\protectedset=\protectedvalue,\responseset=\responsevalue \mid \admissibleset=\admissiblevalue)
= \frac{|\{t \in D' \mid t[\protectedset] = \protectedvalue,
t[\responseset]=\responsevalue,t[\admissibleset]=\admissiblevalue\}|}
{|\{t \in D' \mid t[\admissibleset]=\admissiblevalue\}|} \\
&= \frac{|\{t \in D \mid t[\protectedset]=\protectedvalue,
t[\responseset]=\responsevalue,t[\admissibleset]=\admissiblevalue\}|}
{|\{t \in D \mid t[\admissibleset]=\admissiblevalue\}|}
= \prob(\protectedset=\protectedvalue,\responseset=\responsevalue \mid \admissibleset=\admissiblevalue)
\end{align}
}

For any $(\protectedvalue,\responsevalue,\admissiblevalue)$ such that
$a = \admissiblevalue_{t'}$, this bucket of tuples loses one tuple. Therefore
{\scriptsize
$$
\probneighbor(\protectedset=\protectedvalue,\responseset=\responsevalue \mid \admissibleset=\admissiblevalue)
=
\frac{|\{t \in D \mid t[\protectedset] = \protectedvalue,
t[\responseset]=\responsevalue,t[\admissibleset]=\admissiblevalue_{t'}\}| - 1}
{|\{t \in D \mid t[\admissibleset]=\admissiblevalue_{t'}\}| - 1}
$$
}

Similarly, for any $(\protectedvalue,\responsevalue,\admissiblevalue)$ such that
$a = \admissiblevalue_{t''}$, it holds that
{\scriptsize
$$
\probneighbor(\protectedset=\protectedvalue,\responseset=\responsevalue \mid \admissibleset=\admissiblevalue)
=
\frac{|\{t \in D \mid t[\protectedset] = \protectedvalue,
t[\responseset]=\responsevalue,t[\admissibleset]=\admissiblevalue_{t''}\}| + 1}
{|\{t \in D \mid t[\admissibleset]=\admissiblevalue_{t''}\}| + 1}
$$
}

And so, for any $(\protectedvalue,\responsevalue,\admissiblevalue)$, we obtain that
{\scriptsize
\begin{align}
&\left|\probneighbor(\protectedset=\protectedvalue,\responseset=\responsevalue \mid \admissibleset=\admissiblevalue)
- \prob(\protectedset=\protectedvalue,\responseset=\responsevalue \mid \admissibleset=\admissiblevalue)\right| \\
&\leq 
\begin{cases}
\dfrac{1}{|\{t \in D \mid t[\admissibleset]=\admissiblevalue_{t'}\}| - 1},
& \text{if } \admissiblevalue = \admissiblevalue_{t'}, \\[10pt]
\dfrac{1}{|\{t \in D \mid t[\admissibleset]=\admissiblevalue_{t''}\}| + 1},
& \text{if } \admissiblevalue = \admissiblevalue_{t''}, \\[10pt]
0, & \text{otherwise.}
\end{cases}
\end{align}
}

Therefore
{\scriptsize
$$
\sens_{\prob} \leq \max_{\admissiblevalue}\frac{1}{|\{t \in D \mid t[\admissibleset] = \admissiblevalue\}| - 1}
$$
}
\end{proof}

\begin{lemma}\label{lem:marginal-difference-sensitivity}
Given a database $D$ of size $n$ such that the schema of $D$ is a superset of $\protectedset \cup \responseset \cup \admissibleset$, the sensitivity of the {\em marginal difference} is at most $\frac{3}{n}$ in the unconditional case and, assuming that
$|\{t \in D \mid t[\admissibleset] = \admissiblevalue\}| \ge 2$ for every
$\admissiblevalue \in \admissibleset$, is at most $\frac{7}{n}$ in the conditional case.
\end{lemma}

\begin{proof}

Let there be $D \neighbor D'$ neighboring databases such that they differ in one tuple by replacement. Without loss of generality, assume that $D'$ contains one more occurrence of $t''$ and one less occurrence of $t'$ than $D$.

\textbf{Unconditional case:}

For any $(\protectedvalue,\responsevalue)$, it holds that
{\scriptsize
\begin{align}
&\bigl|\margdiff(\protectedset \independent \responseset, D', (\protectedvalue, \responsevalue)) 
   - \margdiff(\protectedset \independent \responseset, D, (\protectedvalue, \responsevalue))\bigr| \\ 
&= \Bigl|\,
    \bigl|\probneighbor(\protectedset=\protectedvalue,\responseset=\responsevalue) 
       - \probneighbor(\protectedset=\protectedvalue)\,
         \probneighbor(\responseset=\responsevalue)\bigr| \nonumber\\
&\quad - \bigl|\prob(\protectedset=\protectedvalue,\responseset=\responsevalue) 
       - \prob(\protectedset=\protectedvalue)\,
         \prob(\responseset=\responsevalue)\bigr|
   \,\Bigr| \label{eq:4_1}\\
&\leq \bigl|\,
    \bigl(\probneighbor(\protectedset=\protectedvalue,\responseset=\responsevalue) 
       - \probneighbor(\protectedset=\protectedvalue)\,
         \probneighbor(\responseset=\responsevalue)\bigr) \nonumber\\
&\quad - \bigl(\prob(\protectedset=\protectedvalue,\responseset=\responsevalue) 
       - \prob(\protectedset=\protectedvalue)\,
         \prob(\responseset=\responsevalue)\bigr)
   \,\bigr| \label{eq:4_2}\\
&= \bigl|
   \bigl(\probneighbor(\protectedset=\protectedvalue,\responseset=\responsevalue) 
       - \prob(\protectedset=\protectedvalue,\responseset=\responsevalue)\bigr) \nonumber\\
&\quad - \bigl(\probneighbor(\protectedset=\protectedvalue)\,
                 \probneighbor(\responseset=\responsevalue) 
       - \prob(\protectedset=\protectedvalue)\,
         \prob(\responseset=\responsevalue)\bigr)
   \bigr| \label{eq:4_3} \\
&= \bigl|
   \bigl[\probneighbor(\protectedset=\protectedvalue,\responseset=\responsevalue) 
        - \prob(\protectedset=\protectedvalue,\responseset=\responsevalue)\bigr] \nonumber\\
&\quad - \bigl[\probneighbor(\protectedset=\protectedvalue)\,
                 (\probneighbor(\responseset=\responsevalue) 
                  - \prob(\responseset=\responsevalue)) \nonumber\\
&\quad\quad 
         + \prob(\responseset=\responsevalue)\,
           (\probneighbor(\protectedset=\protectedvalue) 
            - \prob(\protectedset=\protectedvalue))\bigr]
   \bigr| \label{eq:4_4} \\
&\leq \bigl|\probneighbor(\protectedset=\protectedvalue,\responseset=\responsevalue) 
           - \prob(\protectedset=\protectedvalue,\responseset=\responsevalue)\bigr| \nonumber\\
&\quad + \bigl|\probneighbor(\protectedset=\protectedvalue)\,
                 (\probneighbor(\responseset=\responsevalue) 
                  - \prob(\responseset=\responsevalue))\bigr| \nonumber\\
&\quad + \bigl|\prob(\responseset=\responsevalue)\,
                 (\probneighbor(\protectedset=\protectedvalue) 
                  - \prob(\protectedset=\protectedvalue))\bigr|
   \label{eq:4_5} \\
&\leq \tfrac{1}{n} 
   + \bigl|\probneighbor(\protectedset=\protectedvalue)\bigr| \cdot \tfrac{1}{n} 
   + \bigl|\prob(\responseset=\responsevalue)\bigr| \cdot \tfrac{1}{n} 
   \label{eq:4_6} \\
&\leq \tfrac{3}{n}, \label{eq:4_7}
\end{align}
}
where the transitions are due to:
\begin{itemize}
  \item $\eqref{eq:4_1} \rightarrow \eqref{eq:4_2}$: \Cref{lem:reverse-triangle-ineq}.
  \item $\eqref{eq:4_3} \rightarrow \eqref{eq:4_4}$: $ab - cd = a(b-d) + d(a-c)$.
  \item $\eqref{eq:4_4} \rightarrow \eqref{eq:4_5}$: triangle inequality.
  \item $\eqref{eq:4_5} \rightarrow \eqref{eq:4_6}$: \Cref{lem:unconditional-empirical-probabilities-sensitivity}.
  \item $\eqref{eq:4_6} \rightarrow \eqref{eq:4_7}$:
        $\probneighbor(\protectedset=\protectedvalue)\le 1$, 
        $\prob(\responseset=\responsevalue)\le 1$.
\end{itemize}

And so
{\scriptsize
\setlength{\abovedisplayskip}{4pt}
\setlength{\belowdisplayskip}{4pt}
$$\sens_{\margdiff(\fairnessdef,D,t)} \leq \frac{3}{n}$$
}

\textbf{Conditional case:}

By the definition of conditional $\margdiff$ it holds that
{\scriptsize
\setlength{\abovedisplayskip}{4pt}
\setlength{\belowdisplayskip}{4pt}
\begin{align}
&\left|\margdiff(\protectedset \independent \responseset \mid \admissibleset,D',t)
      - \margdiff(\protectedset \independent \responseset \mid \admissibleset,D,t)\right| \\
&= \Bigl|\,
\probneighbor(\admissibleset=\admissiblevalue)\,
\bigl|\probneighbor(\protectedset=\protectedvalue,\responseset=\responsevalue \mid \admissibleset=\admissiblevalue)
     - \probneighbor(\protectedset=\protectedvalue \mid \admissibleset=\admissiblevalue)\,
       \probneighbor(\responseset=\responsevalue \mid \admissibleset=\admissiblevalue)\bigr| \nonumber\\
&\quad - \prob(\admissibleset=\admissiblevalue)\,
\bigl|\prob(\protectedset=\protectedvalue,\responseset=\responsevalue \mid \admissibleset=\admissiblevalue)
     - \prob(\protectedset=\protectedvalue \mid \admissibleset=\admissiblevalue)\,
       \prob(\responseset=\responsevalue \mid \admissibleset=\admissiblevalue)\bigr|
\,\Bigr| \label{eq:5_1}\\
&\leq \bigl|\,
\probneighbor(\admissibleset=\admissiblevalue)\,
\bigl(\probneighbor(\protectedset=\protectedvalue,\responseset=\responsevalue \mid \admissibleset=\admissiblevalue)
     - \probneighbor(\protectedset=\protectedvalue \mid \admissibleset=\admissiblevalue)\,
       \probneighbor(\responseset=\responsevalue \mid \admissibleset=\admissiblevalue)\bigr) \nonumber\\
&\quad - \prob(\admissibleset=\admissiblevalue)\,
\bigl(\prob(\protectedset=\protectedvalue,\responseset=\responsevalue \mid \admissibleset=\admissiblevalue)
     - \prob(\protectedset=\protectedvalue \mid \admissibleset=\admissiblevalue)\,
       \prob(\responseset=\responsevalue \mid \admissibleset=\admissiblevalue)\bigr)
\,\bigr| \label{eq:5_2}\\
&\leq \bigl|\probneighbor(\admissibleset=\admissiblevalue) - \prob(\admissibleset=\admissiblevalue)\bigr|
\cdot
\bigl|
\probneighbor(\protectedset=\protectedvalue,\responseset=\responsevalue \mid \admissibleset=\admissiblevalue) \nonumber\\
&\quad - \probneighbor(\protectedset=\protectedvalue \mid \admissibleset=\admissiblevalue)\,
        \probneighbor(\responseset=\responsevalue \mid \admissibleset=\admissiblevalue)
\bigr| \nonumber\\
&\quad + \prob(\admissibleset=\admissiblevalue)\cdot
\bigl|
\bigl(\probneighbor(\protectedset=\protectedvalue,\responseset=\responsevalue \mid \admissibleset=\admissiblevalue)
     - \probneighbor(\protectedset=\protectedvalue \mid \admissibleset=\admissiblevalue)\,
       \probneighbor(\responseset=\responsevalue \mid \admissibleset=\admissiblevalue)\bigr) \nonumber\\
&\quad\quad -
\bigl(\prob(\protectedset=\protectedvalue,\responseset=\responsevalue \mid \admissibleset=\admissiblevalue)
     - \prob(\protectedset=\protectedvalue \mid \admissibleset=\admissiblevalue)\,
       \prob(\responseset=\responsevalue \mid \admissibleset=\admissiblevalue)\bigr)
\bigr| \label{eq:5_3},
\end{align}
}
where the transition from \eqref{eq:5_1} to \eqref{eq:5_2} is due to \Cref{lem:reverse-triangle-ineq}, and the transition from \eqref{eq:5_2} to \eqref{eq:5_3} is due to the fact that $|a'b' - ab| \leq |a' - a|\,|b'| + |a|\,|b' - b|$.

We will start by bounding the first summand. From \Cref{lem:unconditional-empirical-probabilities-sensitivity}, $\sens_{\prob(\admissibleset)} = \frac{1}{n}$. In addition, from the proof of \Cref{prop:contribution-properties-single-criterion} Part 1,
for every $(\protectedvalue,\responsevalue,\admissiblevalue)$,
{\scriptsize
$$
\bigl|\prob(\protectedset=\protectedvalue,\responseset=\responsevalue \mid \admissibleset=\admissiblevalue)
- \prob(\protectedset=\protectedvalue \mid \admissibleset=\admissiblevalue)\,
  \prob(\responseset=\responsevalue \mid \admissibleset=\admissiblevalue)\bigr|
\leq \frac{1}{4}
$$
}
Therefore
{\scriptsize
\begin{align*}
&\bigl|\probneighbor(\admissibleset=\admissiblevalue) - \prob(\admissibleset=\admissiblevalue)\bigr|
\cdot
\bigl|
\probneighbor(\protectedset=\protectedvalue,\responseset=\responsevalue \mid \admissibleset=\admissiblevalue) \\
&\quad - \probneighbor(\protectedset=\protectedvalue \mid \admissibleset=\admissiblevalue)\,
        \probneighbor(\responseset=\responsevalue \mid \admissibleset=\admissiblevalue)
\bigr| \\
&\leq \frac{1}{n} \cdot \frac{1}{4} \\
&= \frac{1}{4n}
\end{align*}
}

Now we will bound the second summand. Denote:
{\scriptsize
\begin{align*}
p' &:= \probneighbor(\protectedset=\protectedvalue,\responseset=\responsevalue \mid \admissibleset=\admissiblevalue), \\
q' &:= \probneighbor(\protectedset=\protectedvalue \mid \admissibleset=\admissiblevalue), \\
r' &:= \probneighbor(\responseset=\responsevalue \mid \admissibleset=\admissiblevalue), \\
p  &:= \prob(\protectedset=\protectedvalue,\responseset=\responsevalue \mid \admissibleset=\admissiblevalue), \\
q  &:= \prob(\protectedset=\protectedvalue \mid \admissibleset=\admissiblevalue), \\
r  &:= \prob(\responseset=\responsevalue \mid \admissibleset=\admissiblevalue).
\end{align*}
}
Then it holds that
{\scriptsize
\begin{align}
&\prob(\admissibleset=\admissiblevalue)\cdot
\Bigl|
\bigl(\probneighbor(\protectedset=\protectedvalue,\responseset=\responsevalue \mid \admissibleset=\admissiblevalue)
     - \probneighbor(\protectedset=\protectedvalue \mid \admissibleset=\admissiblevalue)\,
       \probneighbor(\responseset=\responsevalue \mid \admissibleset=\admissiblevalue)\bigr) \\
&\quad\quad -
\bigl(\prob(\protectedset=\protectedvalue,\responseset=\responsevalue \mid \admissibleset=\admissiblevalue)
     - \prob(\protectedset=\protectedvalue \mid \admissibleset=\admissiblevalue)\,
       \prob(\responseset=\responsevalue \mid \admissibleset=\admissiblevalue)\bigr)
\Bigr| \\
&= \prob(\admissibleset=\admissiblevalue) \cdot \Bigl| (p' - q'r') - (p - qr) \Bigr| \\
&=
\prob(\admissibleset=\admissiblevalue) \cdot \Bigl| (p' - p) - (q'r' - qr) \Bigr| \label{eq:6_1}\\
&\leq \prob(\admissibleset=\admissiblevalue) \cdot \left(|p' - p| + |q'r' - qr|\right) \label{eq:6_2}\\
&\leq \prob(\admissibleset=\admissiblevalue) \cdot \left(|p' - p| + |q'(r' - r) + r(q' - q)|\right) \label{eq:6_3}\\
&\leq \prob(\admissibleset=\admissiblevalue) \cdot \left(|p' - p| + |q'(r' - r)| + |r(q' - q)|\right) \label{eq:6_4}\\
&\leq \prob(\admissibleset=\admissiblevalue) \cdot \left(|p' - p| + |r' - r| + |q' - q|\right) \label{eq:6_5}\\
&\leq \prob(\admissibleset=\admissiblevalue) \cdot 3 \cdot \frac{1}{|\{t \in D \mid t[\admissibleset] = \admissiblevalue\}| - 1} \label{eq:6_6}\\
&= \frac{|\{t \in D \mid t[\admissibleset] = \admissiblevalue\}|}{n} \cdot \frac{3}{|\{t \in D \mid t[\admissibleset] = \admissiblevalue\}| - 1} \\
&= \frac{3}{n}\cdot \frac{|\{t \in D \mid t[\admissibleset] = \admissiblevalue\}|}{|\{t \in D \mid t[\admissibleset] = \admissiblevalue\}| - 1} \label{eq:6_7}\\
&\leq \frac{3}{n}\cdot 2 \label{eq:6_8}\\
&= \frac{6}{n},
\end{align}
}
where the transitions are due to:
\begin{itemize}
  \item $\eqref{eq:6_1} \rightarrow \eqref{eq:6_2}$: triangle inequality.
  \item $\eqref{eq:6_2} \rightarrow \eqref{eq:6_3}$: $ab-cd = a(b-d) + d(a-c)$.
  \item $\eqref{eq:6_3} \rightarrow \eqref{eq:6_4}$: triangle inequality.
  \item $\eqref{eq:6_4} \rightarrow \eqref{eq:6_5}$: $q' \leq 1,\, r \leq 1$.
  \item $\eqref{eq:6_5} \rightarrow \eqref{eq:6_6}$: \Cref{lem:conditional-empirical-probabilities-sensitivity}.
  \item $\eqref{eq:6_7} \rightarrow \eqref{eq:6_8}$: by the assumption, $|\{t \in D \mid t[\admissibleset] = \admissiblevalue\}|\ge 2$.
\end{itemize}

Combining the bounds for the two terms yields
{\scriptsize
$$
\left|\margdiff(\protectedset \independent \responseset \mid \admissibleset,D',t)
      - \margdiff(\protectedset \independent \responseset \mid \admissibleset,D,t)\right|
\leq \frac{1}{4n} + \frac{6}{n}
\leq \frac{7}{n}.
$$
}
\end{proof}

\begin{proof}[Proof of \Cref{prop:contribution-properties-single-criterion} Part 2]

\textbf{Unconditional case:}

Let there be $D \neighbor D'$ neighboring databases such that they differ in one tuple by replacement. Without loss of generality, assume that $D'$ contains one more occurrence of $t'' = (\protectedvalue_{t''},\responsevalue_{t''})$ and one less occurrence of $t' = (\protectedvalue_{t'},\responsevalue_{t'})$ than $D$.

We extend $\margdiff$ to all tuples (even those not observed in $D$) as follows: if $t \notin D$ (i.e., the projection of $t$ onto the attributes in $\fairnessdef$ is not observed in $D$), then $\margdiff(\fairnessdef, D, t) = 0$.

By the definition, $top-k$ is a set of size $k$ that maximizes
$\sum_{t\in S}\margdiff(\fairnessdef,D,t)$ over all sets $S$ such that $|S|=k$. In particular,
{\scriptsize
$$
\sum_{t\in top-k}\margdiff(\fairnessdef,D,t) \geq \sum_{t\in top-k'}\margdiff(\fairnessdef,D,t)
$$
}
, and similarly,
{\scriptsize
$$
\sum_{t\in top-k'}\margdiff(\fairnessdef,D',t) \geq \sum_{t\in top-k}\margdiff(\fairnessdef,D',t).
$$
}
This is also true for the extended definition of $\margdiff$, since an unobserved tuple in any set can be swapped by an observed one (nonnegative by the definition), thereby increasing the total $\margdiff$ of the set.

Assume without loss of generality that
$\sum_{t \in top-k'} \margdiff(\fairnessdef,D',t) \geq \sum_{t \in top-k} \margdiff(\fairnessdef,D,t)$
(the other case is symmetric). We now distinguish two cases. If {\scriptsize $\sum_{t\in top-k}\margdiff(\fairnessdef,D',t) \geq \sum_{t\in top-k}\margdiff(\fairnessdef,D,t)$}, then
{\scriptsize
\begin{align*}
&\bigl|\sum_{t \in top-k'} \margdiff(\fairnessdef,D',t) - \sum_{t \in top-k} \margdiff(\fairnessdef,D,t)\bigr| \\
&= \sum_{t \in top-k'} \margdiff(\fairnessdef,D',t) - \sum_{t \in top-k} \margdiff(\fairnessdef,D,t) \\
&\leq \sum_{t \in top-k'} \margdiff(\fairnessdef,D',t) - \sum_{t \in top-k} \margdiff(\fairnessdef,D',t) \\
&\quad + \sum_{t \in top-k} \margdiff(\fairnessdef,D',t) - \sum_{t \in top-k} \margdiff(\fairnessdef,D,t) \\
&\le 0 + \sum_{t \in top-k}\bigl(\margdiff(\fairnessdef,D',t) - \margdiff(\fairnessdef,D,t)\bigr) \\
&\le \sum_{t \in top-k}\bigl|\margdiff(\fairnessdef,D',t) - \margdiff(\fairnessdef,D,t)\bigr|,
\end{align*}
}

Otherwise, {\scriptsize $\sum_{t\in top-k}\margdiff(\fairnessdef,D',t) < \sum_{t\in top-k}\margdiff(\fairnessdef,D,t)$}, and so
{\scriptsize
\begin{align*}
&\bigl|\sum_{t \in top-k'} \margdiff(\fairnessdef,D',t) - \sum_{t \in top-k} \margdiff(\fairnessdef,D,t)\bigr| \\
&= \sum_{t \in top-k'} \margdiff(\fairnessdef,D',t) - \sum_{t \in top-k} \margdiff(\fairnessdef,D,t) \\
&\leq \sum_{t \in top-k'}\bigl(\margdiff(\fairnessdef,D',t) - \margdiff(\fairnessdef,D,t)\bigr) \\
&\leq \sum_{t \in top-k'}\bigl|\margdiff(\fairnessdef,D',t) - \margdiff(\fairnessdef,D,t)\bigr|
\end{align*}
}

Combining the two cases yields
{\scriptsize
\begin{align*}
&\bigl|\sum_{t \in top-k'} \margdiff(\fairnessdef,D',t) - \sum_{t \in top-k} \margdiff(\fairnessdef,D,t)\bigr| \\
&\leq
\max\Bigl\{
\sum_{t\in top-k}\left|\margdiff(\fairnessdef,D',t)-\margdiff(\fairnessdef,D,t)\right|,\\
&\quad\quad\quad \sum_{t\in top-k'}\left|\margdiff(\fairnessdef,D',t)-\margdiff(\fairnessdef,D,t)\right|
\Bigr\}    
\end{align*}
}

{\scriptsize
\setlength{\abovedisplayskip}{4pt}
\setlength{\belowdisplayskip}{4pt}
\begin{align*}
&\bigl|\contribution(\fairnessdef, D') - \contribution(\fairnessdef, D)\bigr| \\ 
&= \bigl|\sum_{t \in top-k'} \margdiff(\fairnessdef,D',t) - \sum_{t \in top-k} \margdiff(\fairnessdef,D,t)\bigr| \\
&\leq \max\Bigl\{\sum_{t \in top-k'} \bigl|\margdiff(\fairnessdef,D',t) - \margdiff(\fairnessdef,D,t)\bigr|,\\
&\quad\quad\quad \sum_{t \in top-k} \bigl|\margdiff(\fairnessdef,D',t) - \margdiff(\fairnessdef,D,t)\bigr|\Bigr\} \\
&\leq k \cdot \sens_{\margdiff(\fairnessdef,D,t)} \\
&\leq k \cdot \frac{3}{n} \\
&= \frac{3k}{n},
\end{align*}
}
where $\sens_{\margdiff(\fairnessdef,D,t)} \leq \frac{3}{n}$ by \Cref{lem:marginal-difference-sensitivity}.

\textbf{Conditional case:}

Similarly, assume that $D'$ contains one more occurrence of $t'' = (\protectedvalue_{t''},\responsevalue_{t''},\admissiblevalue_{t''})$ and one less occurrence of $t' = (\protectedvalue_{t'},\responsevalue_{t'},\admissiblevalue_{t'})$ than $D$.

Then, similarly to the unconditional case,
{\scriptsize
\setlength{\abovedisplayskip}{4pt}
\setlength{\belowdisplayskip}{4pt}
\begin{align*}
&\bigl|\contribution(\fairnessdef, D') - \contribution(\fairnessdef, D)\bigr| \\ 
&= \bigl|\sum_{t \in top-k'} \margdiff(\fairnessdef,D',t) - \sum_{t \in top-k} \margdiff(\fairnessdef,D,t)\bigr| \\
&\leq \max\Bigl\{\sum_{t \in top-k'} \bigl|\margdiff(\fairnessdef,D',t) - \margdiff(\fairnessdef,D,t)\bigr|,\\
&\quad\quad\quad \sum_{t \in top-k} \bigl|\margdiff(\fairnessdef,D',t) - \margdiff(\fairnessdef,D,t)\bigr|\Bigr\} \\
&\leq k \cdot \sens_{\margdiff(\fairnessdef,D,t)} \\
&\leq k \cdot \frac{7}{n} \\
&= \frac{7k}{n},
\end{align*}
}
where $\sens_{\margdiff(\fairnessdef,D,t)} \leq \frac{7}{n}$ by \Cref{lem:marginal-difference-sensitivity}.
\end{proof}

The proof of \Cref{prop:contribution-properties} Parts 3 and 4 is similar to the respective parts of the proof of \Cref{prop:tvd-properties}.

\subsection{Proofs for \Cref{sec:algorithms}}

\begin{proof}[Proof of \Cref{lem:repair-cost-dominance}]
Let $D_{\alpha_i}$ be the set of tuples $t$ such that $x_t = \texttt{True}$ under assignment $\alpha_i$.

Recall that the repair cost is defined as the symmetric difference between $D$ and $D_{\alpha_i}$
{\scriptsize
$$
\repair(\varphi(D,\jdb), \alpha_i) = |D \setminus D_{\alpha_i}| + |D_{\alpha_i} \setminus D|
$$
}
This expression counts deletions (tuples in $D$ but not in $D_{\alpha_i}$) and insertions (tuples in $D_{\alpha_i}$ but not in $D$). Now, observe that:
\begin{itemize}
    \item Each satisfied soft clause of the form $x_t$ (with $t \in D$) means that $t$ is retained in $D_{\alpha_i}$ (i.e., not deleted). 
    \item Each satisfied soft clause of the form $\neg x_t$ (with $t \in \jdb \setminus D$) means that $t$ is not inserted into $D_{\alpha_i}$.
\end{itemize}

Thus, satisfying more soft clauses corresponds to performing fewer changes to the original database. Specifically, more original tuples are kept, and fewer new tuples are added. Hence, the symmetric difference $|D \setminus D_{\alpha_i}| + |D_{\alpha_i} \setminus D|$ is smaller.

Therefore, if $\alpha_1$ satisfies more soft clauses than $\alpha_2$, it must result in a smaller repair cost by the definition of $\repair$. That is

{\scriptsize
$$
\Delta(\varphi(D,\jdb), \alpha_1) < \Delta(\varphi(D,\jdb), \alpha_2)
$$
}
\end{proof}

\begin{proposition}[Expressing $\repairsat$ in terms of \Cref{def:cnf}]\label{prop:repair-sat-cnf}
Let $D$ be a database, let $\fairnessdef$ be a fairness criterion, and let $\varphi(D,\jdb)$ be the CNF from \Cref{def:cnf}. Denote by $\mathcal{H}(\jdb)$ and $\mathcal{S}(\jdb)$ the sets of hard and soft clauses in $\varphi(D,\jdb)$, respectively. The following holds.
$$
\repairsat(\fairnessdef,D) = \left|\mathcal{S}(\jdb)\right| - \min_{\alpha \models \mathcal{H}(\jdb)}
|\{\text{soft clauses not satisfied by }\alpha\}|
$$
\end{proposition}
\begin{proof}
By \Cref{def:cnf}, each tuple $t \in D$ contributes a soft clause $x_t$, 
and each tuple $t \in \jdb \setminus D$ contributes a soft clause $\neg x_t$.

For any feasible assignment $\alpha$, let $D_R = \{t \in \jdb \mid \alpha(x_t) = \texttt{True}\}$. Then, since soft clauses are a multiset, a soft clause is not satisfied iff the respective tuple appears in $D$ and doesn't appear in $D_R$. That is, the number of satisfied soft clauses is
{\footnotesize
$$\#\{\text{soft clauses satisfied by }\alpha\} = \left|\mathcal{S}(\jdb)\right| - \Delta(\varphi(D,\jdb),\alpha)$$
}
Therefore, minimizing $\Delta(\varphi(D,\jdb),\alpha)$ is equivalent to maximizing the number of satisfied soft clauses. Thus, the two definitions are also equivalent.
\end{proof}

Following are definitions and propositions with proofs for the full $\repairsat$ sensitivity analysis.

\begin{definition}[Assignment extension]\label{def:extension}
Given a database $D$ and a CNF formula $\varphi(D,\jdb) = \mathcal{H}(\jdb) \land \bigwedge_{t \in D} x_t \land \bigwedge_{t \in \jdb - D} (\neg x_t)$, an assignment $\alpha$ for $\varphi(D,\jdb)$, and a CNF formula $\varphi' = \mathcal{H}(\crossdb) \land \bigwedge_{t \in D} x_t \land \bigwedge_{t \in \crossdb - D} (\neg x_t)$ such that $\jdb \subseteq \crossdb$, an extension of the assignment $\alpha$ for $\varphi(D,\crossdb)$ is defined as follows: 
$$
\alpha'(x_t) =
\begin{cases}
\alpha(x_t) & \text{if } t \in \jdb \\
\texttt{False} & \text{if } t \in \crossdb \setminus \jdb
\end{cases}
$$
\end{definition}

With this definition, we show that a minimum repair of the join and cross-product databases by their CNF formulae have identical size. We then prove that extending a CNF formula for the join database to a CNF formula for the cross-product database preserves hard clauses satisfaction.

\begin{lemma}\label{lem:regular-and-simplified-repair-cost-equivalence}
Given a database $D$ and a fairness criterion, let $\jdb$ and $\crossdb$ be the self-join and cross-product databases. Let $\varphi(D,\jdb)$ and $\varphi(D,\crossdb)$ be two CNF formulae defined according to \Cref{def:cnf}. Let $\alpha$ be an optimal assignment for $\varphi(D,\jdb)$, and let $\hat{\alpha}$ be its extension to the assignment for $\varphi(D,\crossdb)$ according to \Cref{def:extension}. 
It holds that
$
\repairsat(\varphi(D,\jdb)) = \repairsat(\varphi(D,\crossdb))
$.
\end{lemma}

\begin{proof}
Let $\alpha$ be an optimal assignment for $\varphi(D,\jdb)$, and let $\hat{\alpha}$ be its extension to $\varphi(D,\hat{D})$, as defined in \Cref{def:extension}.

Since $\alpha$ is an optimal assignment, it satisfies all hard clauses in $\varphi(D,\jdb)$. Therefore, by \Cref{lem:hard-clauses-preservation}, $\hat{\alpha}$ satisfies all hard clauses in $\varphi(D,\hat{D})$. Additionally, by \Cref{lem:soft-clauses-dominance}, the number of satisfied soft clauses in $\varphi(D,\hat{D})$ under $\hat{\alpha}$ is at least that in $\varphi(D,\jdb)$ under $\alpha$. Hence, $\hat{\alpha}$ is a feasible and possibly optimal assignment for $\varphi(D,\hat{D})$.

Therefore, by \Cref{lem:repair-cost-dominance} it holds that
{\scriptsize
$$
\Delta(\varphi(D,\hat{D}), \hat{\alpha}) \leq \Delta(\varphi(D,\jdb), \alpha)
$$
}

Now let $\hat{\alpha}^*$ be an optimal assignment for $\varphi(D,\hat{D})$, and let $\alpha^*$ be its restriction to $\jdb$ defined as follows
{\scriptsize
$$
\alpha^*(x_t) = \hat{\alpha}^*(x_t) \text{ for } t \in \jdb
$$
}

Since $\jdb \subseteq \hat{D}$, then $\mathcal{H}(\jdb) \subseteq \mathcal{H}(\hat{D})$. And because hard clause satisfaction is preserved under restriction (\Cref{lem:hard-clauses-preservation}), it follows that $\alpha^*$ is a feasible assignment for $\varphi(D,\jdb)$.

Moreover, the reverse direction of \Cref{lem:soft-clauses-dominance} implies that the number of soft clauses satisfied under $\alpha^*$ is at least that in $\varphi(D,\jdb)$ under any other assignment. This means that $\alpha^*$ is an optimal assignment $\alpha$ for $\varphi(D,\jdb)$, and $\hat{\alpha}$ is an optimal assignment $\hat{\alpha}^*$ for $\varphi(\hat{D})$ by \Cref{def:extension}.

Therefore, by \Cref{lem:repair-cost-dominance} it holds that
{\scriptsize
$$
\Delta(\varphi(D,\jdb), \alpha) = \Delta(\varphi(D,\jdb), \alpha^*) \leq \Delta(\varphi(D,\hat{D}), \hat{\alpha}) = \Delta(\varphi(D,\hat{D}), \hat{\alpha}^*)
$$
}

So from \Cref{def:repair-sat-proxy} and from both directions we proved, it holds that
{\scriptsize
\begin{align*}
&\repairsat(\fairnessdef,D) = \Delta(\varphi(D,\jdb), \alpha) \\
&= \Delta(\varphi(D,\hat{D}), \hat{\alpha}) = \min_{\alpha \models \mathcal{H}(\hat{D})} \Delta(\varphi(D,\hat{D}), \alpha)
\end{align*}
}
\end{proof}

We will now prove that extending a CNF formula for the join database to a CNF formula for the cross-product database preserves hard clauses satisfaction.

\begin{lemma}\label{lem:hard-clauses-preservation}
Given a database $D$ such that the schema of $D$ is a superset of $\protectedset \cup \responseset \cup \admissibleset$, and a fairness criterion of the form $\fairnessdef = \protectedset \independent \responseset \mid \admissibleset$, let $\jdb$ be the join database and $\crossdb$ be cross-product database, and let $\varphi(D,\jdb)$ and $\varphi(D,\crossdb)$ be the corresponding two CNF formulae from \Cref{def:cnf}. Let $\alpha$ be a feasible assignment for $\varphi(D,\jdb)$ and let $\hat{\alpha}$ be its extension to the assignment for $\varphi(D,\crossdb)$ according to \Cref{def:extension}. Then, $\mathcal{H}(\jdb)$ in $\varphi(D,\jdb)$ are satisfied by $\alpha$ iff $\mathcal{H}(\hat{\jdb})$ in $\varphi(D,\crossdb)$ are satisfied by $\hat{\alpha}$.   
\end{lemma}

\begin{proof}
\textbf{($\Rightarrow$)} Suppose that $\mathcal{H}(\jdb)$ in $\varphi(D,\jdb)$ are satisfied under $\alpha$. Since $\hat{D}$ is a superset of $\jdb$, then each clause in $\mathcal{H}(\jdb)$ appears in $\mathcal{H}(\hat{\jdb})$, and they differ only in additional hard clauses added because of the tuples in $\hat{D} \setminus \jdb$. Since $\hat{\alpha}$ agrees with $\alpha$ on all $t \in \jdb$, the same clauses also evaluate to \texttt{True} under $\hat{\alpha}$.

Now consider the additional clauses introduced due to tuples in $\hat{D} \setminus \jdb$. These clauses are of the form $(\neg x_{t_1} \lor \neg x_{t_2} \lor x_{t_3})$, where at least one of the tuples $t_1, t_2, t_3$ is in $\hat{D} \setminus \jdb$. Since $\hat{\alpha}$ assigns \texttt{False} to all $x_t$ so that $t \in \hat{D} \setminus \jdb$, the clause \( (\neg x_{t_1} \lor \neg x_{t_2} \lor x_{t_3}) \) will evaluate to \texttt{True} as follows
\begin{itemize}
    \item If either $x_{t_1}$ or $x_{t_2}$ is \texttt{False} under $\hat{\alpha}$, their negation is \texttt{True}, so the whole clause is satisfied by $\hat{\alpha}$.
    \item If both $x_{t_1}$ and $x_{t_2}$ are \texttt{True} under $\hat{\alpha}$, this is possible only if both $t_1$ and $t_2$ are in $\jdb$.
    
    According to the construction of the CNF formula (see \Cref{def:cnf}), a clause of the form
    $$
    (\neg x_{t_1} \lor \neg x_{t_2} \lor x_{t_3})
    $$
    is included in the hard clauses, where $t_1 = (\protectedset_1, \responseset_1, \admissibleset) \in \jdb$, $t_2 = (\protectedset_2, \responseset_2, \admissibleset) \in \jdb$, and $t_3 = (\protectedset_1, \responseset_2, \admissibleset)$ is the tuple resulting from the join of $t_1$ and $t_2$ on the shared attribute $\admissibleset$.

    Since both $t_1$ and $t_2$ belong to $\jdb$, and the join $\jdb(\protectedset_1, \responseset_2, \admissibleset) := D(\protectedset_1, \responseset_1, \admissibleset) \bowtie D(\protectedset_2, \responseset_2, \admissibleset)$ preserves $\admissibleset$, the resulting tuple $t_3 = (\protectedset_1, \responseset_2, \admissibleset)$ must also belong to $\jdb$. Therefore, the clause $(\neg x_{t_1} \lor \neg x_{t_2} \lor x_{t_3})$ belongs to $\mathcal{H}(\jdb)$, and since $\hat{\alpha}$ agrees with $\alpha$ on all tuples in $\jdb$, it satisfies this clause as well. Then the full clause is in $\mathcal{H}(\jdb)$ and should be satisfied by $\hat{\alpha}$ because it agrees with $\alpha$ on all $t \in \jdb$.
\end{itemize}

Therefore, in any case $\mathcal{H}(\hat{D})$ in $\varphi(D,\hat{D})$ are satisfied under $\hat{\alpha}$.

\textbf{($\Leftarrow$)} Suppose $\mathcal{H}(\hat{D})$ in $\varphi(D,\hat{D})$ are satisfied under $\hat{\alpha}$. In particular, all clauses in $\mathcal{H}(\jdb)$, which are a subset of $\mathcal{H}(\hat{D})$, are satisfied under $\hat{\alpha}$. Since $\hat{\alpha}$ agrees with $\alpha$ on all $x_t$ for $t \in \jdb$, then $\mathcal{H}(\jdb)$ in $\varphi(D,\jdb)$ are satisfied also under $\alpha$.

From the two directions, it holds that $H(\jdb)$ in $\varphi(D,\jdb)$ are satisfied under $\alpha$ iff $H(\hat{D})$ in $\varphi(D,\hat{D})$ are satisfied under $\hat{\alpha}$.
\end{proof}

Finally, we show that extending a CNF formula for the self-join database to a CNF formula for the cross-product database cannot decrease the number of satisfied soft clauses.


\begin{lemma}\label{lem:soft-clauses-dominance}
Given a database $D$ and a fairness criterion, let $\jdb$ and $\crossdb$ be the join and cross-product databases. Let $\varphi(D,\jdb)$ and $\varphi(D,\crossdb)$ be the two corresponding CNF formulae according to \Cref{def:cnf}. Let $\alpha$ be a feasible assignment for $\varphi(D,\jdb)$ and let $\hat{\alpha}$ be its extension to the assignment for $\varphi(D,\crossdb)$. Then the number of satisfied soft clauses in $\varphi(D,\jdb)$ by $\alpha$ is lower than or equal to the number of satisfied soft clauses in $\varphi(D,\crossdb)$ under $\hat{\alpha}$.
\end{lemma}

\begin{proof}
Let $D$ be a dataset. According to \Cref{def:cnf}:
{\scriptsize
\begin{align*}
&\varphi(D,\jdb) = \mathcal{H}(\jdb) \land \bigwedge_{t \in D} x_t \land \bigwedge_{t \in \jdb \setminus D} \neg x_t \\
&\varphi(D,\hat{D}) = \mathcal{H}(\hat{D}) \land \bigwedge_{t \in D} x_t \land \bigwedge_{t \in \hat{D} \setminus D} \neg x_t
\end{align*}
}

Since $\jdb \subseteq \hat{D}$, then $\varphi(D,\jdb)$ and $\varphi(D,\hat{D})$ differ only in soft clauses for the tuples in $\hat{D} \setminus \jdb$. In particular, the only difference is that $\varphi(D,\hat{D})$ contains additional soft clauses of the form $\neg x_t$ for $t \in \hat{D} \setminus \jdb$.

Let $\alpha$ be an assignment for $\varphi(D,\jdb)$. According to \Cref{def:extension}, it holds that
\begin{itemize}
    \item All soft clauses in $\varphi(D,\jdb)$ are also present in $\varphi(D,\hat{D})$ and are satisfied under $\hat{\alpha}$ exactly as in $\alpha$.
    \item The additional soft clauses in $\varphi(D,\hat{D})$ are of the form $\neg x_t$ for $t \in \hat{D} \setminus \jdb$, and each such clause is satisfied in $\hat{\alpha}$ because $\hat{\alpha}(x_t) = \texttt{False}$.
\end{itemize}

Hence, $\hat{\alpha}$ satisfies all the soft clauses that $\alpha$ satisfies, plus potentially some new ones. Therefore, the number of satisfied soft clauses in $\varphi(D,\jdb)$ under $\alpha$ is less or equal to the number of satisfied soft clauses in $\varphi(D,\hat{D})$ under $\hat{\alpha}$.
\end{proof}

As a precursor to the sensitivity analysis, we show that the difference in repair between neighboring databases is bounded by $2$. 

\begin{lemma}\label{lem:simplified-repair-cost-sensitivity}
Let $D \neighbor D'$, let $\crossdb$ and $\crossdb'$ be their corresponding cross-product databases, and let $\fairnessdef = \protectedset \independent \responseset \mid \admissibleset$ be a fairness criterion. Let $\varphi(D,\crossdb)$ and $\varphi(D',\crossdb')$ be CNF formulae defined according to \Cref{def:cnf}, and let $\hat{\alpha}$ be a feasible assignment for $\varphi(D,\crossdb)$. Finally, let $\Delta(\varphi(D,\crossdb), \hat{\alpha})$ and $\Delta(\varphi(D',\crossdb'), \hat{\alpha})$ be defined according to \Cref{def:repair-cost-for-assignment}. 
It holds that:
$$
|\Delta(\varphi(D,\crossdb), \hat{\alpha}) - \Delta(\varphi(D',\crossdb'), \hat{\alpha})| \leq 2
$$
\end{lemma}

\begin{proof}
Let there be $D \neighbor D'$ neighboring databases such that they differ in one tuple by replacement. Without loss of generality, assume that $D'$ contains one more occurrence of $t''$ and one less occurrence of $t'$ than $D$.

It follows that the only difference between $\varphi(D,\hat{D})$ and $\varphi(D',\hat{D'})$ lies at most in two soft clauses: 

\begin{itemize}
    \item If $t' \notin D'$, then $\varphi(D',\hat{D'})$ contains the soft clause $\neg x_{t'}$, while $\varphi(D,\hat{D})$ contains the soft clause $x_{t'}$.
    \item If $t'' \notin D$, then $\varphi(D',\hat{D'})$ contains the soft clause $x_{t''}$, while $\varphi(D,\hat{D})$ contains the soft clause $\neg x_{t''}$.
\end{itemize}

The set of hard clauses remains the same, as they are defined over $\hat{D}$, which depends only on the attribute domains. And since there is no difference in the hard clauses between $D$ and $D'$, then $\hat{\alpha}$ is also a feasible assignment for $\varphi(D',\hat{D'})$.

Now, observe that under $\hat{\alpha}$

\begin{table*}[t]
\centering
\caption{Results of exploratory queries $q_1$--$q_6$ on the \texttt{Adult} and \texttt{Compas} datasets with different privacy budgets.}
\label{tab:query-results}
\small
\setlength{\tabcolsep}{3pt}
\renewcommand{\arraystretch}{0.95}

\begin{minipage}[t]{0.66\linewidth}
\vspace{0pt}
\centering
\textbf{(a) \texttt{Compas} dataset}

\vspace{0.4em}

\begin{minipage}[t]{0.48\linewidth}
\vspace{0pt}
\centering
\resizebox{\linewidth}{!}{%
\begin{tabular}{l|ccc}
\toprule
\multicolumn{4}{c}{$q_1$: average \texttt{decile\_score} by \texttt{age\_cat}} \\
\multicolumn{4}{c}{(grouped by \texttt{race})} \\
\midrule
& \multicolumn{3}{c}{\textbf{avg-d}} \\
\cmidrule(lr){2-4}
\textbf{age\_cat} & $\boldsymbol{\varepsilon=\infty}$ & $\boldsymbol{\varepsilon=10}$ & $\boldsymbol{\varepsilon=1}$ \\
25--45           & 3.788 & 3.766 & 3.765 \\
Greater than 45  & 2.662 & 2.638 & 2.637 \\
Less than 25     & 5.279 & 5.202 & 5.202 \\
\bottomrule
\end{tabular}}
\end{minipage}
\hfill
\begin{minipage}[t]{0.48\linewidth}
\vspace{0pt}
\centering
\resizebox{\linewidth}{!}{%
\begin{tabular}{l|ccc}
\toprule
\multicolumn{4}{c}{$q_2$: average \texttt{decile\_score} by \texttt{race}} \\
\multicolumn{4}{c}{(grouped by \texttt{age\_cat})} \\
\midrule
& \multicolumn{3}{c}{\textbf{avg-d}} \\
\cmidrule(lr){2-4}
\textbf{race} & $\boldsymbol{\varepsilon=\infty}$ & $\boldsymbol{\varepsilon=10}$ & $\boldsymbol{\varepsilon=1}$ \\
African-American & 5.205 & 5.183 & 5.189 \\
Asian            & 3.058 & 2.989 & 2.965 \\
Caucasian        & 3.882 & 3.837 & 3.837 \\
Hispanic         & 3.753 & 3.467 & 3.476 \\
Native American  & 4.721 & 4.704 & 5.002 \\
Other            & 2.937 & 2.970 & 2.969 \\
\bottomrule
\end{tabular}}
\end{minipage}

\vspace{0.8em}

\begin{minipage}[t]{0.48\linewidth}
\vspace{0pt}
\centering
\resizebox{\linewidth}{!}{%
\begin{tabular}{l|ccc}
\toprule
\multicolumn{4}{c}{$q_3$: median \texttt{decile\_score} by \texttt{race}} \\
\multicolumn{4}{c}{(grouped by \texttt{c\_charge\_degree})} \\
\midrule
& \multicolumn{3}{c}{\textbf{med-d}} \\
\cmidrule(lr){2-4}
\textbf{race} & $\boldsymbol{\varepsilon=\infty}$ & $\boldsymbol{\varepsilon=10}$ & $\boldsymbol{\varepsilon=1}$ \\
African-American & 2.832 & 4.962 & 8.212 \\
Asian            & 2.202 & 1.887 & -3.577 \\
Caucasian        & 1.253 & 1.549 & 29.933 \\
Hispanic         & 1.626 & 2.042 & -7.642 \\
Native American  & 4.321 & 5.611 & 9.217 \\
Other            & 0.752 & 1.779 & 14.010 \\
\bottomrule
\end{tabular}}
\end{minipage}
\hfill
\begin{minipage}[t]{0.48\linewidth}
\vspace{0pt}
\centering
\resizebox{\linewidth}{!}{%
\begin{tabular}{l|ccc}
\toprule
\multicolumn{4}{c}{$q_4$: median \texttt{decile\_score} by \texttt{c\_charge\_degree}} \\
\multicolumn{4}{c}{(grouped by \texttt{race})} \\
\midrule
& \multicolumn{3}{c}{\textbf{med-d}} \\
\cmidrule(lr){2-4}
\textbf{c\_charge\_degree} & $\boldsymbol{\varepsilon=\infty}$ & $\boldsymbol{\varepsilon=10}$ & $\boldsymbol{\varepsilon=1}$ \\
F & 3.440 & 3.192 & 25.401 \\
M & 1.775 & 2.346 & 14.010 \\
O & 2.034 & 2.743 & -6.324 \\
\bottomrule
\end{tabular}}
\end{minipage}
\end{minipage}
\hfill
\begin{minipage}[t]{0.31\linewidth}
\vspace{0pt}
\centering
\textbf{(b) \texttt{Adult} dataset}

\vspace{0.4em}

\resizebox{\linewidth}{!}{%
\begin{tabular}{l|ccc}
\toprule
\multicolumn{4}{c}{$q_5$: average \texttt{income>50K} by \texttt{sex}} \\
\midrule
& \multicolumn{3}{c}{\textbf{avg-i}} \\
\cmidrule(lr){2-4}
\textbf{sex} & $\boldsymbol{\varepsilon=\infty}$ & $\boldsymbol{\varepsilon=10}$ & $\boldsymbol{\varepsilon=1}$ \\
Female & 0.109 & 0.109 & 0.110 \\
Male   & 0.304 & 0.304 & 0.304 \\
\midrule
\multicolumn{4}{c}{$q_6$: average \texttt{income>50K} by \texttt{race}} \\
\midrule
& \multicolumn{3}{c}{\textbf{avg-i}} \\
\cmidrule(lr){2-4}
\textbf{race} & $\boldsymbol{\varepsilon=\infty}$ & $\boldsymbol{\varepsilon=10}$ & $\boldsymbol{\varepsilon=1}$ \\
Amer-Indian-Eskimo & 0.117 & 0.120 & 0.118 \\
Asian-Pac-Islander & 0.269 & 0.267 & 0.261 \\
Black              & 0.121 & 0.121 & 0.123 \\
Other              & 0.123 & 0.122 & 0.124 \\
White              & 0.254 & 0.254 & 0.254 \\
\bottomrule
\end{tabular}}
\end{minipage}

\end{table*}

\begin{itemize}
    \item If $\hat{\alpha}(x_t) = \texttt{True}$, then $\varphi(D',\hat{D'})$ gains one satisfied soft clause ($x_t$), and $\varphi(D,\hat{D})$ loses one satisfied soft clause ($\neg x_t$ violated).
    \item If $\hat{\alpha}(x_t) = \texttt{False}$, then $\varphi(D',\hat{D'})$ loses one soft clause and $\varphi(D,\hat{D})$ gains one.
\end{itemize}

Thus, the number of satisfied soft clauses under $\hat{\alpha}$ differs by at most $2$ between the two formulas. As a result, the sets of tuples $D_{\hat{\alpha}}$ (tuples with $x_t = \texttt{True}$) used to compute the repair values may differ by at most one tuple, and the database $D$ differs from $D'$ by one tuple as well. Therefore, the symmetric difference
{\scriptsize
$$
|D \setminus D_{\hat{\alpha}}| + |D_{\hat{\alpha}} \setminus D|
\quad \text{vs} \quad
|D' \setminus D_{\hat{\alpha}}| + |D_{\hat{\alpha}} \setminus D'|
$$
}
can differ by at most $2$.

Hence
{\scriptsize
$$
|\Delta(\varphi(D,\hat{D}), \hat{\alpha}) - \Delta(\varphi(D',\hat{D'}), \hat{\alpha})| \leq 2
$$
}
\end{proof}

We are now ready to combine the lemmas into our result that bounds the sensitivity of $\repairsat$.

\begin{proposition}[Sensitivity of $\repairsat$]\label{prop:repairsat-sensitivity}
For a database $D$ and a fairness criterion $\fairnessdef$, the sensitivity of $\repairsat(\fairnessdef,D)$ is at most $2$.
\end{proposition}

\begin{proof}[Proof of \Cref{prop:repairsat-sensitivity}]
Let there be $D \neighbor D'$ neighboring databases such that they differ in one tuple by replacement. Without loss of generality, assume that $D'$ contains one more occurrence of $t''$ and one less occurrence of $t'$ than $D$.

Given a fairness criterion of the form $\fairnessdef = \protectedset \independent \responseset \mid \admissibleset$, let $\jdb$ and $\hat{D}$ be the join and cross-product databases, and similarly for $D'$.

Since $\hat{D}=\hat{D'}$, (the cross-product database depends only on domains), the set of hard clauses is identical in both formulas $\varphi(D,\hat{D})$ and $\varphi(D',\hat{D'})$. We can also notice that the only difference between these formulas lies in at most two soft clauses, as formulated in the proof of \Cref{lem:simplified-repair-cost-sensitivity}.

By \Cref{lem:regular-and-simplified-repair-cost-equivalence},
{\scriptsize
$$
\repairsat(\fairnessdef,D)
= \min_{\alpha \models \mathcal{H}(\hat D)} \Delta(\varphi(D,\hat D),\alpha),
\quad
\repairsat(\fairnessdef,D')
= \min_{\alpha \models \mathcal{H}(\hat D)} \Delta(\varphi(D',\hat D),\alpha)
$$
}
Let
{\scriptsize
$$\alpha_D \in \arg\min_{\alpha \models \mathcal{H}(\hat D)} \Delta(\varphi(D,\hat D),\alpha)$$}

and
{\scriptsize
$$\alpha_{D'} \in \arg\min_{\alpha \models \mathcal{H}(\hat D)} \Delta(\varphi(D',\hat D),\alpha)$$}
be optimal assignments.

By \Cref{lem:simplified-repair-cost-sensitivity}, for every feasible $\alpha$, it holds that
{\scriptsize
$$
\left|\Delta(\varphi(D,\hat D),\alpha)-\Delta(\varphi(D',\hat D),\alpha)\right| \leq 2
$$
}

Applying this bound with $\alpha=\alpha_D$ and using optimality of the assignments, we get
{\scriptsize
\begin{align*}
\repairsat(\fairnessdef,D) &= \Delta(\varphi(D,\hat D),\alpha_D) \\
&\leq \Delta(\varphi(D',\hat D),\alpha_D) + 2 \\
&\leq \Delta(\varphi(D',\hat D),\alpha_{D'}) + 2 \\
&= \repairsat(\fairnessdef,D') + 2    
\end{align*}
}

By symmetry (if we swap $D$ and $D'$), we also have
$\repairsat(\fairnessdef,D') \leq \repairsat(\fairnessdef,D) + 2$.
Therefore, we get
{\scriptsize
$$
\sens_{\repairsat(\fairnessdef,D)} = \max_{D' \neighbor D} \left|\repairsat(\fairnessdef,D)-\repairsat(\fairnessdef,D')\right| \leq 2
$$
}
\end{proof}

\begin{proof}[Proof of \Cref{prop:repairsat-properties} Part 4]
According to \Cref{prop:repairsat-sensitivity}, for a single fairness criterion it holds that the difference between the soft clauses in the CNFs for two neighboring databases is at most one, and also that the difference in the cost of repair is at most $2$. So for a set of fairness criteria $\constraints$ it holds that the difference between the soft clauses in the CNFs is at most $|\constraints|$.

By repeating the proof of \Cref{prop:repairsat-sensitivity}, we can obtain that the difference in the cost of repair between two neighboring databases for a set of fairness criteria $\constraints$ is at most $2|\constraints|$.
\end{proof}

\section{Complexity Analysis of the Algorithms}\label{sec:complexity}

\paratitle{Complexity of \Cref{alg:tvd-proxy}}
The complexity of \Cref{alg:tvd-proxy} is $\mathcal{O}\left(|\constraints|n\right)$. 
For every $\fairnessdef$, the algorithm first computes the joint and the marginal empirical probabilities. Assuming the counts are stored in hash tables, so every lookup and update takes $\mathcal{O}(1)$ time, this whole part takes $\mathcal{O}(n)$ time. 
Then, the algorithm computes the TVD and divides into two cases. In the unconditional case it goes over all the tuples in $D$ once when calculating the sum, so the whole computation takes $\mathcal{O}(n)$ time. And in the conditional case it goes over all the values $\admissiblevalue$ in $\admissibleset$ and calculates the sum from the unconditional case for it. This computation takes $\mathcal{O}(n + n) = \mathcal{O}(n)$ time. Updating the cumulating sum takes $\mathcal{O}(1)$ time.
Finally, after the algorithm finishes going over all $\fairnessdef$ in $\constraints$, it adds Laplace noise to the resulting cumulative sum and returns it, which takes $\mathcal{O}(1)$.

\paratitle{Complexity of \Cref{alg:repairsat}}
The complexity of \Cref{alg:repairsat} is $\mathcal{O}\left(|\constraints|\left(n^4 + SAT\right)\right)$, where $\mathcal{O}(SAT)$ is the complexity of the SAT solver, due to the following analysis.
For every $\fairnessdef$, the algorithm first computes the self-join $\jdb$ of $D$ on $\admissibleset$ and goes over it to add the soft clauses to $\varphi$. 
It takes $\mathcal{O}(n)$ time to calculate the per-$\admissiblevalue$ projection sets, so it takes $\mathcal{O}(n^2)$ overall to perform the join and iterate over it to compute the soft clauses. 
Then the algorithm computes the self-join $C$ of $\jdb$ on $\admissibleset$. Since each copy of $\jdb$ contains $\mathcal{O}\left(\left|\protectedset_\admissiblevalue\right| \times \left|\responseset_\admissiblevalue\right|\right)$ tuples for any $\admissiblevalue$, this takes $\mathcal{O}\left(\sum_{\admissiblevalue \in \admissibleset}\left|\protectedset_{\admissiblevalue}\right|^2 \left|\responseset_{\admissiblevalue}\right|^2\right) = \mathcal{O}(n^4)$ time.
Finally, the algorithm runs a solver to get an assignment for the constructed $\varphi$, computes the symmetric difference and updates the cumulative sum. This whole part takes $\mathcal{O}(SAT + n)$.
After the algorithm finishes iterating over all $\fairnessdef \in \constraints$, it adds Laplace noise to the resulting sum and returns it, which takes $\mathcal{O}(1)$ time.

\paratitle{Complexity of \Cref{alg:topk-contribution}}\label{contribution-complexity}
The complexity of \Cref{alg:topk-contribution} is $\mathcal{O}\left(|\constraints|n \log n\right)$. 
For every $\fairnessdef$, the algorithm first computes the joint and the marginal empirical probabilities, which takes $\mathcal{O}(n)$ time. 
Then, the algorithm computes the marginal differences, taking $\mathcal{O}(n)$ time. 
Sorting of the marginal differences for all the tuples and finding the $k$ largest ones takes $\mathcal{O}(n \log k)$ time if we are using a size-$k$ heap. Summing the $k$ largest marginal differences takes $\mathcal{O}(k)$ time and updating the cumulative sum takes $\mathcal{O}(1)$ time. Therefore, this whole part takes $\mathcal{O}(n \log k + k) = \mathcal{O}(n \log n)$. 
After going over all $\fairnessdef \in \constraints$, it adds Laplace noise to the cumulative sum and returns it, which takes $\mathcal{O}(1)$ time.

\begin{table}[t]
\centering
\caption{Exploratory queries on the \texttt{Adult} and \texttt{Compas} datasets.}
\label{tab:eda-queries}
\small
\setlength{\tabcolsep}{4pt}
\renewcommand{\arraystretch}{1.03}

\begin{tabular}{@{}p{0.13\linewidth}p{0.82\linewidth}@{}}
\toprule
\textbf{Dataset} & \textbf{Query} \\
\midrule

\multirow{24}{*}{\texttt{Compas}}
& $q_1$: SELECT \texttt{race}, PERCENTILE\_CONT(0.5) \\
& \hspace{2.0em} WITHIN GROUP (ORDER BY \texttt{med-d}) AS \texttt{med-d} \\
& \hspace{1.3em} FROM (SELECT \texttt{c\_charge\_degree}, \texttt{race}, \\
& \hspace{4.6em} PERCENTILE\_CONT(0.5) \\
& \hspace{2.0em} WITHIN GROUP (ORDER BY \texttt{decile\_score}) AS \texttt{med-d} \\
& \hspace{2.7em} FROM Compas GROUP BY \texttt{c\_charge\_degree}, \texttt{race}) t \\
& \hspace{1.3em} GROUP BY \texttt{race}; \\
\cmidrule(l){2-2}

& $q_2$: SELECT \texttt{c\_charge\_degree}, PERCENTILE\_CONT(0.5) \\
& \hspace{2.0em} WITHIN GROUP (ORDER BY \texttt{med-d}) AS \texttt{med-d} \\
& \hspace{1.3em} FROM (SELECT \texttt{race}, \texttt{c\_charge\_degree}, \\
& \hspace{4.6em} PERCENTILE\_CONT(0.5) \\
& \hspace{2.0em} WITHIN GROUP (ORDER BY \texttt{decile\_score}) AS \texttt{med-d} \\
& \hspace{2.7em} FROM Compas GROUP BY \texttt{race}, \texttt{c\_charge\_degree}) t \\
& \hspace{1.3em} GROUP BY \texttt{c\_charge\_degree}; \\
\cmidrule(l){2-2}

& $q_3$: SELECT \texttt{age\_cat}, AVG(\texttt{avg-d}) AS \texttt{avg-d} \\
& \hspace{1.3em} FROM (SELECT \texttt{race}, \texttt{age\_cat}, AVG(\texttt{decile\_score}) AS \texttt{avg-d} \\
& \hspace{2.7em} FROM Compas GROUP BY \texttt{race}, \texttt{age\_cat}) t \\
& \hspace{1.3em} GROUP BY \texttt{age\_cat}; \\
\cmidrule(l){2-2}

& $q_4$: SELECT \texttt{race}, AVG(\texttt{avg-d}) AS \texttt{avg-d} \\
& \hspace{1.3em} FROM (SELECT \texttt{age\_cat}, \texttt{race}, AVG(\texttt{decile\_score}) AS \texttt{avg-d} \\
& \hspace{2.7em} FROM Compas GROUP BY \texttt{age\_cat}, \texttt{race}) t \\
& \hspace{1.3em} GROUP BY \texttt{race}; \\
\midrule

\multirow{5}{*}{\texttt{Adult}}
& $q_5$: SELECT \texttt{sex}, AVG(\texttt{income>50K}) AS \texttt{avg-i} \\
& \hspace{1.3em} FROM Adult GROUP BY \texttt{sex}; \\
\cmidrule(l){2-2}

& $q_6$: SELECT \texttt{race}, AVG(\texttt{income>50K}) AS \texttt{avg-i} \\
& \hspace{1.3em} FROM Adult GROUP BY \texttt{race}; \\
\bottomrule
\end{tabular}
\end{table}

\subsection{Case Study: Query Results}

We ran the exploratory queries from \Cref{tab:eda-queries} both without privacy considerations and with privacy budgets $\varepsilon = 1$ and $\varepsilon = 10$. For the \texttt{Adult} dataset, the sensitivity of each query is given by 1 divided by the group size, since these queries compute the average of the binary attribute \texttt{income>50K}. Similarly, for the \texttt{Compas} dataset, the sensitivity is 9 divided by the group size for queries that compute averages (since the range of \texttt{decile\_score} is from 1 to 10), and 10 for queries computing medians. The results of the queries are shown in \Cref{tab:query-results}, and they reveal clear disparities across protected groups in both \texttt{Adult} and \texttt{Compas} datasets.

\subsection{Measure Drill-Down}\label{sec:drill-down}

\paratitle{Proxy faithfulness of $\mutualproxytvd$}
Following \Cref{fig:mi-proxies-comparison}, \Cref{fig:experiment-4} shows a comparison of $\mutual$ and $\mutualproxytvd$ values for every criterion from \Cref{tab:fairness-criteria}, with privacy budget $\varepsilon=1$. For every dataset, $\mutualproxytvd$ closely tracks $\mutual$, with an average Kendall's tau correlation~\cite{kendall1938new} being $0.92$. Specifically, for each individual dataset, an increase in $\mutual$ corresponds to an increase in $\mutualproxytvd$, and vice versa, as we would expect from a good proxy. However, this proportionality does not extend across datasets, since for the \texttt{Healthcare} and \texttt{Compas} datasets the same value of $\mutual$ ($0.09$) corresponds to different values of $\mutualproxytvd$. 
Within each dataset, we can further glean the level of unfairness for a criterion by comparing the values of the same measure for it and for another criterion. For example, for \texttt{IPUMS-CPS}, both measures assign higher values to criterion $4$ ($\texttt{HEALTH} \independent \texttt{INCTOT} \mid \texttt{AGE}$) than to criterion $3$ ($\texttt{HEALTH} \independent \texttt{MARST} \mid \texttt{AGE}$), meaning that the predicted unfairness for criterion $4$ is larger. Intuitively, given a person’s age, their income depends more on their health than on marital status.

\begin{figure*}[ht]
  \centering
  \begin{subfigure}[b]{0.8\textwidth}
    \centering
    \includegraphics[width=\linewidth]{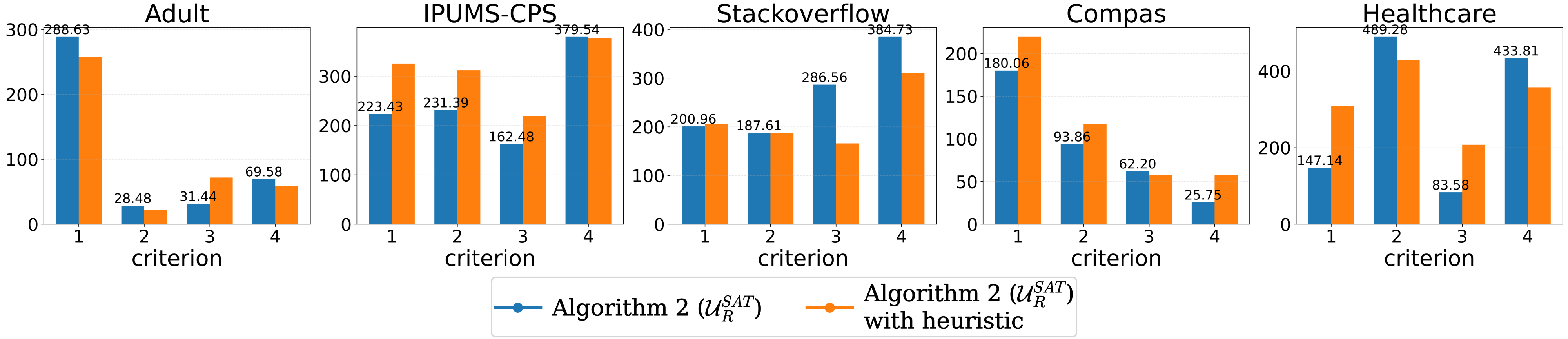}
    \caption{Values of $\repairsat$ with and without the heuristic.}
    \label{fig:experiment-7-1}
  \end{subfigure}
  \hfill
  \begin{subfigure}[b]{0.8\textwidth}
    \centering
    \includegraphics[width=\linewidth]{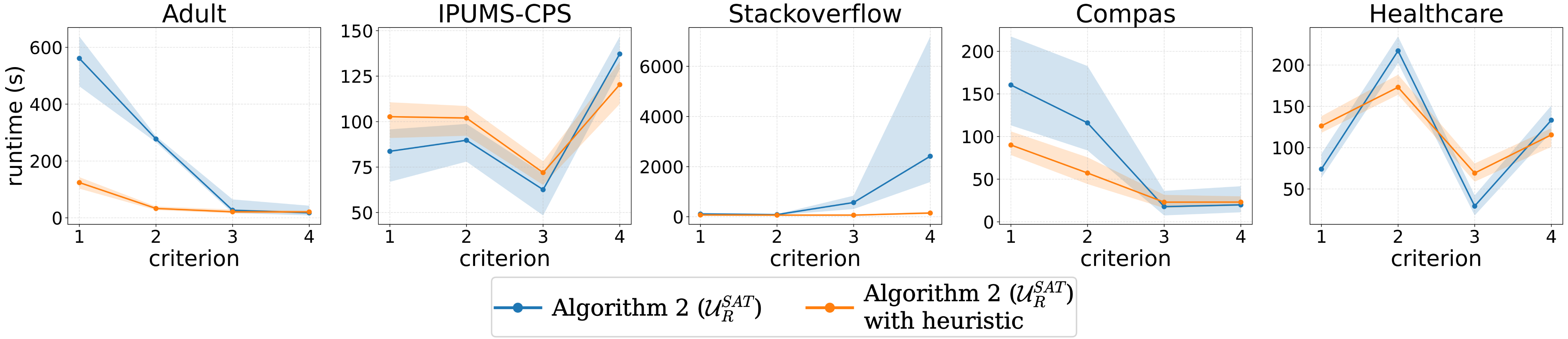}
    \caption{Runtimes of $\repairsat$ with and without the heuristic.}
    \label{fig:experiment-7-2}
  \end{subfigure}
  \caption{Comparison of $\repairsat$ (computed by \Cref{alg:repairsat}) with and without the heuristic.}
  \label{fig:repairsat-repairsatchunked-comparison}
\end{figure*}

\paratitle{Effect of $k$ on $\contribution$}
\Cref{fig:experiments-5-6} depicts the effect of the parameter $k$ on the value and relative $L1$ error of \Cref{alg:topk-contribution} (computing $\contribution$).
In \Cref{fig:experiment-5}, we plot the value of \Cref{alg:topk-contribution} as $k$ increases, with infinite privacy budget, $\varepsilon=\infty$. 
We do so since this experiment examines the effect of $k$ on the measure values, which would be highly distorted by noise, due to the dependency of the noise scale on $k$ (see \Cref{prop:contribution-properties}). 
Overall, as expected, the value is monotone non-decreasing in $k$ because, as $k$ grows, we add an increasing number of non-negative $\margdiff$ values to the resulting sum.
Almost all datasets reach a `plateau' starting from some value of $k$, meaning that a subset of tuples carries most of the $\margdiff$ of the dataset, while the remaining tuples contribute only marginally as $k$ increases. For example, in \texttt{Healthcare}, most tuples have large $\margdiff$ values, and thus $\contribution$ continues to grow almost linearly in $k$. In contrast, for \texttt{IPUMS-CPS}, many tuples have small $\margdiff$ values, so 
the trend does not plateau and remains noisy.

\Cref{fig:experiment-6} depicts the relative $L1$ error as a function of $k$, due to noise.
The relative error is quite large due the sensitivity of $\contribution$ being directly proportional to $k$ and $|\constraints|$, and thus also the magnitude of the added DP noise (\Cref{prop:contribution-properties}). Furthermore, the error is larger and exhibits greater variability for datasets with many distinct values, such as \texttt{IPUMS-CPS}. 
This suggests that a larger budget may be required for $\contribution$ and that it may be used as a secondary measure 
for identifying the cumulative effect of outlier tuples. 
We recognize that selecting an optimal value of $k$ for $\contribution$ is a challenging problem and leave this for future work.

\begin{figure}[h]
  \begin{subfigure}[b]{\linewidth}
    \includegraphics[width=\linewidth]{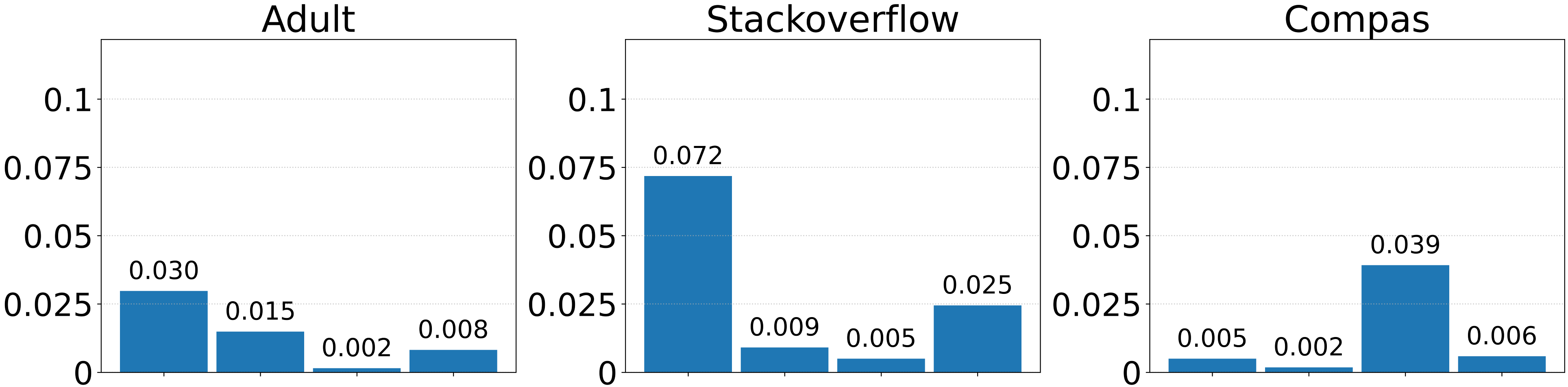}
    \caption{Values of $\mutualproxytvd$.}
    \label{fig:tvd-contribution-comparison-1}
  \end{subfigure}
  \begin{subfigure}[b]{\linewidth}
    \includegraphics[width=\linewidth]{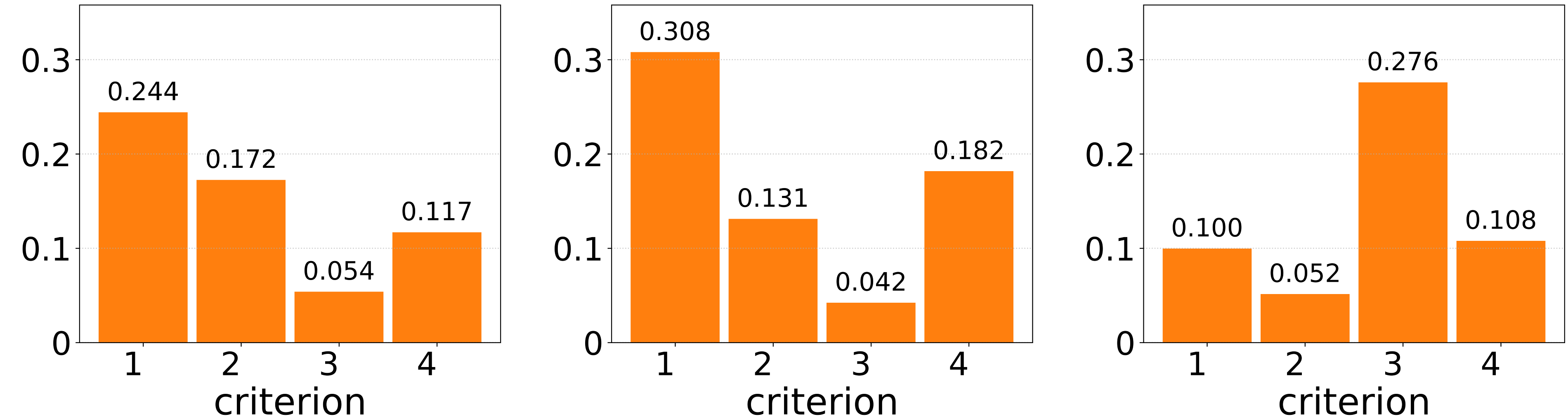}
    \caption{Values of $\contribution$.}
    \label{fig:tvd-contribution-comparison-2}
  \end{subfigure}
  \caption{Comparison of $\mutualproxytvd$ (computed by \Cref{alg:tvd-proxy}) and $\contribution$ (computed by \Cref{alg:topk-contribution}) without privacy considerations.}
    \label{fig:tvd-contribution-comparison}
\end{figure}

\paratitle{Comparing $\mutualproxytvd$ and $\contribution$}\label{sec:tvd-contribution-comparison}
In \Cref{fig:tvd-contribution-comparison}, we compared $\mutualproxytvd$ and $\contribution$ without privacy considerations. While the values of $\mutualproxytvd$ are smaller than the values of $\contribution$ for all criteria, it can be noticed that $\mutualproxytvd$ and $\contribution$ exhibited exactly the same trends in values. That is, higher values of $\mutualproxytvd$ correspond to higher values of $\contribution$, and vice versa. Moreover, relative differences in the values of $\mutualproxytvd$ are reflected by comparable relative differences in the values of $\contribution$.

\begin{figure*}
    \centering
    \includegraphics[width=0.8\textwidth]{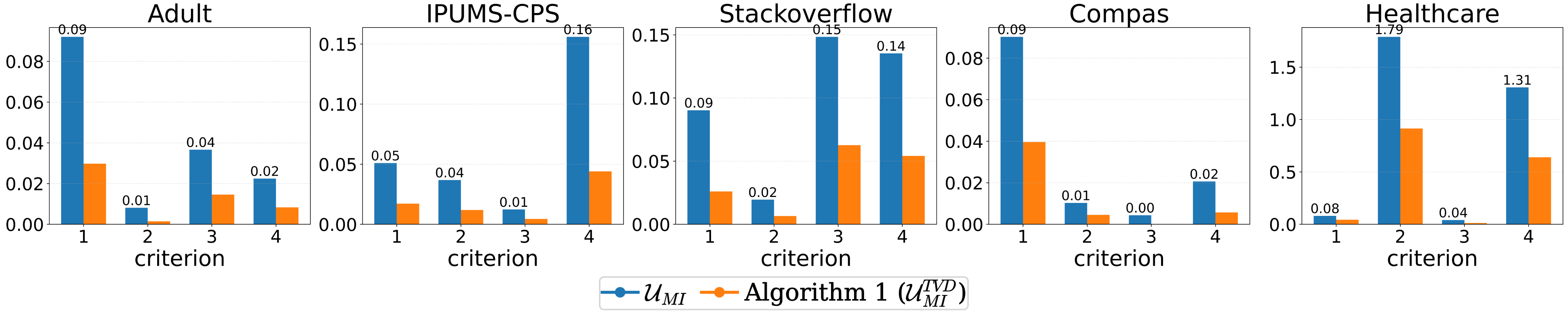}
    \caption{Faithfulness of noisy $\mutualproxytvd$ (computed by \Cref{alg:tvd-proxy}) to noisy $\mutual$ over different datasets and fairness criteria from~\Cref{tab:fairness-criteria}.}
    \label{fig:experiment-4}
    \vspace{-2mm}
\end{figure*}

\begin{figure*}
  \centering
  \begin{subfigure}[b]{0.8\textwidth}
    \centering
    \includegraphics[width=\linewidth]{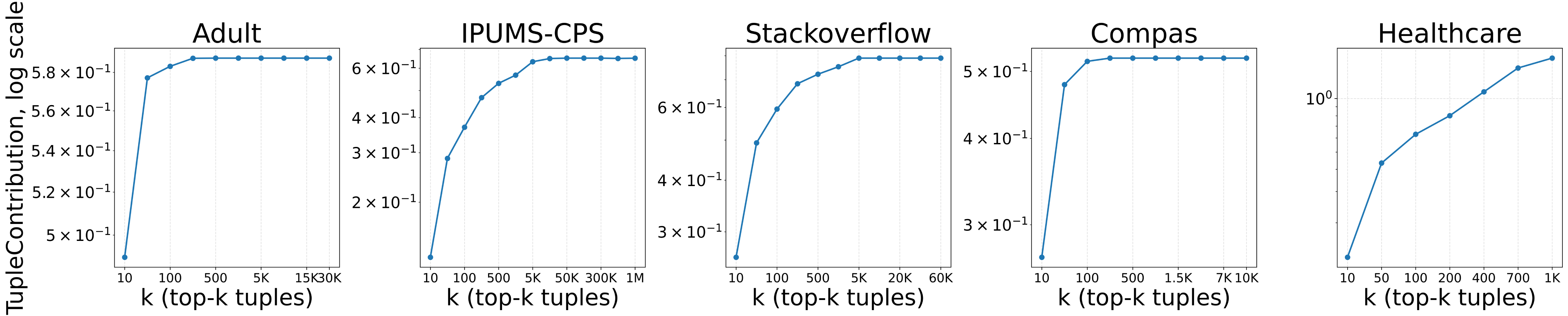}
    \caption{Non-noisy value of $\contribution$ as function of $k$.}
    \label{fig:experiment-5}
  \end{subfigure}
  \hfill
  \begin{subfigure}[b]{0.8\textwidth}
    \centering
    \includegraphics[width=\linewidth]{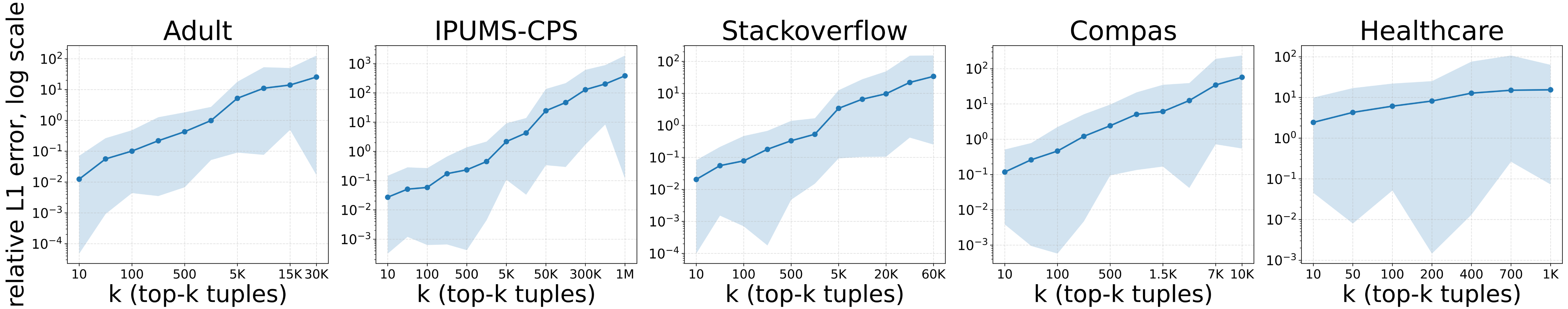}
    \caption{Relative $L1$ error of $\contribution$ as function $k$.}
    \label{fig:experiment-6}
  \end{subfigure}
  \caption{Effect of $k$ on $\contribution$ (computed by \Cref{alg:topk-contribution}) in terms of true value and relative $L1$ error for each dataset with its four criteria.}
  \label{fig:experiments-5-6}
\end{figure*}

\paratitle{Chunk heuristic effect on $\repairsat$}
In \Cref{fig:repairsat-repairsatchunked-comparison}, we compared \Cref{alg:repairsat} with and without the heuristic (i.e., computing $\repairsat$ with chunks of size $100$, and without them), with a budget of $\varepsilon=1$. \Cref{alg:repairsat} with the heuristic is noticeably faster than without it, reducing the total runtime by $56.26\%$ across all datasets. The values for the heuristic almost exactly mirror the trend of the values for the true $\repairsat$.

For the rest of this section we will denote $\repairsat$ with the chunking heuristic as $\repairsatchunked$. We will prove that the sensitivity bound of $\repairsatchunked$ is the same as the sensitivity bound of $\repairsat$.

\begin{proposition}[Sensitivity of $\repairsatchunked$]\label{prop:repairsatchunked-sensitivity}
For a database $D$ and a fairness criterion $\fairnessdef$, the sensitivity of $\repairsatchunked(\fairnessdef,D)$ is at most $2$.
\end{proposition}

\begin{proof}
Let there be $D \neighbor D'$ neighboring databases such that they differ in one tuple by replacement. Let $\mathsf{Chunk}(\cdot)$ be the chunking procedure used by $\repairsatchunked$, which consists of sorting and then partitioning into consecutive blocks of fixed size, such that we assume that all duplicates of any tuple are placed into the same block, the blocks are disjoint and their union equals the full database. We write
{\scriptsize
$$
\mathsf{Chunk}(D)=\{D^{(1)},\dots,D^{(m)}\}, \qquad 
\mathsf{Chunk}(D')=\{{D'}^{(1)},\dots,{D'}^{(m)}\}
$$
}

By the definition of $\repairsatchunked$, we assume that any neighbor $D'$ of $D$ has the same tuple ordering as $D$. Therefore, the replacement of a single tuple can affect the contents of at most one chunk, which is the corresponding chunk in $D'$. Hence, there exists an index $j\in\{1,\dots,m\}$ such that
{\scriptsize
$$
D^{(i)} = {D'}^{(i)} \ \text{for all } i\neq j,
\qquad
D^{(j)} \neighbor {D'}^{(j)}
$$
}

By the definition of $\repairsatchunked$, it holds that
{\scriptsize
$$
\repairsatchunked(\fairnessdef,D) = \sum_{i=1}^{m} \repairsat(\fairnessdef,D^{(i)}),
\qquad
\repairsatchunked(\fairnessdef,D') = \sum_{i=1}^{m} \repairsat(\fairnessdef,{D'}^{(i)})
$$
}

Therefore,
{\scriptsize
\begin{align*}
\Big|\repairsatchunked(\fairnessdef,D)-\repairsatchunked(\fairnessdef,D')\Big|
&= \left|\sum_{i=1}^{m}\repairsat(\fairnessdef,D^{(i)})-\sum_{i=1}^{m}\repairsat(\fairnessdef,{D'}^{(i)})\right|\\
&= \left|\repairsat(\fairnessdef,D^{(j)})-\repairsat(\fairnessdef,{D'}^{(j)})\right|
\end{align*}
}
Since $D^{(j)} \neighbor {D'}^{(j)}$, by \Cref{prop:repairsat-sensitivity}) we get
{\scriptsize
$$
\left|\repairsat(\fairnessdef,D^{(j)})-\repairsat(\fairnessdef,{D'}^{(j)})\right| \le 2.
$$
}
Combining the above yields
{\scriptsize
$$
\left|\repairsatchunked(\fairnessdef,D)-\repairsatchunked(\fairnessdef,D')\right| \le 2
$$
}

Therefore, we get
{\scriptsize
$$
\sens_{\repairsatchunked(\fairnessdef,D)} = \max_{D' \neighbor D} \left|\repairsatchunked(\fairnessdef,D)-\repairsatchunked(\fairnessdef,D')\right| \leq 2
$$
}
\end{proof}

\fi

\end{document}